\documentclass[aps,prd,preprint,superscriptaddress,nofootinbib]{revtex4-1}
\usepackage{braket}
\usepackage{xargs}
\usepackage{mathtools}
\usepackage[english]{babel}
\usepackage[utf8]{inputenc}
\usepackage{amsmath,amssymb,braket,slashed,mathrsfs,amsthm}
\usepackage{pifont}
\usepackage{graphicx}
\usepackage{color,hyperref}
\usepackage{multirow,array}
\usepackage{bm}

\newtheorem{thm}{Theorem}
\newcommand{\ba}{\begin{eqnarray}}
\newcommand{\ea}{\end{eqnarray}}

\makeatletter
\newcommand{\sumint}{\mathop{\mathpalette\sum@int\relax}\slimits@}
\newcommand{\sum@int}[2]{%
	\ooalign{$\m@th#1\sum$\cr\hidewidth$\m@th#1\int$\hidewidth\cr}%
}
\newcommand{\sumintinline}{\mathop{\mathpalette\sum@intinline\relax}\slimits@}
\newcommand{\sum@intinline}[2]{%
	\ooalign{$\m@th#1\scalebox{0.8}{$\sum$}$\cr\hidewidth$\m@th#1\scalebox{1.1}{$\int$}$\hidewidth\cr}%
}
\makeatother

\newcommand{\innovation}{Collaborative Innovation Center of Quantum Matter, Beijing 100871, China}
\newcommand{\chep}{Center for High Energy Physics, Peking University, Beijing 100871, China}
\newcommand{\pkuphy}{School of Physics, Peking University, Beijing 100871,
China}
\newcommand{\BNL}{Physics Department, Brookhaven National Laboratory, Upton, NY 11973, USA}

\begin{document}
\title{Finite-volume formalism for physical processes with an electroweak loop integral}

\author{Xin-Yu Tuo}\email{xtuo@bnl.gov}\affiliation{\pkuphy}\affiliation{\BNL}
\author{Xu~Feng}\email{xu.feng@pku.edu.cn}\affiliation{\pkuphy}\affiliation{\innovation}\affiliation{\chep}

\date{\today}

\begin{abstract}
	This study investigates finite-volume effects in physical processes that involve the combination of long-range hadronic matrix elements with electroweak loop integrals. We adopt the approach of implementing the electroweak part as the infinite-volume version, which is denoted as the EW$_\infty$ method in this work. A general approach is established for correcting finite-volume effects in cases where the hadronic intermediate states are dominated by either a single particle or two particles. For the single-particle case, this work derives the infinite volume reconstruction (IVR) method from a new perspective. For the two-particle case, we provide the correction formulas for power-law finite-volume effects and unphysical terms with exponentially divergent time dependence. The finite-volume formalism developed in this study has broad applications, including the QED corrections in various processes and the two-photon exchange contribution in $K_L\to\mu^+\mu^-$ or $\eta\to\mu^+\mu^-$ decays.
\end{abstract}

\maketitle

\section{Introduction\label{sec:Intro}}
Lattice QCD provides a nonperturbative method for solving the quantum chromodynamics (QCD) theories by discretizing the quark and gluon fields in Euclidean space within a finite size of $T\times L^3$. Studying the finite-volume effects resulting from this finite size is crucial for the development of lattice methodology. With the advancement of computational resources and lattice techniques, lattice QCD has achieved high precision for many physical observables, requiring finite-volume effects to be taken into account. The finite-volume formalism also enables us to extract important physical information from finite-volume systems and extends lattice QCD's computational capacity to physical observables that were previously inaccessible.

Among the developments of finite-volume formalisms, two-particle systems have received considerable attention and have become one of the major research focuses. This direction was initiated by Luscher's quantization condition, which relates the two-particle energy spectrum in finite volume to the physical scattering amplitude~\cite{Luscher:1986pf,Luscher:1990ux,Luscher:1991cf}, enabling lattice QCD to investigate scattering processes and resonances~\cite{Briceno:2017max}.
Based on this work, Ref.~\cite{Lellouch:2000pv} derived the Lellouch-Lüscher formula that connects physical matrix elements of $K\to\pi\pi$ to finite-volume matrix elements on lattice. The two-particle finite-volume formalism was later generalized and applied to a wide range of physical processes, including $0\stackrel{J}{\rightarrow}2$ processes~\cite{Hansen:2012tf,Briceno:2012yi} (e.g. the timelike form factor~\cite{Feng:2014gba,Andersen:2018mau,Erben:2019nmx,Meyer:2011um}), $1\stackrel{J}{\rightarrow}2$ processes~\cite{Briceno:2014uqa,Briceno:2015csa,Briceno:2021xlc} (e.g. $K\to\pi\pi$ decay~\cite{Blum:2015ywa,RBC:2015gro,RBC:2020kdj}, $\pi\pi\to\pi\gamma^*$ process~\cite{Briceno:2015dca,Briceno:2016kkp,Alexandrou:2018jbt}), and $2\stackrel{J}{\rightarrow}2$ processes~\cite{Briceno:2012yi,Briceno:2015tza,Baroni:2018iau,Briceno:2019nns,Briceno:2020xxs}. Two-particle finite-volume effects have also been studied in long-range electroweak amplitudes~\cite{Christ:2015pwa,Christ:2016eae,Briceno:2019opb}.

Despite these developments of two-particle finite-volume formalism, significant challenges remain, particularly in processes where the electroweak propagators and long-range hadronic matrix elements form a loop integral structure. We refer to the former as electroweak weight functions, which have analytically known forms, while the latter can be obtained through lattice calculations. In lattice calculations, these processes have large finite-volume effects from both parts. The electroweak weight functions may be long-range and contain infrared singularities from massless particles, preventing them from being directly implemented in finite volume, thereby resulting in large finite-volume effects. For the hadronic matrix elements, the intermediate states might involve light particles that cause large finite-volume effects. For instance, if the hadronic part contains two-particle intermediate states that are lighter than the initial state, it will have power-law finite-volume effects and unphysical terms with exponentially divergent time dependence. Addressing the finite-volume effects from both the hadronic and electroweak parts in these processes is therefore very challenging.

In Fig.\ref{fig:examples}, we list some typical processes involving such loop integral structure:
\begin{figure}[htbp]
	\centering
	\includegraphics[width=1\textwidth]{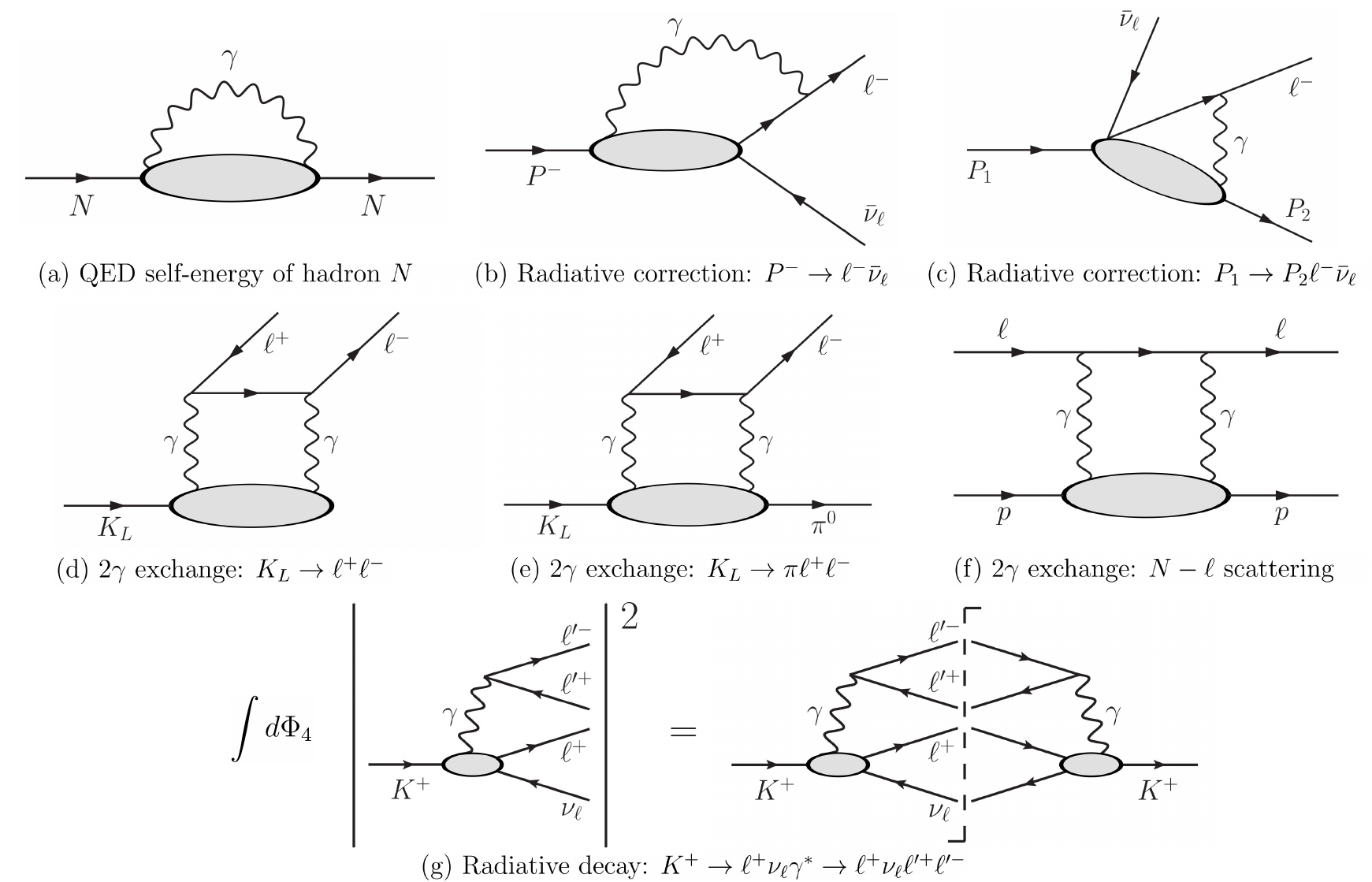}
	\caption{Examples of the physical processes where the electroweak propagators and long-range hadronic matrix elements form a loop integral structure. \label{fig:examples}}
\end{figure}
\begin{itemize}
	\item QED corrections in various processes, such as QED self-energy of hadron (Fig.~\ref{fig:examples}(a)), radiative correction of leptonic or semileptonic decay (Fig.~\ref{fig:examples}(b), (c)), etc. 
	\item Two-photon exchange contributions in decays of neutral mesons such as $\pi^0, \eta, K_L, \cdots$. For example, the long-distance contributions of rare decay $K_L\to\ell^+\ell^-$ and $K_L\to\pi\ell^+\ell^-$ are through two-photon exchange (Fig.~\ref{fig:examples}(d), (e))~\cite{Buras:1997fb,Cirigliano:2011ny}. 
	\item Two-photon exchange contributions in nucleon-lepton interactions (Fig.~\ref{fig:examples}(f)). These processes appears as higher-order corrections of muonic-hydrogen Lamb shift~\cite{Antognini:2013txn,Antognini:2013rsa}, electron-proton scattering~\cite{Arrington:2011dn,Afanasev:2017gsk}, etc. 
	\item Decay width of radiative decays, such as the radiative decay of neutral mesons $P^0\to\gamma^{(*)}\gamma^{(*)}$ or charged mesons $P^+\to\ell^+\nu_\ell\gamma^{(*)}$, where the photon may be a real photon or decay to a lepton pair (e.g. $K^+\to\ell^+\nu_\ell\gamma^*\to\ell^+\nu_\ell\ell'^+\ell'^-$ in Fig.\ref{fig:examples}(g)). The decay width involves phase space integration over the hadronic and electroweak parts of the square of the amplitude, thus forming a generalized loop integral structure.
\end{itemize}

The electroweak weight functions of above processes all have infrared singularities from massless photon propagators. To calculate such processes, we need to consider how to implement these electroweak weight functions in lattice calculations. A commonly used approach is the QED$_L$ method, which subtracts the zero mode of photon propagators. This method has been widely applied to calculate QED self-energy of hadrons~\cite{Hayakawa:2008an,Davoudi:2014qua,Davoudi:2018qpl}, radiative correction of leptonic decay~\cite{Carrasco:2015xwa,Lubicz:2016xro,Frezzotti:2020bfa,DiCarlo:2021apt}, neutrinoless double beta decay~\cite{Davoudi:2020gxs}, etc. The power-law finite-volume effects in QED$_L$ method can be derived and corrected analytically.

Recent studies have proposed implementing the electroweak weight functions as their infinite-volume version and calculating the loop integral in coordinate space. In this paper, we refer to this approach as the EW$_\infty$ method. For example, the hadronic light-by-light contribution of the muon $g-2$ involves multiple photons, and an implementation method of them was proposed using the infinite-volume and continuous QED~\cite{Blum:2017cer}.  Using the same form of photon propagator, Ref.~\cite{Feng:2018qpx} proposed the infinite volume reconstruction (IVR) method for the QED self-energy problem. This method reconstructs the hadronic matrix elements outside the lattice temporal range from lattice data using ground state dominance, and then combines it with the infinite-volume photon propagator, successfully avoiding the power-law finite-volume effects. This method has been applied to pion mass splitting~\cite{Feng:2021zek}, neutrinoless double beta decay~\cite{Tuo:2019bue,Tuo:2022hft}, radiative correction of leptonic decay~\cite{Christ:2020jlp,Christ:2023lcc}, radiative and rare decay of meson~\cite{Meng:2021ecs,Tuo:2021ewr,Christ:2020hwe}, two-photon exchange in muonic-hydrogen Lamb shift~\cite{Fu:2022fgh}, electroweak box diagrams in semileptonic decay of mesons~\cite{Feng:2020zdc,Ma:2021azh} and superallowed beta decay~\cite{Ma:2023kfr}.

The IVR method solved the problem of how to reduce the finite-volume effects when using the EW$_\infty$ method in the case of single-particle intermediate states. In addition to this case, more complex situations with two-particle intermediate states often occur in physical processes. For example, the EW$_\infty$ method has been proposed for the calculation of $\pi^0\to e^+ e^-$ and $K_L\to\mu^+\mu^-$ (Fig.~\ref{fig:examples}(d))~\cite{Christ:2020bzb,Christ:2022rho,Chao:2024vvl}. However, for the $K_L\to\mu^+\mu^-$ decay, currently no finite-volume formalism exists to account for the finite-volume effects caused by the $\pi\pi$ intermediate states. Similar two-particle intermediate states also occur in radiative correction of semileptonic decay (Fig.~\ref{fig:examples}(c)), $2\gamma$ exchange in $K_L\to\pi^0\ell^+\ell^-$ (Fig.~\ref{fig:examples}(e)), $2\gamma$ exchange in $N-l$ scattering when $\ell$ has large input momentum (Fig.~\ref{fig:examples}(f)), and $K^+\to\ell\nu_\ell\ell'^+\ell'^-$ decay width (Fig.~\ref{fig:examples}(g)), etc. Thus, developing a two-particle finite-volume formalism would extend the applicability of the EW$_\infty$ method to these processes.

In this paper, we study how to apply the EW$_\infty$ method to processes where the initial or final hadronic states are single-particle states or vacuum, and the intermediate hadronic states are single-particle or two-particle states. In the single-particle case, we derive the IVR method from a new perspective. We point out that in the EW$_\infty$ method, the infrared singularity of the photon propagator does not cause a singularity in the discrete momentum summand, explaining why the IVR method can avoid power-law finite-volume effects in this case. In the two-particle case, we give the correction formulas for power-law finite-volume effects and unphysical terms with exponentially divergent time dependence. For clarity, we take QED self-energy problem and $\eta\to\mu^+\mu^-$ decay as examples to illustrate the single-particle and two-particle cases, respectively. 

This paper is organized as follows. In Sect.~\ref{sec:EWinf}, we introduce the general form of the EW$_\infty$ method and explain the main idea of how to solve finite-volume effects. We consider two classes of processes: loop integral problems in Sect.~\ref{sec:loop} and decay width of radiative decay in Sect.~\ref{sec:decay}. Next, we derive the corretion formulas of the finite-volume effects in cases of single-particle and two-particle intermediate states in Sect.~\ref{sec:one} and Sect.~\ref{sec:two}, respectively.

\section{\label{sec:EWinf}EW$_\infty$ method}
\subsection{\label{sec:loop}Loop integral problem}
\begin{figure}[htbp]
	\centering
	\includegraphics[width=0.5\textwidth]{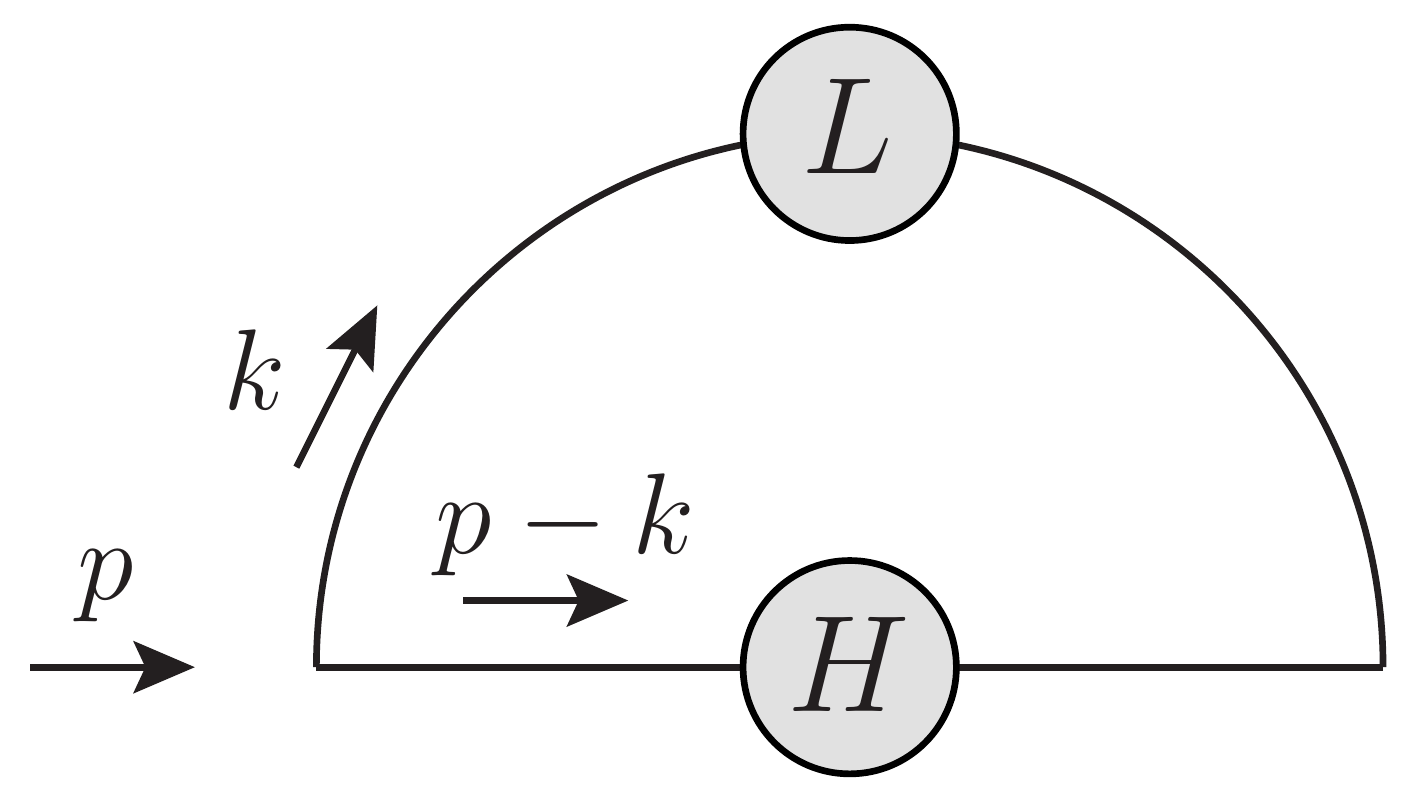}
	\caption{Target physical quantity: loop integral in Minkowski space.\label{fig:Cp}}
\end{figure}
The calculation structure of QED corrections or two-photon exchange can be abstracted as a general loop integral in Fig.~\ref{fig:Cp}, which is defined in infinite-volume Minkowski space as
\begin{equation}\label{Cinf}
	I^\infty =\int \frac{d^4 k}{(2\pi)^4}L^\infty(k)H^\infty(k,p),
\end{equation}
where we use the superscript $\infty$ to indicate that these quantities are defined in infinite volume. In this paper, the quantities are defined by default in Minkowski space. For Euclidean space quantities in the following part, we add a subscript $E$ to distinguish them. In this integral, the total input momentum is defined as $p=(m,\bm{0})$. The discussion here can also be generalized to $\bm{p}\neq \bm{0}$. In this work, we set $\bm{p}=0$ for simplicity. $L^\infty(k)$ is the electroweak weight function with momentum $k=(k^0,\bm{k})$, which may include propagators of the electroweak part, spinors of final state leptons and other coefficients. The hadronic matrix element $H^\infty(k,p)$ is
\begin{equation}
	\label{HMdef}
	H^\infty(k,p)=\int d^3 x\int_{-\infty}^\infty dt e^{ik\cdot x}\langle f|T[J_{1}(t,\bm{x})J_{2}(0)]|i\rangle,
\end{equation}
where the Minkowski space operators are defined as $J_{1}$ and $J_{2}$, and for generality we do not specify their explicit forms for now. In this paper, we consider the case where the initial and final hadronic states $|i/f\rangle$ are vacuum or single-particle states. We insert a complete set of intermediate states to expand $H^\infty(k,p)$ as
\begin{equation}
	\label{HMint0}
	\begin{aligned}
		&H^\infty(k,p)=H^\infty_{t<0}(k,p)+ H^\infty_{t>0}(k,p)\\
		=&i\sumint_{\alpha} \frac{\langle f | J_{2}|\alpha(-\bm{k},E_\alpha)\rangle\langle
			\alpha(-\bm{k},E_\alpha)|J_{1}|i\rangle}{m-k^0-E_\alpha+i\epsilon}+i\sumint_{\alpha'}\frac{\langle f | J_{1}|\alpha'(\bm{k},E_{\alpha'})\rangle\langle
			\alpha'(\bm{k},E_{\alpha'})|J_{2}|i\rangle}{m+k^0-E_{\alpha'}+i\epsilon},
	\end{aligned}
\end{equation}
where $\alpha$ and $\alpha'$ denote infinite-volume intermediate states for the time-ordering $t<0$ and $t>0$, respectively. The notation $\sumintinline_{\alpha/\alpha'}$ denotes the summation over all such infinite-volume states. In this work, we focus on the $t<0$ contribution for simplicity, and the method here can be applied to $t>0$ cases as well. We neglect the $t>0$ contribution and write
\begin{equation}
	\label{HMint}
\begin{aligned}
	H^\infty(k,p)&=H^\infty_{t<0}(k,p)=i\sumint_{\alpha}\frac{A_{\alpha}^\infty(-\bm{k},E_\alpha)}{m-k^0-E_\alpha+i\epsilon},\\
	I^\infty&=I^\infty_{t<0}=\int \frac{d^4 k}{(2\pi)^4}L^\infty(k)H_{t<0}^\infty(k,p).
\end{aligned}	
\end{equation}
Here, we define $A_{\alpha}^\infty(-\bm{k},E_\alpha)=\langle f | J_{2}|\alpha(-\bm{k},E_\alpha)\rangle\langle
\alpha(-\bm{k},E_\alpha)|J_{1}|i\rangle$. We consider the cases where the ground states of $|\alpha(-\bm{k},E_\alpha)\rangle$ are single-particle or two-particle states.

Next, we consider how to calculate the above quantity in the finite-volume Euclidean space. We define the loop integral in the ``EW$_\infty$ method'' as
\begin{equation}\label{CLTx}
	\begin{aligned}
		I^{(LT)}=c_{ME}\int_V d^3x\int_{-t_s}^0 d\tau L_E^\infty(\tau,\bm{x})H_E^{(L)}(\tau,\bm{x}),
	\end{aligned}
\end{equation}
where $t_s\leq T/2$ is the cutoff in the temporal integration, and $V=L^3$ means the spatial integration is in finite volume. In Appendix \ref{sec:App-EM}, we introduce the definitions of the Euclidean space coordinates $(\tau,\bm{x})$, momentum $k_E$ and operators $J_{1/2,E}$ used in this paper, and use $c_{ME}$ to account for the difference between the Euclidean space and Minkowski space conventions. $I^{(LT)}$ combines the infinite-volume electroweak weight function $L_E^\infty(\tau,\bm{x})$ and the finite-volume Euclidean space hadronic matrix element $H_E^{(L)}(\tau,\bm{x})$ defined below. In the case of $\tau<0$, $H_E^{(L)}(\tau,\bm{x})$ can be expanded by finite-volume intermediate states as
\begin{equation}\label{HELx}
	\begin{aligned}
		H_E^{(L)}(\tau<0,\bm{x})&=(\langle f|J_{2,E}(0)J_{1,E}(\tau,\bm{x})|i\rangle)^{(L)}\\
		&=\frac{1}{L^3}\sum_{\bm{k}'\in\Gamma} \sum_{\alpha_L}\langle f|J_{2,E}|\alpha_L(-\bm{k}',E_{\alpha_L})\rangle\langle \alpha_L(-\bm{k}',E_{\alpha_L})|J_{1,E}|i\rangle e^{-(m-E_{\alpha_L})\tau+i\bm{k}'\cdot\bm{x}}\\
		&=\frac{1}{L^3}\sum_{\bm{k}'\in\Gamma} \sum_{\alpha_L}A^{(L)}_{\alpha,E}(-\bm{k}',E_{\alpha_L}) e^{-(m-E_{\alpha_L})\tau+i\bm{k}'\cdot\bm{x}}.
	\end{aligned}
\end{equation}
Here, $\bm{k}'$ is the discrete momentum in the finite volume, which takes values from the set $\Gamma=\left\{\bm{k}'|\bm{k}'=\frac{2\pi}{L}\bm{n}\right\}$. In this paper, we use $\bm{k}$ as continuous momentum in infinite volume, and use $\bm{k}'$ as discrete momentum in finite volume. Since $|f\rangle$ and $|i\rangle$ are defined as vacuum or single-particle states, we ignore the difference between their infinite-volume version and finite-volume version, which is typically an exponentially suppressed finite-volume effect. $\sum_{\alpha_L}$ runs over discrete energy levels of finite-volume states $|\alpha_L(-\bm{k}',E_{\alpha_L})\rangle$ with momentum $-\bm{k}'$. We denote the finite-volume matrix element in Euclidean space as $A^{(L)}_{\alpha,E}(-\bm{k}',E_{\alpha_L})$. In this work, we ignore the around-the-world effects from the temporal boundary condition.

$L_E^\infty(\tau,\bm{x})$ is defined in infinite-volume Euclidean space as
\begin{equation}\label{LE}
	\begin{aligned}
		L_E^\infty(\tau,\bm{x})&=-\int \frac{d^3 k}{(2\pi)^3}\int_{C_{\text{lat}}} \frac{dk_E^0}{2\pi}e^{-ik_E^0 \tau-i\bm{k}\cdot\bm{x}}L^\infty_E(k_E),\\
		L^\infty_E(k_E)&=L^\infty(k)|_{k\to(-ik_E^0,\bm{k})}.
	\end{aligned}
\end{equation}
Here, the negative sign in the formula ensures that $I^{(LT)}$ matches the physical quantity $I^\infty$ in (\ref{Cinf}) after corrections discussed below. The denominator of $L^\infty(k)$ consists of lepton or photon propagators
\begin{equation}
	L^\infty(k)\sim\frac{1}{D_1D_2\cdots D_n}.
\end{equation}
\begin{figure}[htbp]
	\centering
	\includegraphics[width=0.9\textwidth]{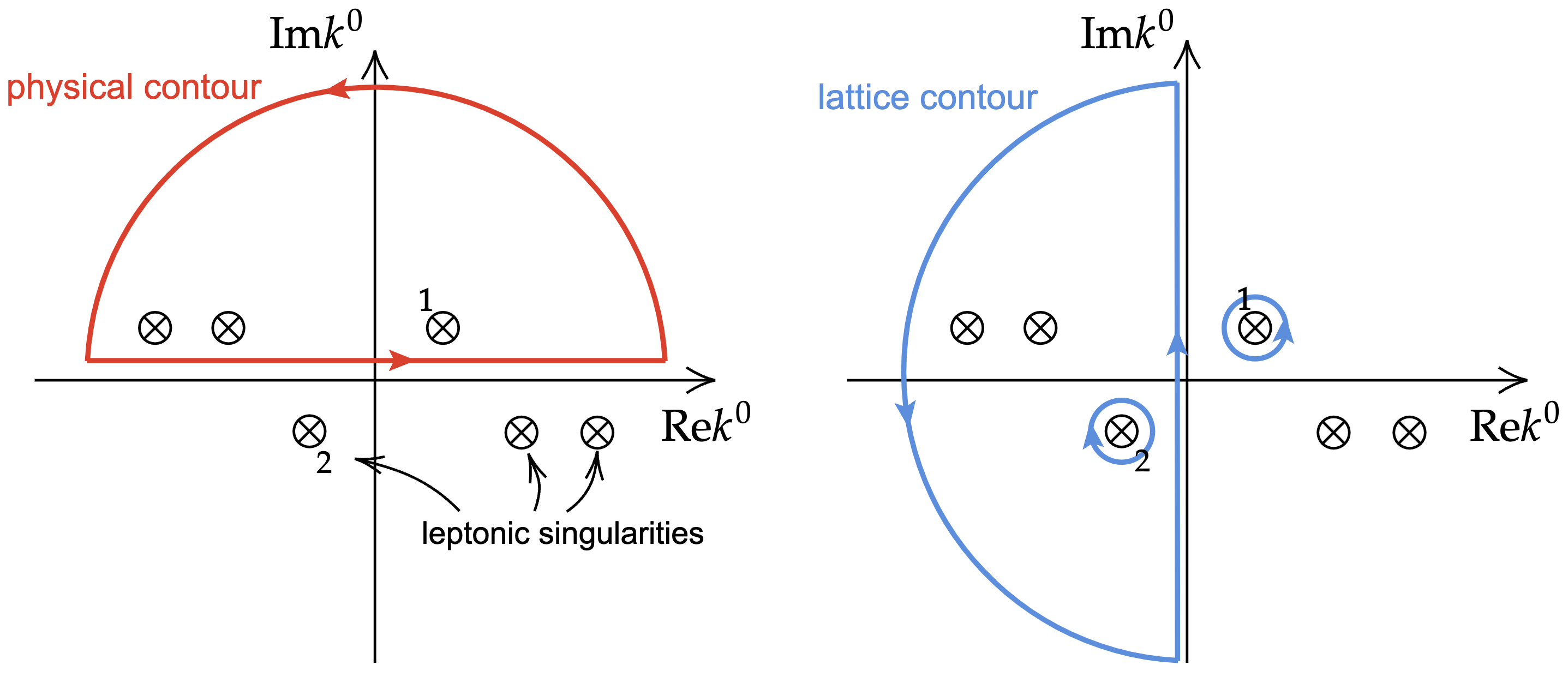}
	\caption{We define the physical contour $C$ and the lattice contour  $C_{\text{lat}}$ in the complex plane of $k^0$. $C$ is the correct Minkowski space contour, and $C_{\text{lat}}$ is the actual contour used in the electroweak weight function $L_E^\infty(\tau,\bm{x})$ in the lattice calculation.\label{fig:Wick}}
\end{figure}
Fig.~\ref{fig:Wick} shows the singularities of these electroweak propagators on the complex $k^0$ plane. The physical contour $C$ in Minkowski space is defined in the left figure. Due to the singularities like 1 and 2, the Wick rotation can not be directly done. We follow Ref.~\cite{Christ:2022rho} to define $L^\infty_E(\tau,\bm{x})$ using the modified contour $C_{\text{lat}}$ in the right figure, which contains the correct singularities of the electroweak weight function.

\subsection{\label{sec:sys}Theoretical analysis of finite-volume effects}
In this part, we study the difference between $I^{(LT)}$ and the physical quantity $I^\infty$:
\begin{equation}
	\Delta I=I^{(LT)}-I^\infty.
\end{equation}

When only the singularities of the electroweak weight function are considered, the lattice contour $C_{\text{lat}}$ is equivalent to the physical contour $C$, as shown in Fig.~\ref{fig:Wick}.
However, we also need to consider whether the hadronic singularities in the loop integral cause problems. For simplicity, we focus on the singularities $k^0=m-E_\alpha+i\epsilon$ that arise from the $t<0$ time-ordering contribution in Eq.~(\ref{HMint0}).
\begin{figure}[htbp]
	\centering
	\includegraphics[width=0.9\textwidth]{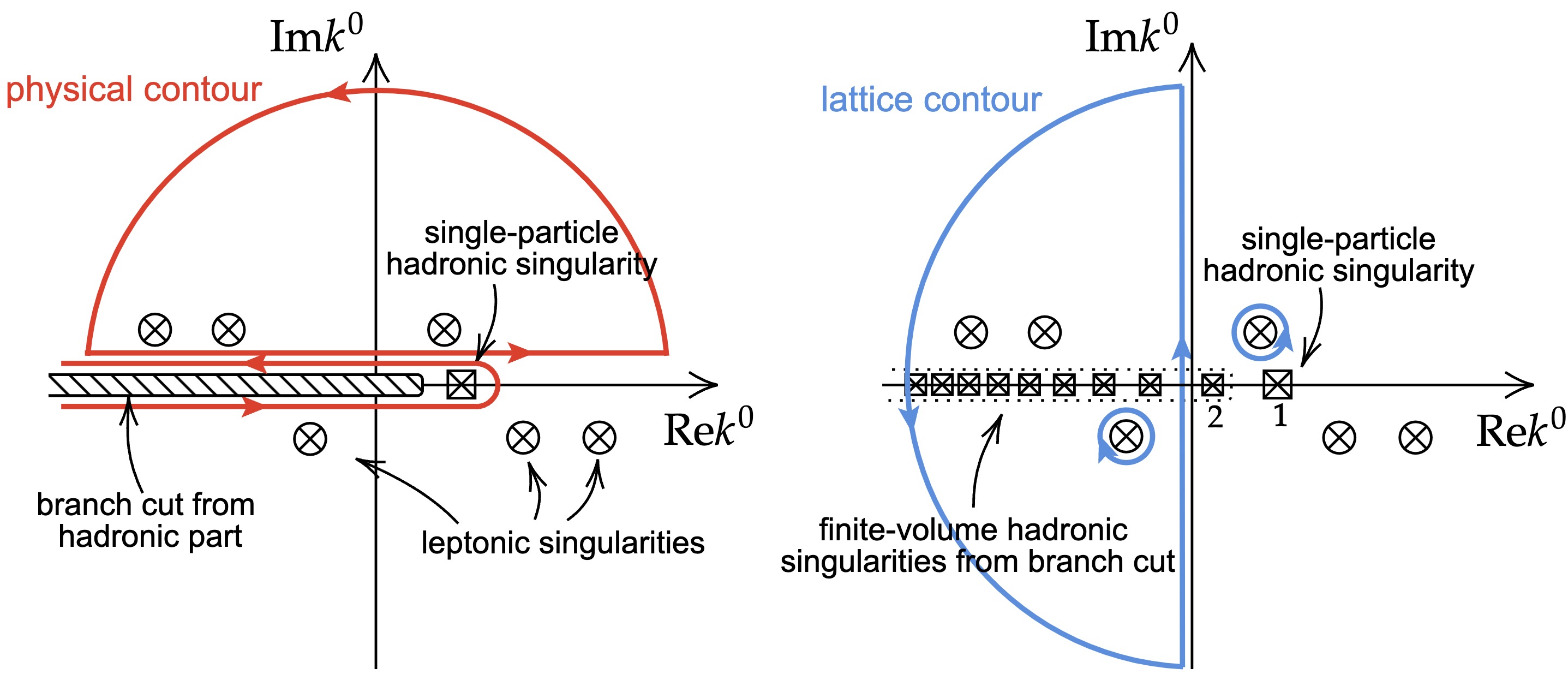}
	\caption{The definition of the physical contour $C$ and the lattice contour $C_{\text{lat}}$ in the presence of the hadronic singularities.\label{fig:Wick2}}
\end{figure}
Fig.~\ref{fig:Wick2} illustrates the physical contour $C$ and the lattice contour $C_{\text{lat}}$ in the presence of the hadronic singularities. The left panel shows the physical contour $C$ in infinite-volume Minkowski space. The hadronic singularities in this case may include isolated single-particle poles (the square pole in the figure), and a branch cut from multi-particle states (the shaded region in the figure). The right panel shows the lattice contour $C_{\text{lat}}$ in finite-volume Euclidean space. 
The finite-volume states $|\alpha_L\rangle$ replace the infinite-volume states $|\alpha\rangle$, and thus the branch cut from multi-particle states transforms into discrete finite-volume singularities. These finite-volume hadronic singularies fall into two categories:
\begin{enumerate}
	\item Heavy intermediate states with $E_\alpha>m$, which correspond to the poles on the left side of the $\mathrm{Im}k^0$ axis. They may be single-particle or multi-particle states. These states do not affect the equivalence between $C$ and $C_{\text{lat}}$. The finite-volume effects from these states are exponentially suppressed, and thus they are not the focus of this paper.
	\item Light intermediate states with $E_\alpha\leq m$, which correspond to the poles on the right side of the $\mathrm{Im}k^0$ axis. They may be single-particle states (the pole 1 in the figure), or multi-particle states (the pole 2 in the figure). In this paper, we mainly focus on solving systematic errors related to these states.
\end{enumerate}

For the convenience of theoretical analysis, we substitute Eq.~(\ref{HELx}) and Eq.~(\ref{LE}) into Eq.~(\ref{CLTx}) and write the expression of $I^{(LT)}$ in momentum space as
\begin{equation}
	\label{CLTp}
	\begin{aligned}
		I^{(LT)}&=\frac{1}{L^3}\sum_{\bm{k}'\in\Gamma}\int \frac{d^3 k}{(2\pi)^3}\delta_L(\bm{k}'-\bm{k})\int_{C_{\text{lat}}} \frac{(-id k_E^0)}{2\pi} L^\infty_E(k_E)c_{ME}H_{E}^{(LT)}(k'_E,p_E),
	\end{aligned}
\end{equation}
where we define the finite-volume $\delta$ function
\begin{equation}
	\label{deltaL}
	\delta_L(\bm{q})=\int_V d^3x e^{i\bm{q}\cdot\bm{x}}=L^3\frac{\sin(\frac L2 \bm{q}_x)}{\frac L2 \bm{q}_x}\frac{\sin(\frac L2 \bm{q}_y)}{\frac L2 \bm{q}_y}\frac{\sin(\frac L2 \bm{q}_z)}{\frac L2 \bm{q}_z},
\end{equation}
and the Euclidean hadronic matrix element in momentum space
\begin{equation}\label{HLT}
	H_{E}^{(LT)}(k'_E,p_E)=i\sum_{\alpha_L}\frac{ A^{(L)}_{\alpha,E}(-\bm{k}',E_{\alpha_L})}{m+ik_E^0-E_{\alpha_L}}\left[1-e^{(m+ik_E^0-E_{\alpha_L})t_s}\right].
\end{equation}
Here, we define the total energy in Euclidean spac e as $p_E=(im,\bm{0})$, and the momentum carried by the current as $k'_E=(k_E^0,\bm{k}')$ ($\bm{k}'\in\Gamma$). The $ik_E^0$ in this equation comes from the wick rotation $k^0\to -ik_E^0$, as defined in Appendix~\ref{sec:App-EM}.

In $H_{E}^{(LT)}(k'_E,p_E)$, the temporal truncation $t_s$ removes the singularities in the denominator
\begin{equation}
	\begin{aligned}
		\left[\frac{1-e^{(m+ik_E^0-E_{\alpha_L})t_s}}{m+ik_E^0-E_{\alpha_L}}\right]\Bigg|_{m+ik_E^0-E_{\alpha_L}\to 0}=-t_s.
	\end{aligned}
\end{equation}
Therefore, the light hadronic singularities (the single-particle pole 1 or the multi-particle pole 2 in Fig.\ref{fig:Wick2}) do not affect the equivalence between $C$ and $C_{\text{lat}}$ in the loop integral of $I^{(LT)}$. Using this equivalence, we can rotate $I^{(LT)}$ back to Minkowski space as
\begin{equation}
	\label{CLT2}
	I^{(LT)}=\frac{1}{L^3}\sum_{\bm{k}'\in\Gamma} \int \frac{d^3 k}{(2\pi)^3}\delta_L(\bm{k}'-\bm{k})\int_{C} \frac{d k^0}{2\pi} L^\infty(k)H^{(LT)}(k',p).
\end{equation}
Here, we define the Minkowski space momentum as $k'=(k^0,\bm{k}')$ ($\bm{k}'\in\Gamma$), and the Minkowski space matrix element as
\begin{equation}\label{HMLT}
	H^{(LT)}(k',p)=i\sum_{\alpha_L}\frac{A^{(L)}_{\alpha}(-\bm{k}',E_{\alpha_L})}{m-k^0-E_{\alpha_L}}\left(1-e^{(m-k^0-E_{\alpha_L})t_s}\right),
\end{equation}
where $A^{(L)}_{\alpha}(-\bm{k}',E_{\alpha_L})=\langle f|J_{2}|\alpha_L(-\bm{k}',E_{\alpha_L})\rangle\langle \alpha_L(-\bm{k}',E_{\alpha_L})|J_{1}|i\rangle$ is defined in the finite-volume Minkowski space, and satisfies $A^{(L)}_{\alpha}(-\bm{k}',E_{\alpha_L})=c_{ME}A^{(L)}_{\alpha,E}(-\bm{k}',E_{\alpha_L})$.

The momentum space expression (\ref{CLT2}) shows that the EW$_\infty$ method mixes the the infinite-volume electroweak weight function $L^\infty(k)$ (which has continuous momentum $k=(k^0,\bm{k})$) and the finite-volume hadronic function $H^{(LT)}(k',p)$ (which has discrete momentum $k'=(k^0,\bm{k}')$ with $\bm{k}'\in\Gamma$). $\bm{k}$ and $\bm{k}'$ are related by the finite-volume $\delta$ function $\delta_L(\bm{k}'-\bm{k})$. If we compare $I^{(LT)}$ in Eq. (\ref{CLT2}) with $I^\infty$ in Eq. (\ref{Cinf}), we can see that $\Delta I=I^{(LT)}-I^\infty$ has two different sources:
\begin{enumerate}
	\item The difference between the finite-volume hadronic function $H^{(LT)}(k',p)$ and the infinite-volume hadronic function $H^\infty(k',p)$ at discrete momenta $\bm{k}'\in\Gamma$.
	\item The finite-volume effects introduced by the loop integral structure in EW$_\infty$ method
	\begin{equation}
		\frac{1}{L^3}\sum_{\bm{k}'\in\Gamma} \int \frac{d^3 k}{(2\pi)^3}\delta_L(\bm{k}'-\bm{k}).
	\end{equation}
	This integral structure comes from writing the coordinate-space integral in momentum space. As the volume goes to infinity, $\delta_L(\bm{k}'-\bm{k})$ becomes the normal $\delta$ function and this integral form converges to the physical form $\int \frac{d^3 k}{(2\pi)^3}$. Thus, this is a new kind of finite-volume effect. 
\end{enumerate}
We can seperate them as the first type of correction (denoted as $\Delta I_1$) and the second type of correction (denoted as $\Delta I_2$) by introducing an intermediate quantity
\begin{equation}\label{Itilde}
	\begin{aligned}
		\tilde{I}^{(L)}&=\frac{1}{L^3}\sum_{\bm{k}'\in\Gamma}\int \frac{d^3 k}{(2\pi)^3}\delta_L(\bm{k}'-\bm{k})\int_{C} \frac{d k^0}{2\pi} L^\infty(k)H^\infty(k',p),\\
		\Delta I_1&=I^{(LT)}-\tilde{I}^{(L)},\\
		\Delta I_2&=\tilde{I}^{(L)}-I^\infty.
	\end{aligned}
\end{equation}

We illustrate the idea of this paper in Fig. \ref{fig:deltaC}. Here, we use the $H$ bubble with solid box with subscript $LT$ to represent $H^{(LT)}(k',p)$, and the unboxed $H$ bubble to represent $H^{\infty}(k,p)$. We use the $L$ bubble to represent $L^\infty(k)$. We use the large dashed box to represent the loop integral structure introduced in the EW$_\infty$ method. Through these two correction steps, we can convert the lattice calculated quantity $I^{(LT)}$ into the physical quantity $I^\infty$.
\begin{figure}[htbp]
	\centering
	\includegraphics[width=0.98\textwidth]{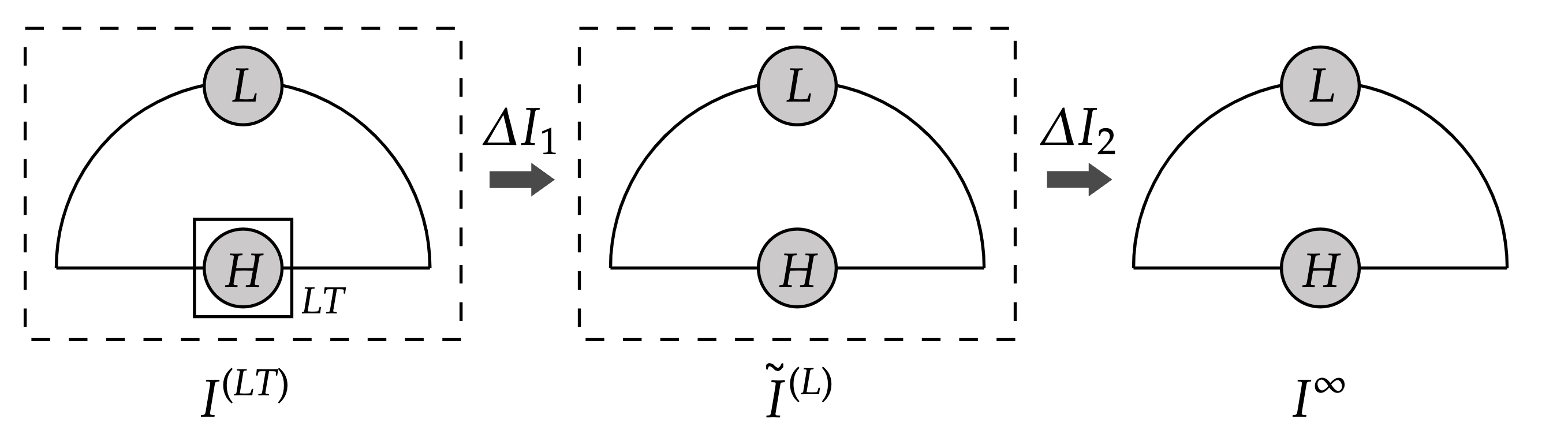}
	\caption{The correction idea of this paper.\label{fig:deltaC}}
\end{figure}

Next, we make some general discussion about these two corrections. In Section \ref{sec:one} and Section \ref{sec:two}, we apply the discussion here to the cases of single-particle and two-particle intermediate states respectively.

\subsubsection{First type of correction}
According to Eq. \ref{Itilde}, we can derive the correction formula for $\Delta I_1$ as
\begin{equation}\label{DC1}
	\begin{aligned}
		\Delta I_{1}&=I^{(LT)}-\tilde{I}^{(L)}\\
		&=\frac{1}{L^3}\sum_{\bm{k}'\in\Gamma} \int \frac{d^3 k}{(2\pi)^3}\delta_L(\bm{k}'-\bm{k})\int_{C} \frac{d k^0}{2\pi} L^\infty(k)\Delta H(k',p),\\
		\Delta H(k',p)&=H^{(LT)}(k',p)-H^\infty(k',p).
	\end{aligned}
\end{equation}
This shows that $\Delta I_1$ depends on the difference between the finite-volume hadronic matrix element $H^{(LT)}(k',p)$ and the infinite-volume hadronic matrix element $H^{\infty}(k',p)$ at discrete momentum $\bm{k}'\in\Gamma$. This difference is a correction term that has been studied extensively by previous literature, for example, Ref. \cite{Briceno:2019opb} provides the correction formula for the case of two-particle intermediate state. In Sections \ref{sec:one-1} and \ref{sec:two-1}, we will introduce the correction formulas for the cases of single-particle and two-particle, respectively.

\subsubsection{Second type of correction}
Another source of systematic error is the finite-volume effect caused by the integral structure in EW$_\infty$ method. To analyze this effect, we write $\tilde{I}^{(L)}$ as discrete momentum summation
\begin{equation}
	\begin{aligned}
		\tilde{I}^{(L)}&=\frac{1}{L^3}\sum_{\bm{k}'\in\Gamma} \hat{I}(\bm{k}'),\\
		\hat{I}(\bm{k}')&=\int \frac{d^3 k}{(2\pi)^3}\delta_L(\bm{k}'-\bm{k})\int_{C} \frac{d k^0}{2\pi} L^\infty(k)H^\infty(k',p).
	\end{aligned}
\end{equation}
In the infinite-volume limit $L\to\infty$, the sum $\frac{1}{L^3}\sum_{\bm{k}'\in\Gamma}$ becomes the integral $\int \frac{d^3 k'}{(2\pi)^3}$, and the finite-volume $\delta$ function $\delta_L(\bm{k}'-\bm{k})$ becomes the infinite-volume one $(2\pi)^3\delta^{(3)}(\bm{k}'-\bm{k})$, so $\tilde{I}^{(L)}$ converges to $I^\infty$. Therefore, $\Delta I_2$ is a new type of finite-volume effect introduced in the EW$_\infty$ method. 

We can see that the summation form here is more complex than the conventional Poisson summation formula, but the volume suppression behavior of $\Delta I_{2}=\tilde{I}^{(L)}-I^\infty$ also depends on the smoothness of the summand $\hat{I}(\bm{k}')$. In Appendix \ref{sec:Append sum}, we state and prove the following theorem:
\begin{thm}\label{thm:FV}
	Let $\tilde{l}(k)$ and $\tilde{h}(k)$ be functions of four-momentum $k=(k^0,\bm{k})$, and define their loop integral forms
	\begin{equation}
		\begin{aligned}
			I^{\infty}&=\int \frac{d^3k}{(2\pi)^3}\int \frac{dk^0}{2\pi}\tilde{l}(k)\tilde{h}(k),\\
			I^{(L)}&=\frac{1}{L^3}\sum_{\bm{k}'\in\Gamma}\int \frac{d^3k}{(2\pi)^3}\delta_L(\bm{k}'-\bm{k})\int \frac{dk^0}{2\pi}\tilde{l}(k)\tilde{h}(k')=\frac{1}{L^3}\sum_{\bm{k}'\in\Gamma}\hat{I}(\bm{k}'),
		\end{aligned}
	\end{equation}
	with $k=(k^0,\bm{k})$, $k'=(k^0,\bm{k}')$. $\bm{k}'\in\Gamma$ is the discrete momentum in finite volume. The finite-volume effect $I^{(L)}-I^{\infty}$ is related to the smoothness of $\hat{I}(\bm{k}')$ by (here, $\bm{k}'$ is extended to be a continuous real variable)
	\begin{enumerate}
		\item If $\hat{I}(\bm{k}')$ has singularities, then the finite-volume effect suppresses as $O(1/L)$;
		\item If $\hat{I}(\bm{k}')$ is continuously differentiable up to order $N$, then the finite-volume effect suppresses as $O(1/L^{N+1})$;
		\item If $\hat{I}(\bm{k}')$ is infinitely differentiable, then the finite-volume effect suppresses as $O(e^{-\Lambda L})$, where $\Lambda$ is the hadronic mass scale in the problem.
	\end{enumerate}
\end{thm}
This theorem implies that the smoothness of $\hat{I}(\bm{k}')$ needs to be investigated in the specific physical problem to analyze the behavior of this finite-volume effect $\Delta I_2$, which is very similar to the previous studies of finite-volume effects. The proof of this theorem relies on the mathematical relation between the smoothness of a function and the suppression behavior of its inverse Fourier transformation \cite{Reed:1975uy}, and we present the detailed proof in Appendix \ref{sec:Append sum}.

\subsection{\label{sec:decay}Decay width problem}
We now turn to another type of physical problem: the radiative decay of mesons. In these processes, mesons decay into final state leptons or photons by exchanging photons or $W^\pm$ bosons through two currents. Examples include the radiative decay to two (virtual) photons $P\to\gamma^{(*)}\gamma^{(*)}$ (such as $\pi^0/\eta/K_L\to \ell^+\ell^-\gamma$ or $\pi^0/\eta/K_L\to \ell^+\ell^-\ell^{'+}\ell^{'-}$), and the weak decay with (virtual) photon emission $P\to\ell \nu_\ell\gamma^{(*)}$ (such as $K^+\to \ell^+\nu_\ell\gamma$ or $K^+\to \ell^+\nu_\ell\ell^{'+}\ell^{'-}$). Fig.~\ref{fig:decay} shows the structure of the decay width of these processes, using $K^+\to \ell^+\nu_\ell\ell^{'+}\ell^{'-}$ as an example. They involve phase space integration over the hadron part and electroweak part of the square of the amplitude, so we can regard this structure as a generalized loop integral.
\begin{figure}[htbp]
	\centering
	\includegraphics[width=0.75\textwidth]{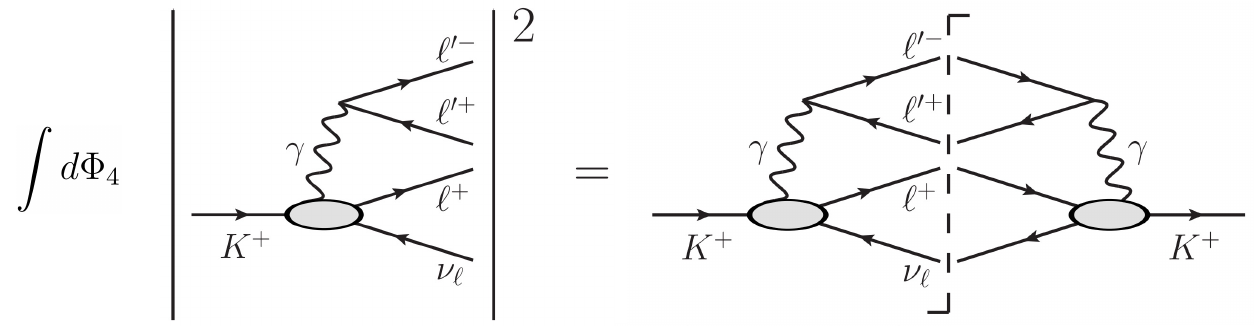}
	\caption{The structure of $K^+\to \ell^+\nu_\ell\ell^{'+}\ell^{'-}$ decay width: phase space integration over the hadron part and electroweak part of the square of the amplitude. \label{fig:decay}}
\end{figure}

We consider a general process of meson $P$ decaying into $l$ leptons or photons $n_1,\cdots,n_l$ through two currents $J_{1}$ and $J_{2}$, which can be either electromagnetic or weak currents. We denote the momentum of the initial meson by $p=(m_P,\bm{0})$, and the momentum of the final particles by $p_1,\cdots,p_l$. The two currents carry the momentum $k=(k^0,\bm{k})$ and $p-k=(m_P-k^0,-\bm{k})$, respectively, and they satisfy the energy-momentum conservation with the final particle momentum $p_1,\cdots,p_l$. The hadronic part in the infinite-volume Minkowski space is
\begin{equation}\label{HMdecay}
	H^{\infty,\mu \nu}(k,p)=\int  d^3 x \int_{-\infty}^\infty dt e^{i k\cdot x}\left\langle 0\left|T\left\{J_{1}^\mu(x) J_{2}^\nu(0)\right\}\right|P\right\rangle.
\end{equation}

The decay width can be generally written as
\begin{equation}\label{Ga}
	\begin{aligned}
		\Gamma^\infty=&\frac{1}{2m_P}\int d\Phi_l|\mathcal{M}(P\to n_1 \cdots n_l)|^2\\
		=&\frac{1}{2m_P}\int d\Phi_l H^{\infty,\mu\nu}(k,p)L^\infty_{\mu\nu\mu'\nu'}(p_1,\cdots,p_l)H^{\infty,\mu'\nu'}(k,p),
	\end{aligned}
\end{equation}
where $d\Phi_l$ is the $l$-body phase space of the final state particles, and $L_{\mu\nu\mu'\nu'}(p_1,\cdots,p_l)$ is the electroweak part of the square of the amplitude. Comparing the decay width here with the loop integral problem introduced before, we can see that they have a similar structure of multiplying the hadronic matrix element and the electroweak weight function and integrating over the momentum. The loop integral is replaced by the phase space integral here. In Ref.~\cite{Tuo:2021ewr}, we propose a lattice calculation method for such decay width. In this part, we point out that the method used in Ref.~\cite{Tuo:2021ewr} is an extension of the EW$_\infty$ method into the decay width problem.

In Ref.~\cite{Tuo:2021ewr}, for a specific continuous momentum $k=(k^0,\bm{k})$ in the phase space, we calculate the hadronic matrix element on the lattice as
\begin{equation}\label{HLTdecay}
		H_{\text{lat}}^{ \mu \nu}(k,p)=-i \int_{-t_s}^{t_s} d \tau \int_V d^3 x e^{k^0 \tau-i \bm{k} \cdot \bm{x}} \left(\left\langle 0\left|T\left\{J_{1,E}^\mu(x) J_{2,E}^\nu(0)\right\}\right|P\right\rangle\right)^{(L)}.
\end{equation}
Since $\bm{k}$ is defined as a continuous momentum, this definition is different from lattice calculation of hadronic matrix elements at discrete momentum. Using $H_{\text{lat}}^{ \mu \nu}(k,p)$ as an input, we define the decay width in the EW$_\infty$ method as
\begin{equation}\label{GaLT}
	\Gamma^{(LT)}=\frac{1}{2m_P}\int d\Phi_l [c_{ME}^{\mu\nu}H_{\text{lat}}^{ \mu \nu}(k,p)]L_{\mu\nu\mu'\nu'}(p_1,\cdots,p_l)[c_{ME}^{\mu'\nu'}H_{\text{lat}}^{ \mu \nu}(k,p)],
\end{equation}
where $c^{\mu\nu}_{ME}$ accounts for the difference between the operator definitions in Euclidean space and Minkowski space (see Appendix \ref{sec:App-EM}). 

To see this is a generalization of the EW$_\infty$ method, we rewrite the expression of $\Gamma^{(LT)}$ in momentum space, analogous to Eq.~\ref{CLT2}, as
\begin{equation}\label{GaLT2}
	\begin{aligned}
		\Gamma^{(LT)}=&\frac{1}{L^3}\sum_{\bm{k}'_1\in\Gamma}\frac{1}{L^3}\sum_{\bm{k}'_2\in\Gamma}\frac{1}{2m_P}\int d\Phi_l \delta_L(\bm{k}'_1-\bm{k})\delta_L(\bm{k}'_2-\bm{k})\\
		&\times  H^{(LT),\mu\nu}(k'_1,p)L^\infty_{\mu\nu\mu'\nu'}(p_1,\cdots,p_l)H^{(LT),\mu'\nu'}(k'_2,p),
	\end{aligned}
\end{equation}
where $H^{(LT),\mu\nu}(k'_{1/2},p)$ is the finite-volume Minkowski hadronic matrix element defined similarly to Eq.~(\ref{HMLT}). Here, the continuous momentum $\bm{k}$ is associated with the infinite-volume electroweak weight function $L^\infty_{\mu\nu\mu'\nu'}(p_1,\cdots,p_l)$, and is related to $p_1,\cdots,p_l$ by the energy-momentum conservation. On the other hand, $\bm{k}'_{1,2}$ are the discrete momenta related to the finite-volume hadronic matrix elements. $\Gamma^{(LT)}$ has a very similar structure as $I^{(LT)}$ in Eq.~(\ref{CLT2}), thus giving a generalization of the EW$_\infty$ method to the decay width problem.

Similar to the loop integral case, the correction $\Delta \Gamma=\Gamma^{(LT)}-\Gamma^\infty$ consists of two parts:
\begin{enumerate}
	\item 
	The first type of correction accounts for the difference between $H^{(LT),\mu\nu}(k'_{1,2},p)$ and $H^{\infty,\mu\nu}(k'_{1,2},p)$ at discrete momentum $\bm{k}'_{1,2}\in\Gamma$. 
	\item 
	The second type of correction is a new finite-volume effect introduced by the integral form in the EW$_\infty$ method $$\frac{1}{L^3}\sum_{\bm{k}'_1\in\Gamma}\frac{1}{L^3}\sum_{\bm{k}'_2\in\Gamma}\int d\Phi_l \delta_L(\bm{k}'_1-\bm{k})\delta_L(\bm{k}'_2-\bm{k}). $$
	As $L$ approaches infinity, the finite-volume $\delta_L$ function converges to the standard $\delta$ function, and the above expression converges to the phase space integral $\int d\Phi_l$.
\end{enumerate}

In the subsequent section, we focus on discussing the corrections in the loop integral problem. The following correction method can also be applied to the decay width problem.

\section{\label{sec:one}Single-particle intermediate state}
In this section, we apply the EW$_\infty$ method to the cases where the ground state of $H^{\infty}(k,p)$ in Eq. (\ref{HMdef}) is a single-particle state, denoted as $N$.  Since the finite-volume effects and temporal truncation effects are dominated by this single-particle ground state, we will neglect other excited states in this section. Take $t<0$ time ordering as an example, we can write the single-particle contribution as
\begin{equation}
\begin{aligned}
	\label{HMpi}
	H^\infty_{t<0}(k,p)&=i\frac{A_{N}(-\bm{k})}{2E_{N}(\bm{k})(m-k^0-E_{N}(\bm{k})+i\epsilon)}+\cdots,\\
	A_{N}(-\bm{k})&=\langle f|J_{2}|N(-\bm{k})\rangle \langle N(-\bm{k}) |J_{1}|i\rangle,
\end{aligned}
\end{equation}
where we define the energy of the single particle as $E_{N}(\bm{k})=\sqrt{\bm{k}^2+m_{N}^2}$. $J_{1}$ and $J_{2}$ represent operators in the physical problem, and we omit their Lorentz indices. 

The lattice data $H_{E}^{(L)}(x)$ in Eq. (\ref{HELx}) is defined in the finite-volume Euclidean space. For $\tau<0$ time ordering, the single-particle contribution can be expressed as
\begin{equation}\label{HELone}
	\begin{aligned}			
		H_E^{(L)}(\tau<0,\bm{x})&=\frac{1}{L^3}\sum_{\bm{k}'\in\Gamma}\frac{A_{N,E}(-\bm{k}')}{2E_{N}(\bm{k}')}e^{(E_{N}(\bm{k}')-m_n)\tau+i\bm{k}'\cdot\bm{x}}+\cdots,\\
		A_{N,E}(-\bm{k}')&=\langle f|J_{2,E}|N(-\bm{k}')\rangle \langle N(-\bm{k}') |J_{1,E}|i\rangle,
	\end{aligned}
\end{equation}
where $\bm{k}'$ is the discrete momentum $\Gamma=\left\{\bm{k}'|\bm{k}'=\frac{2\pi}{L}\bm{n}\right\}$. Since $|i/f\rangle$ and $|N(-\bm{k}')\rangle$ are both vacuum or single-particle states, we can ignore the exponentially suppressed finite-volume effects in $A_{N,E}(-\bm{k}')$. The convention difference of operators defined in Euclidean space or Minkowski space is $A_{N}(-\bm{k})=c_{ME}A_{N,E}(-\bm{k})$.

The physical quantity of interest in this section is the loop integral involving this single-particle dominated hadronic matrix element. The physical examples of using the EW$_\infty$ method include QED self-energy of hadrons~\cite{Feng:2018qpx,Feng:2021zek}, radiative correction of leptonic decay~\cite{Christ:2020jlp,Christ:2023lcc}, rare kaon decay $K^+\to\pi^+\nu\bar{\nu}$~\cite{Christ:2020hwe}, kaon decay with virtual photon emission~\cite{Tuo:2021ewr}, etc. In this section, we discuss how to correct the systematic errors in Eq.~(\ref{Itilde}) in this case. For convenience, we use a model generalized from the QED self-energy problem to illustrate our method.

\subsection{\label{sec:one-eg}Example: model generalized from QED self-energy problem}
\begin{figure}[htbp]
	\centering
	\includegraphics[width=0.6\textwidth]{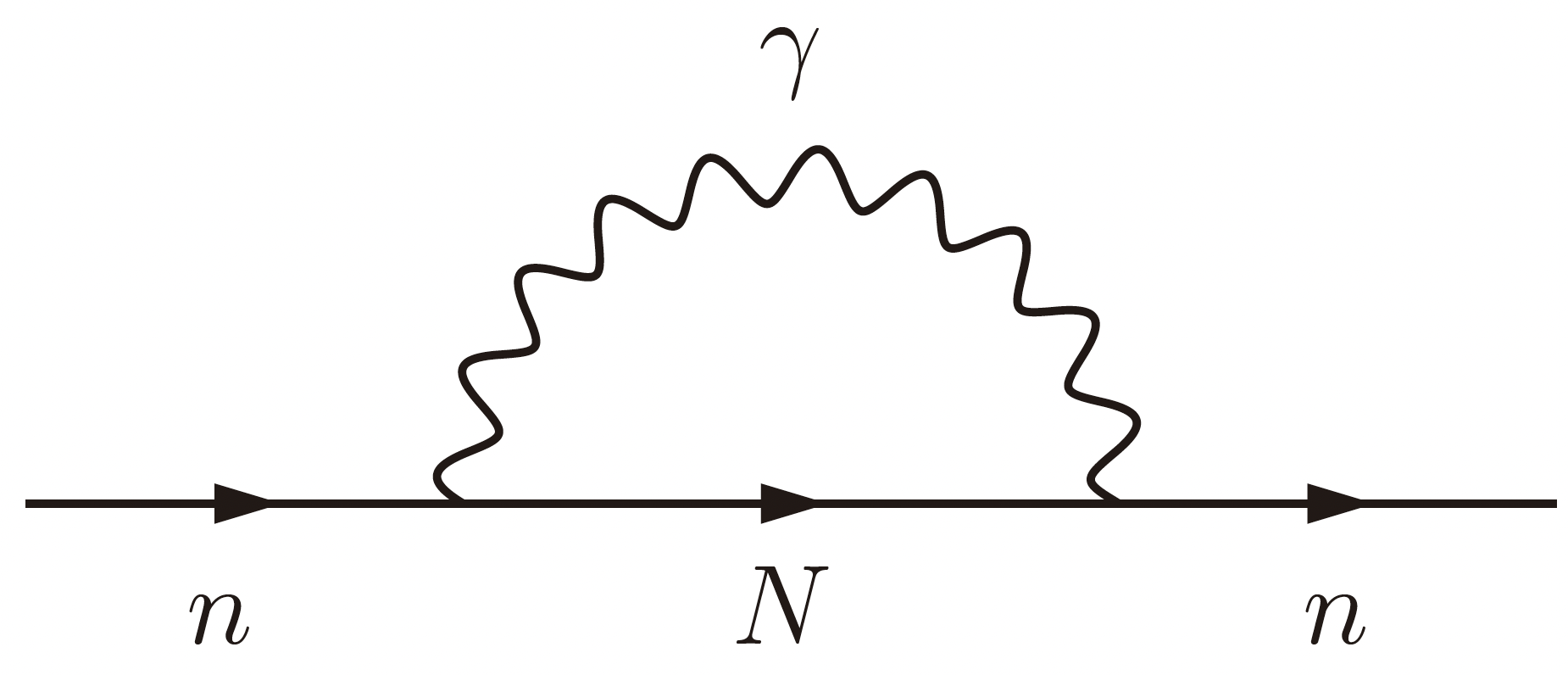}
	\caption{An example of a loop integral with a single-particle intermediate state.\label{fig:selfEmodel}}
\end{figure}
As shown in Fig. \ref{fig:selfEmodel}, we generalize the self-energy problem to a model where the intermediate particle can have different mass compared to the initial/final particle. We denote the initial/final particle and intermediate particle as $n$ and $N$, respectively. This model contains three different cases with $m_N>m_n$, $m_N=m_n$, or $m_N<m_n$. This generalization is for the convenience of including physical problems where the intermediate state is lighter than the initial/final state, such as rare kaon decay $K^+\to\pi^+\nu\bar{\nu}$ with a $\pi$ intermediate state \cite{Christ:2020hwe}. The loop integral is given by
\begin{equation}
	I^\infty =\int \frac{d^4 k}{(2\pi)^4}L^\infty(k)H^\infty(k,p),
\end{equation}
where $p=(m_n,\bm{0})$ is the total input momentum of the integral. $L^\infty(k)$ is the electroweak weight function that contains the photon propagator and other coefficients in the self-energy definition:
\begin{equation}
	\begin{aligned}
		L^{\infty}(k)&=\frac{i}{2}\frac{-i}{k^2+i\epsilon}.
	\end{aligned}
\end{equation}

The hadronic matrix element in this problem is defined as
\begin{equation}\label{HMpi-model}
	H^\infty(k,p)=\int d^3 x\int_{-\infty}^\infty dt e^{ik\cdot x}g_{\mu\nu}\langle n|T[J_{\text{em}}^\mu(t,\bm{x})J_{\text{em}}^\nu(0)]|n\rangle,
\end{equation}
where $J_{\text{em}}^\mu=(\frac 23 e\bar{u}\gamma^\mu u-\frac 13 e \bar{d}\gamma^\mu d)_M$ is the Minkowski space electromagnetic current, and $|i\rangle=|f\rangle=|n\rangle$ are the initial and final states. In this problem, the $t>0$ and $t<0$ parts correspond to the same physical contribution, which is dominated by the single-particle intermediate state $N$. For the $t<0$ part, $A_{N}(-\bm{k})$ in Eq. (\ref{HMpi}) is defined as $A_{N}(-\bm{k})=g_{\mu\nu}\langle n|J_{\text{em}}^\nu|N(-\bm{k})\rangle \langle N(-\bm{k}) |J_{\text{em}}^\mu|n\rangle$.

We use the EW$_\infty$ method to calculate the loop integral as
\begin{equation}\label{ILT_1pt}
	I^{(LT)}=\int_V d^3 x\int_{-t_s}^{t_s} d\tau L_E^\infty(\tau,\bm{x})H_E^{(L)}(\tau,\bm{x}),
\end{equation}
where $L_E^\infty(\tau,\bm{x})$ is the electroweak weight function in infinite-volume Euclidean space, defined by Eq. (\ref{LE}) as
\begin{equation}
	L_E^\infty(\tau,\bm{x})=\frac{1}{2}\frac{1}{4\pi^2 (\bm{x}^2+\tau^2)}.
\end{equation}
The lattice input $H_E^{(L)}(\tau,\bm{x})$ is defined as
\begin{equation}\label{HELone-model}
	H_E^{(L)}(\tau,\bm{x})=\delta^{\mu\nu}(\langle  n|T[J_{\text{em},E}^\mu(\tau,\bm{x})J_{\text{em},E}^\nu(0)]|n\rangle)^{(L)},
\end{equation}
where $J_{\text{em},E}^\mu=(\frac 23 e\bar{u}\gamma^\mu u-\frac 13 e\bar{d}\gamma^\mu d)_E$ is the Euclidean space electromagnetic current. For the $t<0$ part, we take $A_{N,E}(-\bm{k})=\delta_{\mu\nu}\langle n|J_{\text{em},E}^\nu|N(-\bm{k})\rangle \langle N(-\bm{k}) |J_{\text{em},E}^\mu|n\rangle$ in Eq. (\ref{HELone}). In Appendix \ref{sec:App-EM}, we show that in this problem $c_{ME}=1$ and $A_{N,E}(-\bm{k})=A_{N}(-\bm{k})$, so we use $A_{N}(-\bm{k})$ for both in the following.

We use this model as an example to introduce how to do the correction in the single-particle case. In this problem, the $t>0$ and $t<0$ parts correspond to the same physical contribution. For convenience, we focus on $t<0$ contribution and ignore $t>0$ part in the following discussion.

\subsection{\label{sec:one-1}First type of correction}
We consider the first type of correction ($\Delta I_1$). According to Eq. (\ref{DC1}), $\Delta I_1$ depends only on the difference between the finite-volume matrix element $H^{(LT)}(k',p)$ and the infinite-volume matrix element $H^{\infty}(k',p)$ at the discrete momentum $\bm{k}'\in\Gamma$. Under the single-particle dominance, $H^{\infty}(k',p)$ is given by Eq. (\ref{HMpi}), and $H^{(LT)}(k',p)$ is
\begin{equation}
	H^{(LT)}(k',p)=i\frac{A_{N}(-\bm{k}')}{2E_{N}(\bm{k}')(m_n-k^0-E_{N}(\bm{k}'))}\left(1-e^{(m_n-k^0-E_{N}(\bm{k}'))t_s}\right)+\cdots.
\end{equation}

The correction of the hadronic matrix element at the discrete momentum $\bm{k}'$ is
\begin{equation}
	\begin{aligned}
		\Delta H(k',p)&=H^{(LT)}(k',p)-H^\infty(k',p),\\
		&=i\frac{A_{N}(-\bm{k}')}{2E_N(\bm{k}')(m_n-k^0-E_{N}(\bm{k}'))}\left(1-e^{(m_n-k^0-E_{N}(\bm{k}'))t_s}\right)\\&-i\frac{A_{N}(-\bm{k}')}{2E_{N}(\bm{k}')(m_n-k^0-E_{N}(\bm{k}')+i\epsilon)}.
	\end{aligned}
\end{equation}
Here, $H^{(LT)}(k',p)$ is a function without singularities, so we can add $i\epsilon$ to its denominator and simplify $H^\infty(k',p)$ as
\begin{equation}
	\Delta H(k',p)=i\frac{A_{N}(-\bm{k}')}{2E_{N}(\bm{k}')(m_n-k^0-E_{N}(\bm{k}')+i\epsilon)}\left(-e^{(m_n-k^0-E_{N}(\bm{k}'))t_s}\right).
\end{equation}
Then $\Delta I_1$ is given by
\begin{equation}
	\begin{aligned}
		\Delta I_{1}&=\frac{1}{L^3}\sum_{\bm{k}'\in\Gamma} \int \frac{d^3 k}{(2\pi)^3}\delta_L(\bm{k}'-\bm{k})\int_{C} \frac{d k^0}{2\pi} L^\infty(k)\Delta H(k',p).
	\end{aligned}
\end{equation}

We can further perform the contour integration in the $k_0$ direction and simplify the correction formula as
\begin{equation}\label{DC1_1pt}
	\begin{aligned}
		\Delta I_{1}=&\frac{1}{L^3}\sum_{\bm{k}'\in\Gamma} \int \frac{d^3 k}{(2\pi)^3}\delta_L(\bm{k}'-\bm{k})\frac{A_{N}(-\bm{k}')}{2E_\gamma(\bm{k})2E_{N}(\bm{k}')(E_\gamma(\bm{k})+E_{N}(\bm{k}')-m_n-i\epsilon)},\\ &\times \left(-e^{(m_n-E_\gamma(\bm{k})-E_{N}(\bm{k}'))t_s}\right),
	\end{aligned}
\end{equation}
where $E_\gamma(\bm{k})=|\bm{k}|$ is the energy of the photon. To avoid introducing model dependence when calculating $\Delta I_{1}$, $A_N(-\bm{k}')$ can be directly reconstructed from lattice data $H_{E}^{(L)}(\tau,\bm{x})$. We assume that the time cutoff $t_s$ is large enough for the ground state dominance and ignore the around-the-world effects, we have
\begin{equation}
	\begin{aligned}
		\frac{A_N(-\bm{k}')}{2E_N(\bm{k}')}=&e^{(E_N(\bm{k}')-m_n)t_s}\int_V d^3x H_E^{(L)}(-t_s,\bm{x})e^{-i\bm{k}'\cdot\bm{x}}\\
		=&e^{(E_N(\bm{k}')-m_n)t_s}H_E^{(L)}(-t_s,\bm{k}'),
	\end{aligned}
\end{equation}
where $H_E^{(L)}(-t_s,\bm{k}')$ is the discrete Fourier transformation of the lattice data $H_E^{(L)}(-t_s,\bm{x})$. Using $H_E^{(L)}(-t_s,\bm{k}')$ as input, we can rewrite correction formula as
\begin{equation}\label{DC1_1pt_IVR}
	\begin{aligned}
		\Delta I_{1}=&\frac{1}{L^3}\sum_{\bm{k}'\in\Gamma} \int \frac{d^3 k}{(2\pi)^3}\delta_L(\bm{k}'-\bm{k})\frac{H_E^{(L)}(-t_s,\bm{k}')}{2E_\gamma(\bm{k})(E_\gamma(\bm{k})+E_{N}(\bm{k}')-m_n-i\epsilon)} \left(-e^{-E_\gamma(\bm{k})t_s}\right).
	\end{aligned}
\end{equation}
Therefore, $\Delta I_1$ can be calculated from the lattice data without introducing form factors of $A_N(-\bm{k}')$. In practical calculations, since the high momentum contributions are suppressed by the large $t_s$ in the exponential term, only a few lower momentum modes $\bm{k}'$ need to be calculated. 

We numerically implements $I^{(LT)}$ and the corrected $\tilde I^{(L)}=I^{(LT)}-\Delta I_1$, as shown in Fig.~\ref{fig:CLT}.  For $m_N<m_n$, we only show the real part of the loop integral. In the model, we choose $L=24$ and $a=1$, and take the form factor as the same function with electromagnetic form factor of $\pi$, $\langle n(p_1)|J^\mu|N(p_2)\rangle=e F^{(\pi)}((p_1-p_2)^2)(p_1+p_2)^\mu$. The initial/final state mass $m_n=0.14$ is fixed, and the intermediate state mass $m_N$ is taken as three different cases: $m_N>m_n$, $m_N=m_n$, and $m_N<m_n$. We should mention that only $m_N=m_n$ case corresponds to physical self energy problem. Another two cases are only for the illustration of the method. In calculation of $\Delta I_1$, we choose modes with $\frac{L}{2\pi}|(\bm{k}')_{x,y,z}|\leq 2$ for correction. From the figure, we see that $\Delta I_1$ corrects the temporal truncation effect on the lattice, which is exponentially convergent when $m_N\geq m_n$, and is exponentially divergent when $m_N<m_n$.
\begin{figure}[htbp]
	\centering \includegraphics[width=1\textwidth]{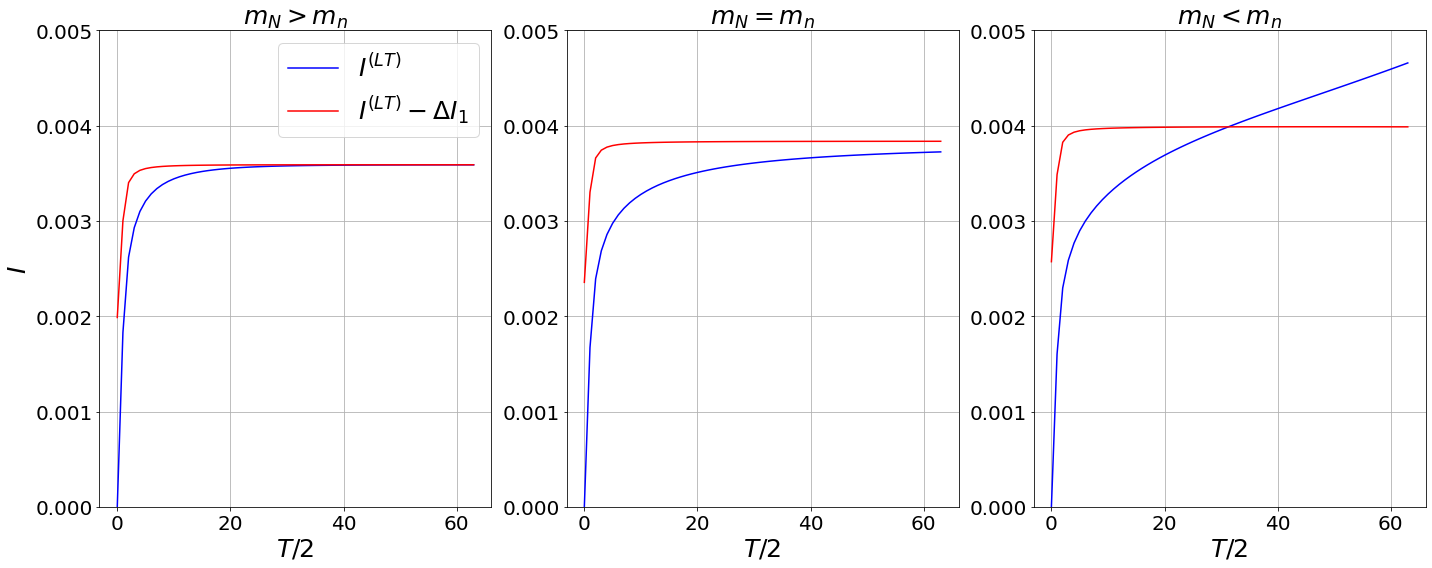} \caption{The loop integral $I^{(LT)}$ and the result after the first type of correction $\tilde{I}^{(L)}=I^{(LT)}-\Delta I_1$ ($L=24$). $m_n$ is fixed at 0.14, and $m_N$ is taken as three cases: $m_N=0.2>m_n$, $m_N=0.14=m_n$, and $m_N=0.1<m_n$.\label{fig:CLT}} 
\end{figure}

\subsection{\label{sec:one-2}Second type of correction}
The second type of correction ($\Delta I_{2}$) is the finite-volume effect caused by summation
\begin{equation}
	\begin{aligned}
		\tilde{I}^{(L)}&=\frac{1}{L^3}\sum_{\bm{k}'\in\Gamma} \hat{I}(\bm{k}'),\\
		\hat{I}(\bm{k}')&=\int \frac{d^3 k}{(2\pi)^3}\delta_L(\bm{k}'-\bm{k})\int_{C} \frac{d k^0}{2\pi} L^\infty(k)H^\infty(k',p).
	\end{aligned}
\end{equation}
According to Theorem \ref{thm:FV}, the volume suppression behavior of $\Delta I_2$ with increasing volume depends on the smoothness of $\hat{I}(\bm{k}')$. We first prove that $\hat{I}(\bm{k}')$ has no singularity, and then discuss whether it is infinitely differentiable.

\subsubsection{Proof of no singularity}
We first study the integral structure of $\hat{I}(\bm{k}')$ and examine whether possible on-shell intermediate states can cause singularities in the summand. As shown in Fig.~\ref{fig:selfEcut}, the most singular case is the single-particle intermediate state $N$ and the photon are both on-shell (the cut line indicates on-shell). We prove that this contribution does not cause singularity in $\hat{I}(\bm{k}')$. This contribution can be written as
\begin{equation}
	\hat{I}(\bm{k}')=\int \frac{d^3 k}{(2\pi)^3}\delta_L(\bm{k}'-\bm{k})\frac{A_N(-\bm{k}')}{2E_\gamma(\bm{k})2E_N(\bm{k}')\left(E_\gamma(\bm{k})+E_N(\bm{k}')-m_n-i\epsilon\right)}+\cdots.
\end{equation}

\begin{figure}[htbp]
	\centering
	\includegraphics[width=0.6\textwidth]{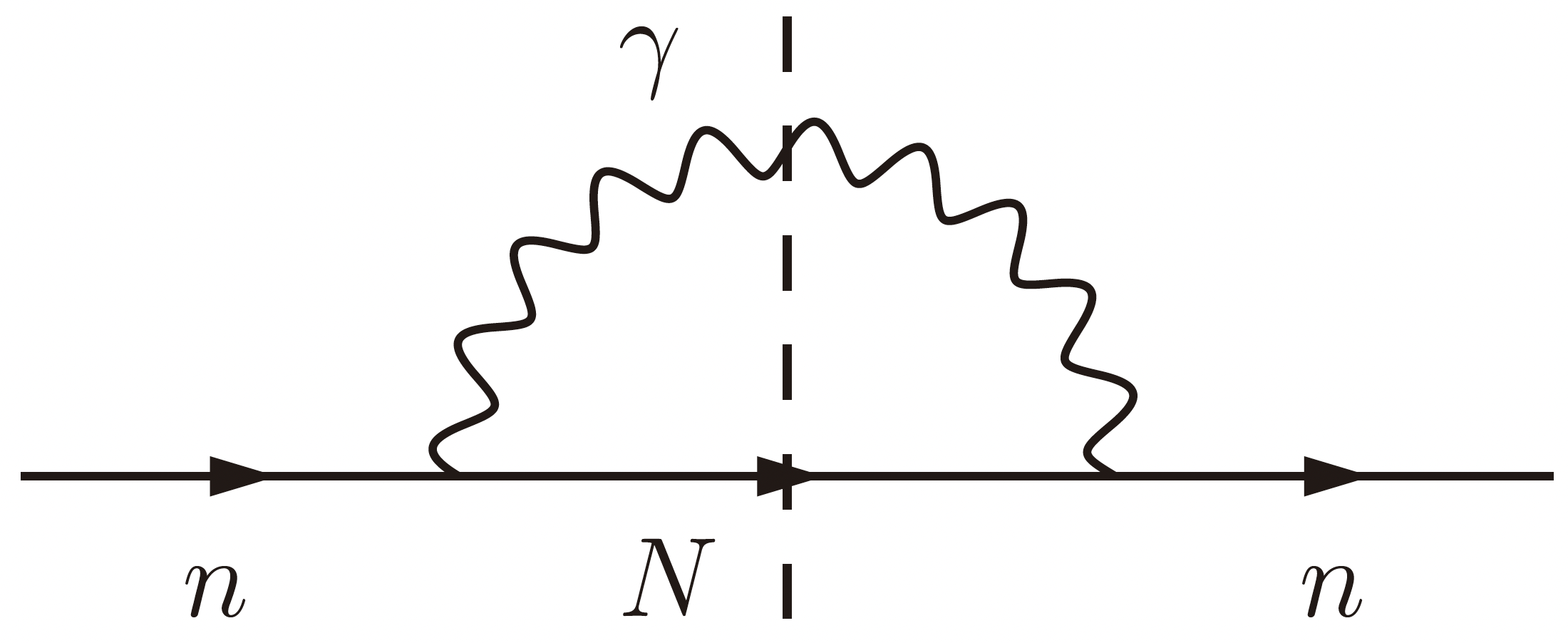}
	\caption{Analysis of the singularity structure of $\hat{I}(\bm{k}')$ by the on-shell intermediate state contribution.\label{fig:selfEcut}}
\end{figure}

In order to study whether the above integral has singularities, we introduce a mathematical theorem that is well-known in quantum field theory research~\cite{Eden:1966dnq,Sterman:1993hfp}:
\begin{thm}\label{thm:singularity}
	Consider a function of $w$:
	\begin{equation}
		I(w)=\int_C d\zeta F(\zeta,w),
	\end{equation}
	where C is the integration path of $\zeta$ in the complex plane, with endpoints $\zeta_a$ and $\zeta_b$. $F(\zeta,w)$ is an analytic function of $\zeta$, except for a few isolated singularities at $\zeta=\zeta_n(w)$. Then $I(w)$ can have singularities only in the following two cases:
	\begin{enumerate}
		\item Endpoint singularity. A singularity $\zeta_i(w)$ approaches the endpoint of the integration path $\zeta_a$ or $\zeta_b$ as $w\to w_0$, then $I(w)$ may have an endpoint singularity at $w=w_0$.
		\item Pinch singularity. Two singularities $\zeta_i(w)$ and $\zeta_{i'}(w)$ approach from opposite sides of the integration path C to the same point $\zeta_i(w_0)=\zeta_{i'}(w_0)$ as $w\to w_0$, then $I(w)$ may have a pinch singularity at $w=w_0$.
	\end{enumerate}
	Fig. \ref{fig:singularity} shows the possible situations of singularities near the integration path $C$. $P_1$ is an isolated pole that can be avoided by deforming the integration contour. $P_2$ and $P_3$ approach each other and may cause a pinch singularity. $P_4$ approaches $\zeta_b$ and may lead to an endpoint singularity.
\end{thm}

\begin{figure}[htbp]
	\centering
	\includegraphics[width=0.4\textwidth]{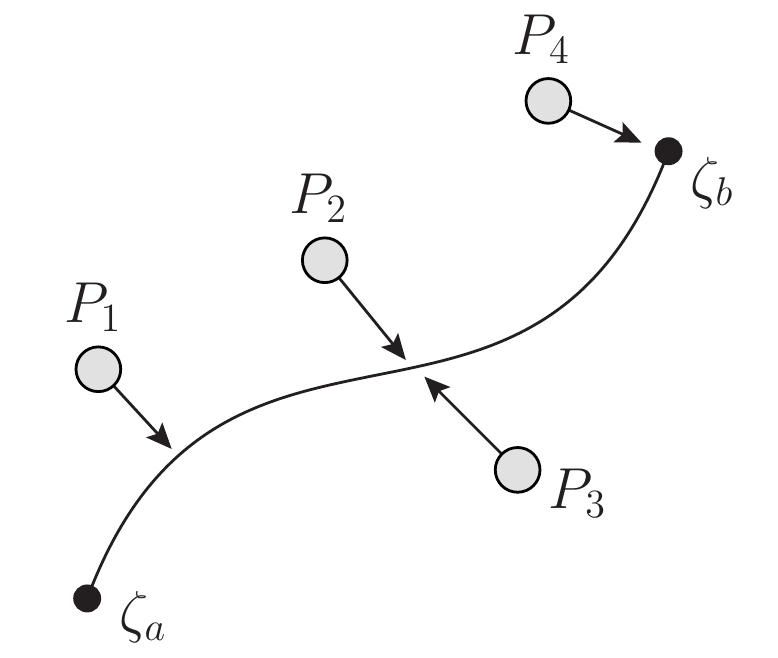}
	\caption{Isolated pole, endpoint singularies and pinch singularities. \label{fig:singularity}}
\end{figure}

In the integral expression of $\hat{I}(\bm{k}')$, $\zeta$ and $w$ in the theorem are replaced by $\bm{k}$ and $\bm{k}'$ respectively, and the endpoint singularities in the physical problem are infrared singularities. The denominator of the integrand have poles at $(\bm{k})^{(0)}=\bm{0}$ and $(\bm{k})^{(0)}=m_n-E_N(\bm{k}')$.

When $m_{n} \geq m_{N}$, both the pinch singularity and the endpoint singularity (infrared singularity) occur at $|\bm{k}'|=\sqrt{m_{n}^2-m_N^2}$. The integral form of $\hat{I}(\bm{k}')$ around this point is
\begin{equation}
	\hat{I}(\bm{k}')|_{|\bm{k}'|=\sqrt{m_{n}^2-m_N^2}}=\int\frac{d^3 k}{(2\pi)^3}\delta_L(\bm{k}'-\bm{k})\frac{A_{N}(-\bm{k}')}{2|\bm{k}|2m_{N}(|\bm{k}|+m_n-m_{n})}\sim \int d^3 k \frac{1}{|\bm{k}|^2}.
\end{equation}

When $m_n< m_N$, only the endpoint singularity (infrared singularity) occurs, and the integral form of $\hat{I}(\bm{k}')$ is
\begin{equation}
	\hat{I}(\bm{k}')=\int\frac{d^3 k}{(2\pi)^3}\delta_L(\bm{k}'-\bm{k})\frac{A_{N}(-\bm{k}')}{2|\bm{k}|2E_{N}(\bm{k}')(|\bm{k}|+E_{N}(\bm{k}')-m_n)}\sim \int d^3 k \frac{1}{|\bm{k}|}.
\end{equation}

We can see that in both cases above, the integral in the infrared region is finite. This is because $\hat{I}(\bm{k}')$ has a very similar integral structure with the infrared safe quantity $I^\infty$. This is the main advantage of EW$_\infty$ method: the electroweak weight function $L^\infty(k)$ may have singularities, but they do not cause singularity in the summand $\hat{I}(\bm{k}')$. Fig. \ref{fig:Chat} numerically confirms that $\hat{I}(\bm{k}')$ has no infrared singularity.

\subsubsection{Cusp effect}
We have shown that $\hat{I}(\bm{k}')$ is free of singularities, but we still need to check its smoothness and differentiability. In this part, we show that when $m_N\geq m_n$, i.e., the intermediate state is not lighter than the initial and final states, $\hat{I}(\bm{k}')$ is infinitely smooth and differentiable; when $m_N<m_n$, i.e., the intermediate state is lighter than the initial and final states, the summand is not smooth near the threshold $|\bm{k}'|=\sqrt{m_n^2-m_N^2}$, which we call the cusp effect.
To prove this, we simplify $\hat{I}(\bm{k}')$ as
\begin{equation}
	\begin{aligned}
		\hat{I}(\bm{k}')=&\int \frac{d^3 k}{(2\pi)^3}\delta_L(\bm{k}'-\bm{k})\frac{A_N(-\bm{k}')}{2E_\gamma(\bm{k})2E_N(\bm{k}')(E_\gamma(\bm{k})+E_N(\bm{k}')-m_n-i\epsilon)}\\
		=&-\frac{A_N(-\bm{k}')}{32\pi^3\sqrt{|\bm{k}'|^{2}+m_N^2}}\int_0^\infty d|\bm{k}|\frac{|\bm{k}|}{m_n-\sqrt{|\bm{k}'|^{2}+m_N^2}-|\bm{k}|+i\epsilon}\int d\Omega_{\bm{k}}\delta_L(\bm{k}'-\bm{k}).
	\end{aligned}
\end{equation}
We let $w=|\bm{k}|$, $z=m_n-\sqrt{|\bm{k}'|^{2}+m_N^2}$, and denote the other parts by a function $g(w,z)$, then the integral in the above equation has the form
\begin{equation}\label{fzsingle}
	f(z)=\int_0^\infty \frac{wg(w,z)}{z-w+i\epsilon}dw,
\end{equation}
where $z$ depends on $|\bm{k}'|$, and $z=0$ corresponds to the threshold $|\bm{k}'|=\sqrt{m_n^2-m_N^2}$. To analyze the cusp effect, we decompose $f(z)$ as
	\begin{equation}
		f(z)=\int_0^\Lambda \frac{zg(z,z)}{z-w+i\epsilon}dw+\int_0^\Lambda \frac{\left[wg(w,z)-zg(z,z)\right]}{z-w+i\epsilon}dw
		+\int_\Lambda^\infty \frac{wg(w,z)}{z-w+i\epsilon}dw,
	\end{equation}
	where $\Lambda$ is a cutoff satisfied $\Lambda>m_n>z$. Apparently, both the second term and the third term do not have the singularity on the denominator and are irrelevant to the cusp effect.  The first term can be calculated explicitly as
	\begin{equation}
		\int_0^\Lambda \frac{zg(z,z)}{z-w+i\epsilon}dw=\begin{cases}
			zg(z,z)\left[\log(z)-\log(\Lambda-z)+i\pi\right]&\quad z> 0;\\
			zg(z,z)\left[\log(-z)-\log(\Lambda-z)\right]&\quad z< 0.
		\end{cases} 
	\end{equation}
	Then, the non-smoothness around the threshold $z=0$, i.e., the cusp effect, is given by $\log(z)+i\pi$ with $z>0$ and $\log(-z)$ with $z<0$.

When $m_N\geq m_n$, $z$ always satisfies $z\leq 0$, so the integral does not cross the non-smooth point of $z=0$. Thus, $\hat{I}(\bm{k}')$ is infinitely differentiable. When $m_N< m_n$, $\hat{I}(\bm{k}')$ is non-smooth at the threshold of $z=0$. For $z>0$, i.e., $0\leq|\bm{k}'|<\sqrt{m_\pi^2-m_N^2}$, the real part of the integral is given by the principal value integral, and the imaginary part is given by the residue at the pole; whereas for $z<0$, i.e., $|\bm{k}'|>\sqrt{m_n^2-m_N^2}$, the real part is a regular integral, and there is no imaginary part. Therefore, although the real and imaginary parts of $\hat{I}(\bm{k}')$ are continuous at the threshold, they are not infinitely differentiable. As shown in Fig.~\ref{fig:Chat}, we numerically calculated the real part and its first derivative of $\hat{I}(\bm{k}')$ for three cases: $m_N>m_n$, $m_N=m_n$, and $m_N<m_n$, confirming the conclusion here.
\begin{figure}[htbp]
	\centering
	\includegraphics[width=1\textwidth]{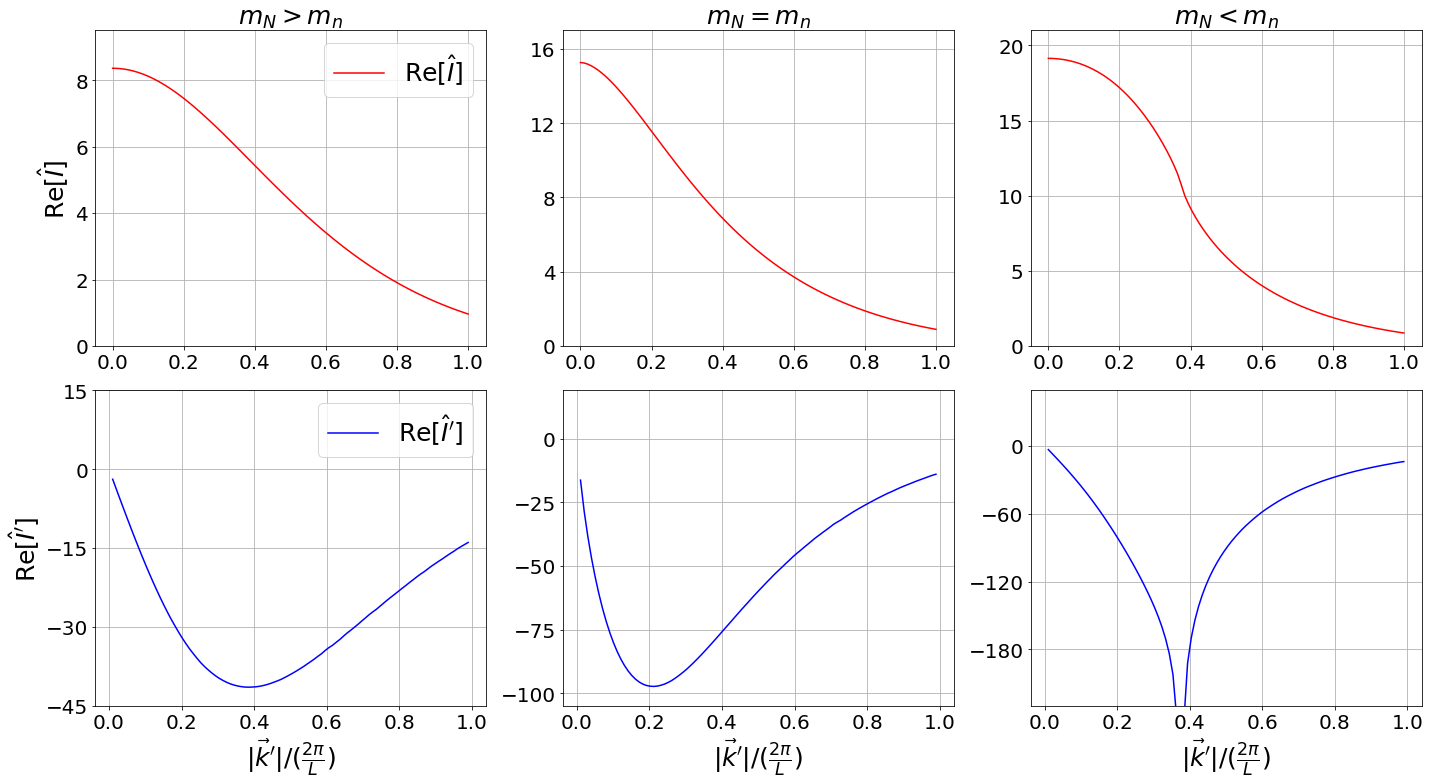}
	\caption{The real part and its first derivative of $\hat{I}(\bm{k}')$. Here $m_n$ is fixed at 0.14, and $m_N$ is taken as three cases: $m_N=0.2>m_n$, $m_N=0.14=m_n$, and $m_N=0.1<m_n$.\label{fig:Chat}}
\end{figure}

According to Theorem \ref{thm:FV}, we can conclude that: when the intermediate state is heavier than or equal to the initial and final states ($m_N\geq m_n$), $\Delta I_2$ is an exponentially suppressed finite-volume effect; when the intermediate state is lighter than the initial and final states ($m_N<m_n$), $\Delta I_2$ is a power-law suppressed finite-volume effect due to the cusp effect. To avoid this power-law finite-volume effect, next we show that the infinite volume reconstruction (IVR) method proposed in Ref. \cite{Feng:2018qpx,Christ:2020hwe} offers a technique to avoid this cusp effect.

\subsubsection{Avoiding cusp effect by IVR method}
The infinite volume reconstruction (IVR) method \cite{Feng:2018qpx,Christ:2020hwe} is defined as
\begin{equation}\label{IVRV}
	\begin{aligned}
		I^{(L)}_{\mathrm{IVR}}&=I_{\mathrm{IVR},s}^{(L)}+I_{\mathrm{IVR},l}^{(L)},\\
		I_{\mathrm{IVR},s}^{(L)}&=\int_V d^3x \int_{-t_s}^{0} d\tau L_E^\infty(\tau,\bm{x})H_{E}^{(L)}(\tau,\bm{x}),\\
		I_{\mathrm{IVR},l}^{(L)}&=\int_V d^3x \tilde{L}_E^\infty(t_s,\bm{x})H_{E}^{(L)}(-t_s,\bm{x}).
	\end{aligned}
\end{equation}
Here we only consider the calculation of $t<0$ time ordering. The short-range part $I_{\mathrm{IVR},s}^{(L)}$ is identical to $I^{(LT)}$, so the IVR method is based on the EW$_\infty$ method. The long-range part $I_{\mathrm{IVR},l}^{(L)}$ provides the correction of the temporal truncation effect. The electroweak weight function in $I_{\mathrm{IVR},l}^{(L)}$ is given by
\begin{equation}\label{IVRl}
	\tilde{L}_E^\infty(t_s,\bm{x})=\int \frac{d^3 k}{(2\pi)^3}e^{-i\bm{k}\cdot\bm{x}}\frac{e^{-E_{\gamma}(\bm{k})t_s}}{2E_\gamma(\bm{k})(E_\gamma(\bm{k})+E_N(\bm{k})-m_n-i\epsilon)}.
\end{equation}

We can rewrite $I_{\mathrm{IVR},l}^{(L)}$ in momentum space and compare it with $\Delta I_{1}$ in Eq. (\ref{DC1_1pt_IVR}):
\begin{equation}\label{DC1_1pt_IVR2}
	\begin{aligned}
		-I_{\mathrm{IVR},l}^{(L)}=&\frac{1}{L^3}\sum_{\bm{k}'\in\Gamma} \int \frac{d^3 k}{(2\pi)^3}\delta_L(\bm{k}'-\bm{k})\frac{H_E^{(L)}(-t_s,\bm{k}')}{2E_\gamma(\bm{k})(E_\gamma(\bm{k})+E_{N}(\bm{k})-m_n-i\epsilon)} \left(-e^{-E_\gamma(\bm{k})t_s}\right),\\
		\Delta I_{1}=&\frac{1}{L^3}\sum_{\bm{k}'\in\Gamma} \int \frac{d^3 k}{(2\pi)^3}\delta_L(\bm{k}'-\bm{k})\frac{H_E^{(L)}(-t_s,\bm{k}')}{2E_\gamma(\bm{k})(E_\gamma(\bm{k})+E_{N}(\bm{k}')-m_n-i\epsilon)} \left(-e^{-E_\gamma(\bm{k})t_s}\right).
	\end{aligned}
\end{equation}
The only difference is that $E_N(\bm{k}')$ in the denominator of $\Delta I_{1}$ is replaced by $E_N(\bm{k})$ in $I_{\mathrm{IVR},l}^{(L)}$. Due to the existence of $\delta_L(\bm{k}'-\bm{k})$, they are equivalent in the infinite-volume limit: $\lim_{L\to\infty} (-I_{\mathrm{IVR},l}^{(L)})=\lim_{L\to\infty} \Delta I_1$. Therefore, the IVR method provides a modified version of the temporal truncation effect $\Delta I_{1}$. In the infinite-volume limit, both methods converge to the physical quantity
\begin{equation}
	\lim_{L\to\infty}(I^{(LT)}+I_{\mathrm{IVR},l}^{(L)})=\lim_{L\to\infty}(I^{(LT)}-\Delta I_1)=I^\infty.
\end{equation}

The finite-volume effect in the IVR method (denoted as $\delta_{\text{IVR}}=I_{\text{IVR}}^{(L)}-I^\infty$) can also be analyzed using Theorem \ref{thm:FV}. We write $I_{\mathrm{IVR}}^{(L)}$ as a discrete summation in momentum space
\begin{equation}
	\begin{aligned}
		I_{\mathrm{IVR}}^{(L)}=&\frac{1}{L^3}\sum_{\bm{k}'\in\Gamma} \int \frac{d^3 k}{(2\pi)^3}\delta_L(\bm{k}'-\bm{k})\frac{A_{N}(-\bm{k}')}{2E_\gamma(\bm{k})2E_{N}(\bm{k}')}\\ &\times\left[ \frac{\left(1-e^{(m_n-E_\gamma(\bm{k})-E_{N}(\bm{k}'))t_s}\right)}{E_\gamma(\bm{k})+E_{N}(\bm{k}')-m_n}-\frac{\left(-e^{(m_n-E_\gamma(\bm{k})-E_{N}(\bm{k}'))t_s}\right)}{E_\gamma(\bm{k})+E_{N}(\bm{k})-m_n-i\epsilon}\right],
	\end{aligned}
\end{equation}
where the two terms in the bracket come from $I_{IVR,s}^{(L)}=I^{(LT)}$ and $I_{IVR,l}^{(L)}$ respectively. Due to the temporal truncation $t_s$, the first term is a nonsingular function and does not cause non-smoothness in the summand. The second term avoids the dependence of $\bm{k}'$ in the denominator, thus avoiding the non-smoothness (cusp effect) of $\bm{k}'$ near the threshold. Therefore, in the case of $m_n>m_N$, the summand in the IVR method is still an infinitely differentiable function, and there is always an exponentially suppressed finite-volume effect.

Fig. (\ref{fig:res1}) shows the relative error of finite-volume effects in both methods with different values of $m_N$. The red line shows relative error of $\Delta I_2=\tilde{I}^{(L)}-I^\infty$, and the blue line shows the relativer error of $\delta_{\text{IVR}}=I_{\text{IVR}}^{(L)}-I^\infty$. In the numerical implementation, $\tilde{I}^{(L)}$ is calculated by $\tilde{I}^{(L)}=I^{(LT)}-\Delta I_1$ as shown in Fig. (\ref{fig:CLT}). $I_{\text{IVR}}^{(L)}$ is calculated by the definition in Eq. (\ref{IVRV}). $I^\infty$ can be determined by a sufficient large volume $\lim_{L\to\infty}I_{\text{IVR}}^{(L)}$. From these figures we can see, in the case of $m_N\geq m_n$ (such as QED self-energy of hadrons), two methods are almost identical and only have exponentially suppressed finite-volume effects. In the case of $m_N<m_n$, $\Delta I_2$ has a slow convergence behavior as volume increases, indicating a power-law finite-volume effect. $\delta_{\text{IVR}}$ converges much faster and only has exponentially suppressed finite-volume effects. 

\begin{figure}[htbp]
	\centering
	\includegraphics[width=1\textwidth]{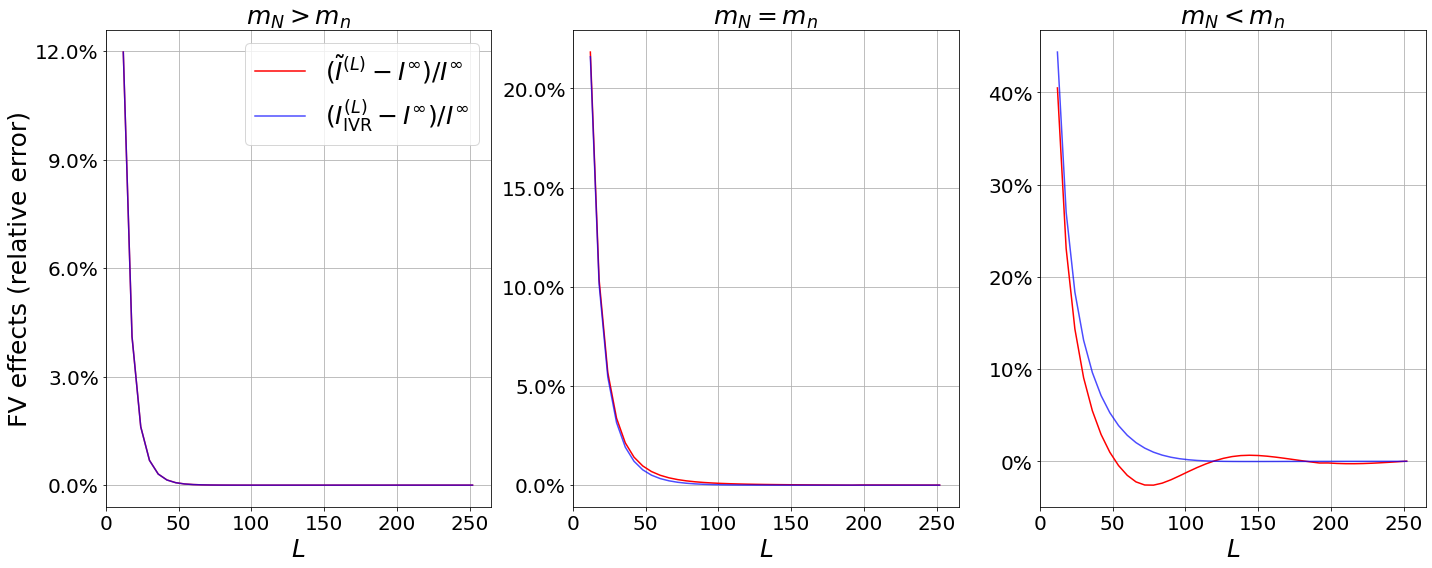}
	\caption{
		Finite-volume effects in two methods with different values of $m_N$.
		The red line shows relative error of $\Delta I_2=\tilde{I}^{(L)}-I^\infty$, and the blue line shows the relativer error of $\delta_{\text{IVR}}=I_{\text{IVR}}^{(L)}-I^\infty$. 
		Here $m_n$ is fixed at 0.14, and $m_N$ is taken as three cases: $m_N=0.2>m_n$, $m_N=0.14=m_n$, and $m_N=0.1<m_n$.
		\label{fig:res1}
	}
\end{figure}
To summarize, when applying the EW$_\infty$ method to the single-particle intermediate state case, we identify two types of systematic error as the temporal truncation effect $\Delta I_1$ and the finite-volume effect $\Delta I_2$. The IVR method modified the form of $\Delta I_1$, avoiding the power-law finite-volume effect due to the cusp effect when $m_N<m_n$. In practical lattice calculations, the exponentially suppressed finite-volume effect in IVR method can either be avoided by increasing the lattice volume, or be corrected by calculating $\delta_{\text{IVR}}=I_{\text{IVR}}^{(L)}-\lim_{L\to\infty}I_{\text{IVR}}^{(L)}$ with input of form factors~\cite{Tuo:2019bue,Tuo:2021ewr}. The conclusion in this section also holds for other physical processes with more complex electroweak weight functions. The IVR technique in this part is limited to the single-particle case and cannot be easily extended to the two-particle case discussed in the next section.

\section{\label{sec:two}Two-particle intermediate state}
In this section, we study the EW$_\infty$ method and its finite-volume effects for the two-particle intermediate state case. We focus on the contributions from time ordering $t<0$ and neglect contribution from other excited states. The contribution of two-particle intermediate state (denoted as $\pi\pi$ for simplicity) in infinite volume is given by
\begin{equation}\label{HM2pi}
	\begin{aligned}
		H^\infty_{t<0}(k,p)=&i\sumint_{\alpha_{\pi\pi}}\frac{A_{\pi\pi}^{\infty}(-\bm{k},E_{\alpha_{\pi\pi}})}{m-k^0-E_{\alpha_{\pi\pi}}+i\epsilon}+\cdots,\\
		A_{\pi\pi}^{\infty}(-\bm{k},E_{\alpha_{\pi\pi}})=&\langle f | J_{2}|\alpha_{\pi\pi}(-\bm{k},E_{\alpha_{\pi\pi}})\rangle  \langle
		\alpha_{\pi\pi}(-\bm{k},E_{\alpha_{\pi\pi}})|J_{1}|i\rangle.
	\end{aligned}
\end{equation}
Here, $\sumintinline_{\alpha_{\pi\pi}}$ denotes the sum over the infinite-volume two-particle states $|\alpha_{\pi\pi}(-\bm{k},E_{\alpha_{\pi\pi}})\rangle$ with momentum $-\bm{k}$. $J_{1}$ and $J_{2}$ represent operators in the physical problem, and we omit their Lorentz indices for simplicity. 

In lattice calculations, we obtain the finite-volume Euclidean hadronic matrix element $H_{E}^{(L)}(x)$, which is also dominated by two-particle states for $\tau<0$ as
\begin{equation}\label{HEx2pi}
	\begin{aligned}
		H_{E}^{(L)}(\tau<0,\bm{x})=&\frac{1}{L^3}\sum_{\bm{k}'\in\Gamma}\sum_{\Lambda,n}A_{\pi\pi,E}^{(L)}(-\bm{k}',E_{\Lambda,n}) e^{(E_{\Lambda,n}-m)\tau+i\bm{k}'\cdot\bm{x}}+\cdots,\\
		A_{\pi\pi,E}^{(L)}(-\bm{k}',E_{\Lambda,n})=&\langle f | J_{2,E}|\pi\pi_L(-\bm{k}',E_{\Lambda,n})\rangle \langle \pi\pi_L(-\bm{k}',E_{\Lambda,n})|J_{1,E}|i\rangle.
	\end{aligned}
\end{equation}
Here, the finite-volume two-particle states are $|\pi\pi_L(-\bm{k}',E_{\Lambda,n})\rangle$, where $-\bm{k}'$ is the total (discrete) momentum and $E_{\Lambda,n}$ is the discrete energy level of $\pi\pi$. $\Lambda=(\Gamma,\alpha,r)$ labels the quantum numbers under the lattice cubic group's irreducible representation, with $\Gamma$ being the representation, $\alpha=1,\cdots,d_{\Gamma}$ being the dimension of $\Gamma$, and $r$ being the number of occurrence of $\Gamma$. $n$ labels the $n$-th state with quantum number $\Lambda$. The normalization condition of the states is taken as
\begin{equation}
	\langle \pi\pi_L(\bm{k}'_1,E_{\Lambda_1,n_1})|\pi\pi_L(\bm{k}'_2,E_{\Lambda_2,n_2})\rangle=\delta_{\bm{k}'_1,\bm{k}'_2}\delta_{\Lambda_1,\Lambda_2}\delta_{n_1,n_2}.
\end{equation}

In this section, we derive the formula for correcting the finite-volume effects caused by two-particle intermediate states. For the first type of correction ($\Delta I_1$), unlike the single-particle case, it has both the power-law finite volume effect and the temporal truncation effects. For the second type of correction ($\Delta I_2$), we show that the discrete momentum summand is still free from infrared singularities, and the cusp effect causes power-law finite-volume effect. We find that this cusp effect is strongly suppressed if $\pi\pi$ particles are in P wave. 

We use $\eta\to\mu^+\mu^-$ decay as an example to introduce our method. $\eta\to\mu^+\mu^-$ can be regarded as the flavor-conserving version of $K_L\to\mu^+\mu^-$. They have the same $\pi\pi$ intermediate state, so the discussion here can be directly applied to $K_L\to\mu^+\mu^-$ decay. We choose $\eta\to\mu^+\mu^-$ decay for the convenience of discussion, since $K_L\to\mu^+\mu^-$ includes other intermediate states such as $\pi^0$ and $\eta$, whose treatment is not the concern of this work and was discussed in Ref.~\cite{Zhao:2022pbs,Wang:2019try}.

\subsection{\label{sec:two-eg}Example: $\eta\to\mu^+\mu^-$ decay amplitude}
\begin{figure}[htbp]
	\centering
	\includegraphics[width=0.48\textwidth]{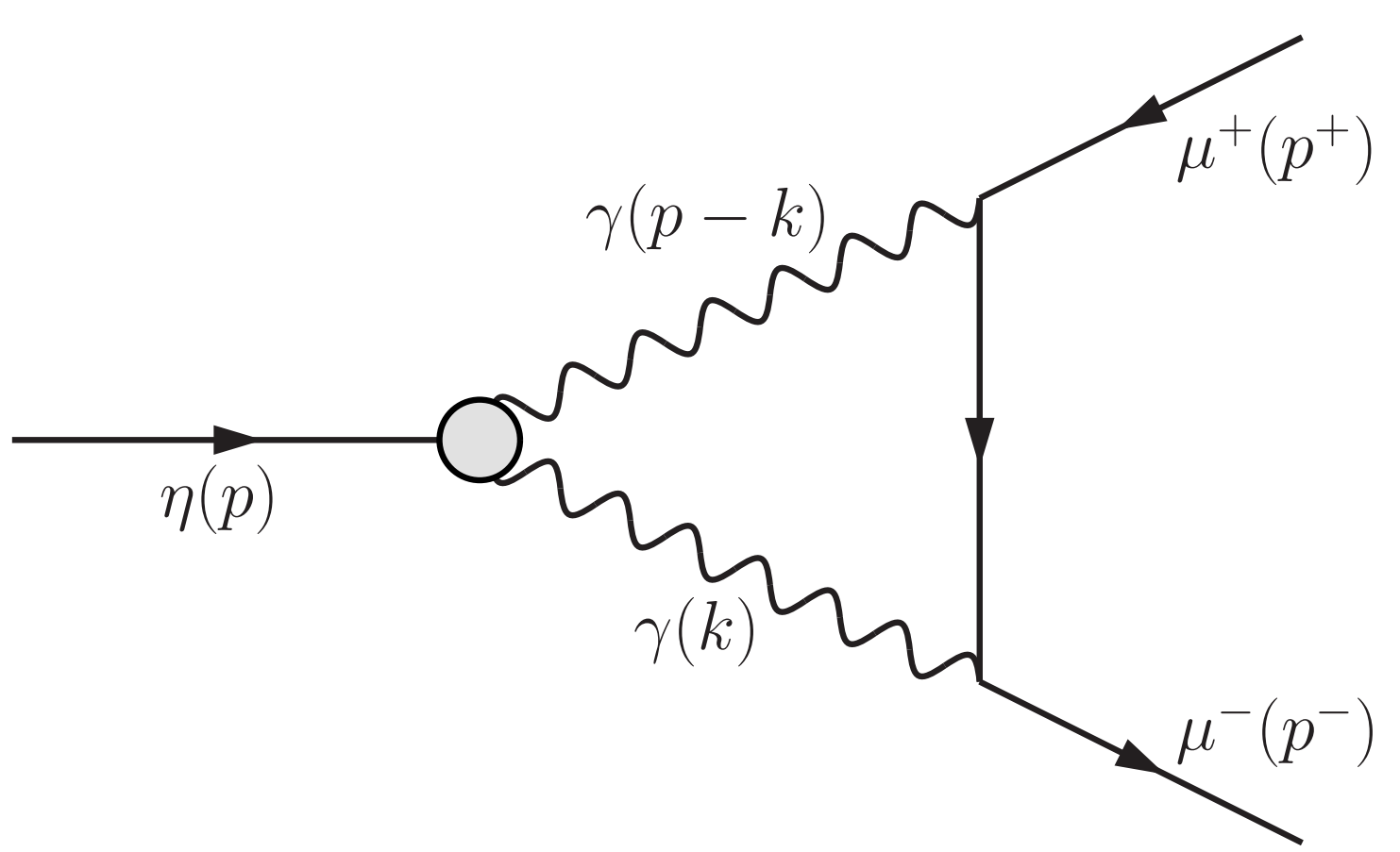}
	\includegraphics[width=0.48\textwidth]{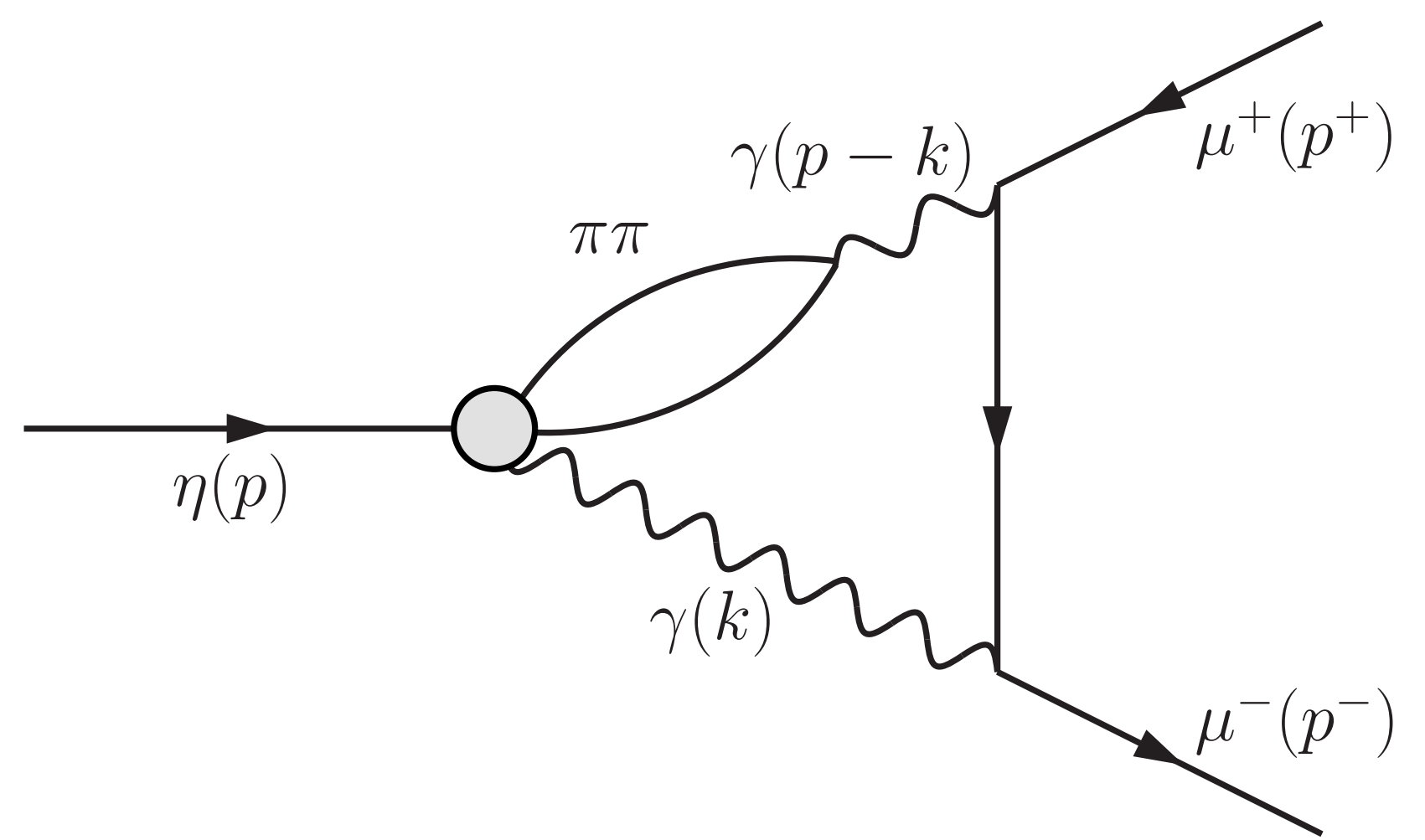}
	\caption{The process of $\eta\to\mu^+\mu^-$ decay and the $\pi\pi$ intermediate states.\label{fig:etall}}
\end{figure}
Fig.~\ref{fig:etall} shows the $\eta\to \mu^+\mu^-$ decay amplitude and the $\pi\pi$ intermediate state in this process. The total input momentum is $p=(m_\eta,\bm{0})$. The two photons have energy-momentum $k$ and $p-k$, respectively, and the final-state muons have energy-momentum $p^\pm=(\frac{m_\eta}{2},\bm{p}^\pm)$, satisfying the on-shell condition $|\bm{p}^{\pm}|=\sqrt{\frac{m_\eta^2}{4}-m_{\mu}^2}$. The decay width of $\eta\to\mu^+\mu^-$ is \cite{Silagadze:2006rt}
\begin{equation}
	\Gamma(\eta\to\mu^+\mu^-)=\frac{\alpha^4}{2\pi}m_\eta m_\mu^2\beta |\mathcal{A}|^2,
\end{equation}
where $\beta=\sqrt{1-\frac{4m_\mu^2}{m_\eta^2}}$. $\mathcal{A}$ is the loop integral of the hadronic matrix element and the electroweak weight function
\begin{equation}
	I^\infty=\mathcal{A}=\int \frac{d^4 k}{(2\pi)^4}L_{\mu\nu}^{\infty}(k)H^{\infty,\mu\nu}(k,p),
\end{equation}
where $L_{\mu\nu}^{\infty}(k)$ includes the photon and muon propagators and the projection of the spinors of the final-state muons on the $^1S_0$ channel \cite{Silagadze:2006rt}:
\begin{equation}\label{Letamumu}
	\begin{aligned}
		L_{\mu\nu}^\infty(k)=
		\frac{16\pi^2}{m_\eta^2}\frac{\varepsilon_{\mu\nu\alpha\beta}k^\alpha p^\beta}{(k^2+i\epsilon)((p-k)^2+i\epsilon)((p^--k)^2-m_\mu^2+i\epsilon)}.
	\end{aligned}
\end{equation}

The hadronic matrix element $H^{\infty,\mu\nu}(k,p)$ is
\begin{equation}\label{Hetamumu}
	H^{\infty,\mu\nu}(k,p)=\int d^4 x e^{ik\cdot x}\langle 0|T[J^\mu_{\text{em}}(x)J^\nu_{\text{em}}(0)]|\eta(p)\rangle=i\varepsilon^{\mu\nu\rho\sigma}p_\rho k_\sigma\mathcal{F}(k^2,(p-k)^2),
\end{equation}
where the currents are $J_{\text{em}}^\mu=(\frac 23 \bar{u}\gamma^\mu u-\frac 13 \bar{d}\gamma^\mu d)_M$, and the initial and final states are $|i\rangle=|\eta\rangle$ and $|f\rangle=|0\rangle$. $H^{\infty,\mu\nu}(k,p)$ can be parameterized by the form factor $\mathcal{F}(k^2,(p-k)^2)$, where $\mathcal{F}(0,0)$ is related to the $\eta\to\gamma\gamma$ decay width $\Gamma(\eta\to\gamma\gamma)=\frac{m_\eta^3\pi\alpha^2}{4}\mathcal{F}^2(0,0)$. Using $\mathcal{F}(k^2,(p-k)^2)$, we can write the amplitude $\mathcal{A}$ as a conventional form in the literature:
\begin{equation}
	\mathcal{A}=\frac{1}{i\pi^2 m_\eta^2}\int d^4 k \frac{2(p^2 k^2-(p\cdot k)^2)}{(k^2+i\epsilon)((p-k)^2+i\epsilon)((p^--k)^2-m_\mu^2+i\epsilon)}\mathcal{F}(k^2,(p-k)^2).
\end{equation}

The right side of Fig.~\ref{fig:etall} shows the $\pi\pi$ intermediate states in this decay process. Since $m_\eta<4m_\pi$ and $4\pi$ intermediate states cannot be on-shell, we ignore the finite-volume effects caused by them. The two-particle state contribution to $H^{\infty,\mu\nu}(k,p)$ for $t<0$ is given by Eq. (\ref{HM2pi}) with specific initial/final states and currents defined here.

In the EW$_\infty$ method, the calculation of the amplitude takes the form
\begin{equation}
	I^{(LT)}=c_{ME}^{\mu\nu}\int_V d^3 x\int_{-t_s}^{t_s} d\tau L_{E}^{\infty,\mu\nu}(\tau,\bm{x})H_{E}^{(L),\mu\nu}(\tau,\bm{x}),
\end{equation}
where the electroweak weight function is
\begin{equation}\label{LEetall}
	\begin{aligned}
		&L_{E}^{\infty,\mu\nu}(\tau,\bm{x})=\frac{16\pi^2}{m_\eta}\int \frac{d^3 k}{(2\pi)^3}\int_{C_{\mathrm{lat}}} \frac{dk_E^0}{2\pi}e^{-ik_E^0 \tau-i\bm{k}\cdot\bm{x}}
		\\ &\times
		\frac{\varepsilon^{\mu\nu\alpha 0}k_E^\alpha}{((k_E^0)^2+E_\gamma(\bm{k})^2-i\epsilon)((p_E^0-k_E^0)^2+E_\gamma(\bm{k})^2-i\epsilon)((\frac{p_E^{0}}{2}-k_E^0)^2+E_\mu(\bm{p}^--\bm{k})^2-i\epsilon)}.
	\end{aligned}
\end{equation}
Here, the photon and muon energies are $E_\gamma(\bm{k})=|\bm{k}|$ and $E_\mu(\bm{k})=\sqrt{\bm{k}^2+m_\mu^2}$. We use the deformed contour $C_{\mathrm{lat}}$ in Fig. \ref{fig:Wick} to avoid the destruction of the Wick rotation caused by singularities crossing the imaginary axis. The numerical implementation of $L_{E}^{\infty,\mu\nu}(\tau,\bm{x})$ is similar to that of $\pi\to e^+e^-$ in Ref.~\cite{Christ:2022rho}. We define $L_E^{\infty,ij}(\tau,\bm{x})=L_E^{\infty}(\tau,\bm{x})\epsilon^{ijk}x_k/|\bm{x}|^2$ with $i,j,k\neq 0$, then $L_E^{\infty}(\tau,\bm{x})$ is given by
\begin{equation}\label{LEcalc}
	\begin{aligned}
		&L_{E}^\infty(\tau,\bm{x})=i\frac{2}{m_\eta^3}\int d|\bm{k}| |\bm{k}||\bm{x}|j_1(|\bm{k}||\bm{x}|)e^{\frac{m_\eta}{2}\tau}\\
		&\times\left\{\frac{1}{\beta}\ln\frac{1+\beta}{1-\beta}\left[-\frac{e^{(\frac{m_\eta}{2}-|\bm{k}|)|\tau|}}{m_\eta-2|\bm{k}|+i\epsilon}+\frac{e^{-(\frac{m_\eta}{2}+|\bm{k}|)|\tau|}}{m_\eta+2|\bm{k}|}\right]+\int d\cos\theta \frac{e^{-E_\mu|\tau|}}{E_\mu[\beta\cos\theta-1][\beta\cos\theta+1]}\right\},
	\end{aligned}
\end{equation}
where we define $j_1(\rho)=(\sin\rho-\rho\cos\rho)/\rho^2$, and $E_\mu=\sqrt{m_\eta^2/4+|\bm{k}|^2-\beta m_\eta |\bm{k}|\cos\theta}$. We derive this form by integrating out $k_E^0$ using Cauthy's theorem and averaging over the direction of the outgoing muon $\bm{p}^-$.

The hadronic matrix element $H_E^{(L),\mu\nu}(\tau,\bm{x})$ is
\begin{equation}\label{HELtwo-model}
	H_E^{(L),\mu\nu}(\tau,\bm{x})=(\langle  0|T[J_{\text{em},E}^\mu(\tau,\bm{x})J_{\text{em},E}^\nu(0)]|\eta\rangle)^{(L)},
\end{equation}
where the Euclidean electromagnetic currents are $J_{\text{em},E}^\mu=(\frac 23\bar{u}\gamma^\mu u-\frac 13 \bar{d}\gamma^\mu d)_E$. The contribution of $\pi\pi$ intermediate states is given by Eq. (\ref{HEx2pi}) with $|i\rangle=|\eta\rangle$, $|f\rangle=|0\rangle$, and $J_{1/2,E}=J_{\text{em},E}$.

\subsection{\label{sec:two-1}First type of correction}
The first type of correction $\Delta I_1$ depends on the difference between the finite-volume matrix element $H^{(LT)}(k',p)$ and the infinite-volume matrix element $H^{\infty}(k',p)$ at discrete momentum $\bm{k}'\in\Gamma$. The formula of this difference is investigated in Ref.~\cite{Briceno:2019opb}. In this part, we follow their work and present a general correction formula for $\Delta I_1$ in the $t<0$ time-ordering. For convenience, we do not write out the Lorentz indices explicitly in this part.

We first introduce the energy-momentum variables of the two-particle system used here. The total energy-momentum of the $\pi\pi$ system is $p_{\pi\pi}=(E_{\pi\pi},-\bm{k}')=(m_\eta-k^0,-\bm{k}')$. The energy-momentum of each particle is defined as $q$ and $q'$. In the center-of-mass frame, the energy-momentum is $p_{\pi\pi}^*=(E_{\pi\pi}^*,\bm{0})=(E_{\pi\pi}/\gamma,\bm{0})$, where $\gamma=1/\sqrt{1-\left(\bm{k}'/E_{\pi\pi}\right)^2}$ is the boost factor. In this frame, one of the single particle momenta and its magnitude are $\bm{q}^*$ and $q^*=|\bm{q}^*|$, and its energy is $E_\pi(\bm{q}^*)=E^*_{\pi\pi}/2=\sqrt{q^{*2}+m_\pi^2}$. The invariant mass square of two particles is $s=(m_\eta-k^0)^2-\bm{k}'^2$.

The difference between $H^{(LT)}(k',p)$ and $H^{\infty}(k',p)$ is defined as
\begin{equation}
	\begin{aligned}
		\Delta H(k',p)&=H^{(LT)}(k',p)-H^{\infty}(k',p),\\
		H^{(LT)}(k',p)=&i\sum_{\Lambda,n}\frac{A_{\pi\pi}^{(L)}(-\bm{k}',E_{\Lambda,n})}{m_\eta-k^0-E_{\Lambda,n}}\left(1-e^{(m_\eta-k^0-E_{\Lambda,n})t_s}\right)+\cdots,\\
		H^{\infty}(k',p)=&i\sumint_{\alpha_{\pi\pi}}\frac{A_{\pi\pi}^{\infty}(-\bm{k}',E_{\alpha_{\pi\pi}})}{m_\eta-k^0-E_{\alpha_{\pi\pi}}+i\epsilon}+\cdots,
	\end{aligned}
\end{equation}
where $A_{\pi\pi}^{(L)}(-\bm{k}',E_{\Lambda,n})$ and $A_{\pi\pi}^{\infty}(-\bm{k}',E_{\alpha_{\pi\pi}})$ are defined in finite volume and infinite volume, respectively (see Eq. (\ref{HM2pi} and Eq. (\ref{HEx2pi})). Unlike the single-particle case, $\Delta H(k',p)$ has both the temporal truncation effect and the power-law finite-volume effect:
\begin{equation}
	\begin{aligned}
		\Delta H(k',p)&=\Delta H^T(k',p)+\Delta H^V(k',p),\\
		\Delta H^T(k',p)&=i\sum_{\Lambda,n} \frac{A_{\pi\pi}^{(L)}(-\bm{k}',E_{\Lambda,n})}{m_\eta-k^0-E_{\Lambda,n}} \left(-e^{(m_\eta-k^0-E_{\Lambda,n})t_s}\right),\\
		\Delta H^V(k',p)&=i\sum_{\Lambda,n}\frac{A_{\pi\pi}^{(L)}(-\bm{k}',E_{\Lambda,n})}{m_\eta-k^0-E_{\Lambda,n}}-i\sumint_{\alpha_{\pi\pi}}\frac{A_{\pi\pi}^{\infty}(-\bm{k}',E_{\alpha_{\pi\pi}})}{m_\eta-k^0-E_{\alpha_{\pi\pi}}+i\epsilon}.
	\end{aligned}
\end{equation}

The power-law finite-volume effect in $H^V(k',p)$ can be derived as~\cite{Briceno:2019opb,Christ:2015pwa}
\begin{equation}\label{DeltaHV}
	\Delta H^V=\Gamma^{\text{in}}\frac{iF^{i\epsilon}}{1+F^{i\epsilon}\mathcal{M}_2}\Gamma^{\text{out}}.
\end{equation}
Next, we explain each term and rewrite this formula using the partial wave expansion. $\Gamma^{\text{in}}$ and $\Gamma^{\text{out}}$ are the physical infinite-volume matrix elements, and their partial wave expansions are
\begin{equation}\label{Gammatilde}
	\begin{aligned}
		\Gamma^{\text{in}}_{lm}&= \langle 0|J_{\text{em}}|\pi\pi^{\text{in}}(-\bm{k}',q^*,l,m)\rangle,\\
		\Gamma^{\text{out}}_{lm}&=\langle \pi\pi^{\text{out}}(-\bm{k}',q^*,l,m)|J_{\text{em}}|\eta\rangle,
	\end{aligned}
\end{equation}
where we define the infinite-volume two-particle states with angular momentum quantum numbers $(l,m)$ in the center-of-mass frame as
\begin{equation}
	\begin{aligned}
		|\pi\pi^{\text{in}}(-\bm{k}',q^*,l,m)\rangle&=\frac{1}{\sqrt{4\pi}}\int d\Omega_{\hat{q}^*}Y_{lm}(\hat{q}^*)|\pi\pi^{\text{in}}(q,q')\rangle,\\
		\langle\pi\pi^{\text{out}}(-\bm{k}',q^*,l,m)|&=\frac{1}{\sqrt{4\pi}}\int d\Omega_{\hat{q}^*}Y^*_{lm}(\hat{q}^*) \langle\pi\pi^{\text{out}}(q,q')|.
	\end{aligned}
\end{equation}
Here, the normalization condition of $\pi\pi$ states is chosen to be
\begin{equation}
	\langle\pi\pi^{\text{in/out}}(q_1,q'_1)|\pi\pi^{\text{in/out}}(q_2,q'_2)\rangle=2E_{\pi}(\bm{q}_1)2E_{\pi}(\bm{q}'_{1})(2\pi)^3\delta^{3}\left(\bm{q}_1-\bm{q}_2\right)(2\pi)^3\delta^{3}\left(\bm{q}'_1-\bm{q}'_2\right).
\end{equation}

For a specific partial wave, the function \( F^{i\epsilon} \) is defined as~\cite{Hansen:2012tf}
\begin{equation}\label{F_modified}
	\begin{aligned} 
		F^{i\epsilon}_{l m; l' m'}(s) &= \frac{\operatorname{Re} q^*}{8 \pi \sqrt{s}} \eta\delta_{l l'} \delta_{m m'}-\frac{i}{2 \sqrt{s}}\eta \sum_{\tilde{l},\tilde{m}} \frac{\sqrt{4 \pi}}{q^{* l}} c_{l m}^P\left(q^{* 2}\right) \int d \Omega Y_{l m}^* Y_{\tilde{l}\tilde{m}}^* Y_{l', m'},
	\end{aligned}
\end{equation}
where \(\eta = 1/2\) is the symmetry factor for the \(\pi\pi\) system. The definition of the function \(c_{l m}^P(q^{*2})\) is given in Ref.~\cite{Kim:2005gf,Hansen:2012tf}. As discussed in Ref.~\cite{Hansen:2012tf}, this definition of \( F^{i\epsilon} \) remains valid even when \( s < (2m_\pi)^2 \) (or equivalently \( q^{*2} < 0 \)), such as the presence of shallow bound states. In this regime, \( q^{*2} - k^{*2}<0 \) always holds, so there is no pole in the denominator. 
As a result, for  \( s < (2m_\pi)^2 \), the first term in Eq.~\eqref{F_modified}, representing the difference between the principal value and the \( i\epsilon \) prescriptions, drops out, while the second term yields only exponentially suppressed finite-volume effects. 
If these exponentially suppressed effects are neglected, the expression for \( F^{i\epsilon} \) simplifies to
\begin{equation}
	\begin{aligned}
		F_{lm;l'm'}^{i\epsilon}(s)&=\frac{q^*}{16\pi\sqrt{s}}\theta(s-(2m_\pi)^2)\left(\cot\left(\phi_{l'm';lm}\left(q^*\right)\right)+i\right),\\
		\cot\left(\phi_{l'm';lm}\left(q^*\right)\right)&=\frac{4\pi}{q^*}\sum_{\tilde{l},\tilde{m}}\frac{\sqrt{4\pi}}{q^{*l}}c_{lm}^{P}(q^{*2})\int d\Omega_{\hat{q}^*}Y_{lm}^* Y_{\tilde{l}\tilde{m}}^*Y_{l'm'}.
	\end{aligned}
\end{equation}
Here, we define a function $\phi_{l'm';lm}\left(q^*\right)$ to simplify the final expressions of $\Delta H^V(k',p)$. The symbol \( i \) is a shorthand for \( i \times I \), where \( I \) is the identity operator in partial wave space, i.e., \( I_{\ell m; \ell' m'} = \delta_{\ell \ell'} \delta_{m m'} \).

$\mathcal{M}_2$ is the two-particle scattering amplitude, which can be parameterized by the scattering phase shift $\delta_l(s)$ as
\begin{equation}\label{M2}
	\mathcal{M}_{2,lm;l'm'}(s)=\delta_{ll'}\delta_{mm'}\frac{16\pi\sqrt{s}}{q^*}\frac{\exp[2i\delta_l(s)]-1}{2i}.
\end{equation}

Using the partial wave expansion, we can write $\Delta H^V$ in an explicit form~\cite{Christ:2015pwa}
\begin{equation}\label{dHV}
	\begin{aligned}
		\Delta H^V(k',p)=&\sum_{l,l',m,m'}i\theta(s-(2m_\pi)^2)\frac{q^*}{16\pi\sqrt{s}} \left(\cot(\phi_{lm;l'm'}(q^*)+\delta_l(s))+i\right)\\ &\times A^\infty_{\pi\pi,lm;l'm'}(-\bm{k}',E_{\pi\pi})\big|_{E_{\pi\pi}=m_\eta-k^0},
	\end{aligned}
\end{equation}
where the partial wave expansion of the infinite-volume matrix element is
\begin{equation}\label{Alm}
	A^{\infty}_{\pi\pi,lm;l'm'}(-\bm{k}',E_{\pi\pi})=\langle 0 | J_{em}|\pi\pi^{\mathrm{in/out}}(-\bm{k}',q^*,l,m)\rangle \langle \pi\pi^{\mathrm{in/out}}(-\bm{k}',q^*,l',m')|J_{em}|\eta\rangle.
\end{equation}

Finally, the correction formula for $\Delta I_1$ in the two-particle case is given by
\begin{equation}\label{dI1}
	\begin{aligned}
		\Delta I_1&=\frac{1}{L^3}\sum_{\bm{k}'\in\Gamma} \int \frac{d^3 k}{(2\pi)^3}\delta_L(\bm{k}'-\bm{k})\int_{C} \frac{d k^0}{2\pi} L^\infty(k)\Delta H(k',p),\\
		\Delta H(k',p)&=\Delta H^T(k',p)+\Delta H^V(k',p)\\
		&=i\sum_{\Lambda,n} \frac{A_{\pi\pi}^{(L)}(-\bm{k}',E_{\Lambda,n})}{m_\eta-k^0-E_{\Lambda,n}} \left(-e^{(m_\eta-k^0-E_{\Lambda,n})t_s}\right)\\&+\sum_{l,l',m,m'}i\theta(s-(2m_\pi)^2)\frac{q^*}{16\pi\sqrt{s}}\left(\cot(\phi_{lm;lm'}(q^*)+\delta_l(s))+i\right)\\ &\times A^{\infty}_{\pi\pi,lm;l'm'}(-\bm{k}',E_{\pi\pi})|_{E_{\pi\pi}=m_\eta-k^0}.
	\end{aligned}
\end{equation}

Ref.~\cite{Briceno:2019opb} demonstrates that both $\Delta H^T$ and $\Delta H^V$ have hadronic singularity at $m_\eta-k^0-E_{\Lambda,n}=0$, but these singularities cancel out with each other. Therefore, $\Delta H(k',p)$ has no singularity. In practical calculation of $\Delta I_1$, we need to first perform the integral $\int_{C} dk^0$, and extract the residues from the singularities in $L^\infty(k)$. The inputs of $\Delta I_1$ are the finite-volume matrix element $A_{\pi\pi}^{(L)}$, the infinite-volume matrix element $A^{\infty}_{\pi\pi, lm;lm'}$, and the phase shift $\delta_l(s)$. Ignoring exponentially suppressed finite-volume effects, $A_{\pi\pi}^{(L)}$ and $A^{\infty}_{\pi\pi, lm;lm'}$ are related by the Lellouch-Lüscher formula. These inputs can be obtained either by lattice QCD calculations or by experimental results. In calculation of $\Delta H^T$, we only need to consider the states with $E_{\Lambda,n}<m_\eta$ that cause exponential divergence with large $t_s$. In the calculation of $\Delta H^V$, the condition $\theta(s-(2m_\pi)^2)$ also involves only a few momentum modes. 

In $\eta \to \mu^+\mu^-$ decay, there could be a suppression of $\Delta I_1$ because $\pi\pi$ states are in the P wave. Hadronic matrix elements related to P-wave $\pi\pi$ states are strongly enhanced around the $\rho$ resonance region, where $s \sim m_\rho^2$. Since the light $\pi\pi$ states that cause the exponentially growing term and power-law finite-volume effects in $\Delta I_1$ must satisfy $E_{\pi\pi} < m_\eta$, the contribution from these light states could be only a small part of the total decay amplitude. Similar suppression behavior also appears in $K_L\to\mu^+\mu^-$ decays, and the contribution of such light $\pi\pi$ intemediate states is estimated to be at most a few percent~\cite{Chao:2024vvl}.

\subsection{\label{sec:two-2}Second type of correction}
According to Theorem \ref{thm:singularity}, we analyze the finite-volume effect $\Delta I_2$ by studying the smoothness of the summand $\hat{I}(\bm{k}')$, which is defined in
\begin{equation}\label{Ihat2pt}
	\begin{aligned}
		\tilde{I}^{(L)}&=\frac{1}{L^3}\sum_{\bm{k}'\in\Gamma} \hat{I}(\bm{k}'),\\
		\hat{I}(\bm{k}')&=\int \frac{d^3 k}{(2\pi)^3}\delta_L(\bm{k}'-\bm{k})\int_{C} \frac{d k^0}{2\pi} L^\infty(k)H^\infty(k',p).
	\end{aligned}
\end{equation}
In the two-particle case, $\Delta I_2$ behaves similarly as the single-particle case discussed earlier. Specifically, due to the integral form of $\hat{I}(\bm{k}')$, the singularities in the electroweak propagators do not directly induce singularities in the summand itself. However, since $2m_\pi<m_\eta$, $\tilde{I}^{(L)}$ still have power-law finite-volume correction from the cusp effect. 

This singularity-free property needs to be proved according to specific problems, and we do so for the example of $\eta\to\mu^+\mu^-$. Next, we derive a formula for the power-law finite-volume correction caused by the cusp effect. If $\pi\pi$ particles are in high partial waves, such as P wave $\pi\pi$ states in $\eta \to \mu^+ \mu^-$ decay, we show that the cusp effect is strongly suppressed and does not appear at the leading order of $O(1/L)$. A numerical test in $\eta \to \mu^+ \mu^-$ process shows that $\Delta I_2$ is dominated by the exponentially suppressed finite-volume effects, and the cusp effect can be safely neglected.

\subsubsection{Proof of no singularity}
Similar to the single-particle case, we need to investigate whether various possible on-shell intermediate states induce singularities in $\hat{I}(\bm{k}')$. As shown in Fig.~\ref{fig:etallcut}, we illustrate the most singular case where both $\pi\pi$ and one photon are on-shell.
\begin{figure}[htbp]
	\centering
	\includegraphics[width=0.6\textwidth]{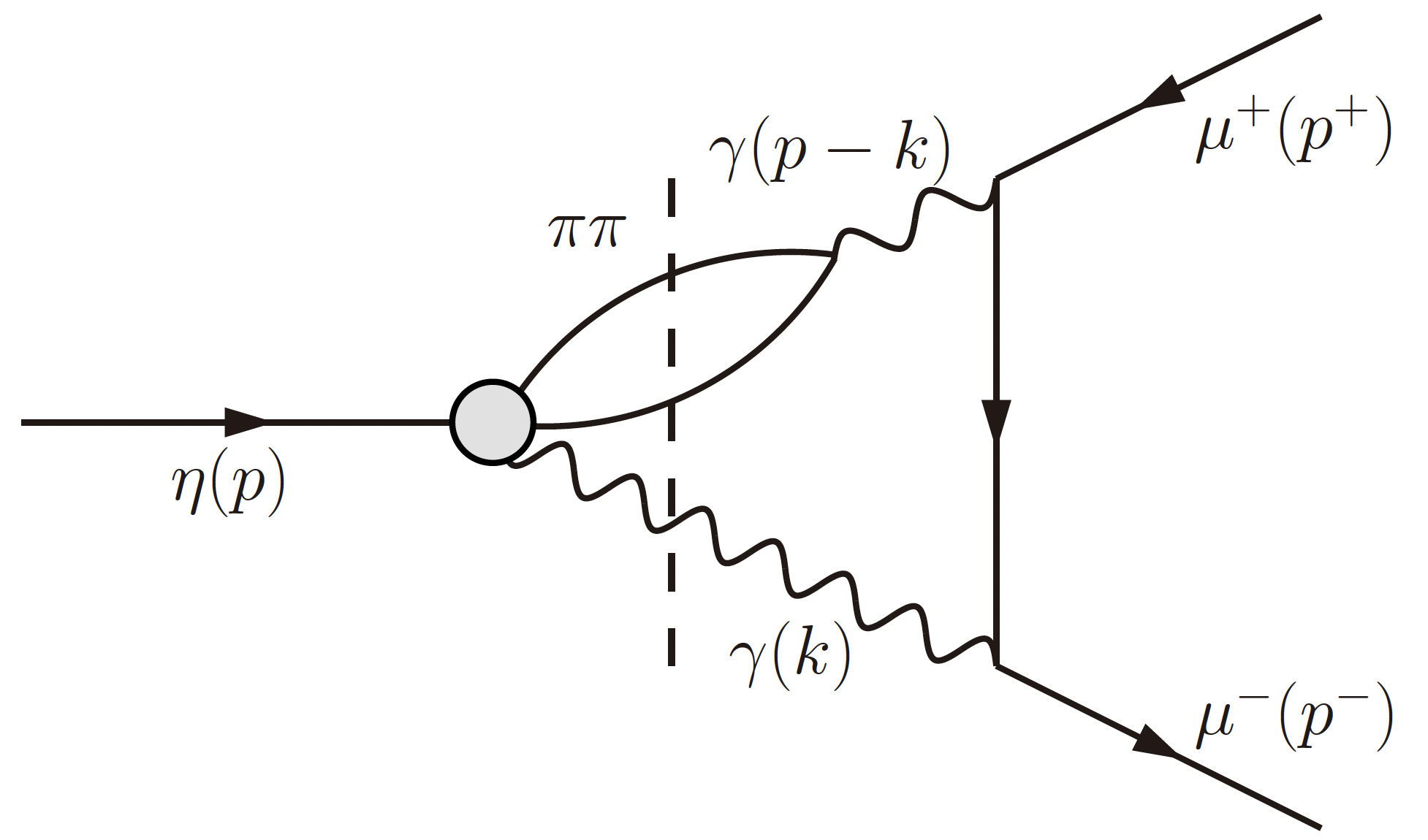}
	\caption{Analysis of the singularity structure of $\hat{I}(\bm{k}')$ by the contribution of on-shell intermediate states.\label{fig:etallcut}}
\end{figure}

This on-shell contribution has the integral structure as
\begin{equation}
	\begin{aligned}
		\hat{I}(\bm{k}')&\sim \int \frac{d^3 k}{(2\pi)^3}\delta_L(\bm{k}'-\bm{k})\bm{k}\cdot\bm{k}' \\ &\times\frac{1}{2E_\gamma(\bm{k})\left(\left(\frac{m_\eta}{2}-E_\gamma(\bm{k})\right)^2-E_\mu(\bm{p}^--\bm{k})^2+i\epsilon\right)\left(\left(m_\eta-E_\gamma(\bm{k})\right)^2-E_\gamma(\bm{k})^2+i\epsilon\right)}\\
		&\times\int\frac{d^3 q}{(2\pi)^3}\frac{1}{2E_\pi(\bm{q})2E_\pi(\bm{k}'+\bm{q})\left(m_\eta-E_\gamma(\bm{k})-E_{\pi}(\bm{q})-E_{\pi}(\bm{k}'+\bm{q})+i\epsilon\right)},
	\end{aligned}
\end{equation}
where $\bm{q}$ is the momentum of one of the pions. $E_\gamma(\bm{k})=|\bm{k}|$ and $E_\mu(\bm{p}^--\bm{k})=\sqrt{(\bm{p}^--\bm{k})^2+m_\mu^2}$ denote the energy of the photon and muon. The term $\bm{k}\cdot \bm{k}'$ comes from contracting $\epsilon_{\mu\nu\alpha\beta}k^\alpha p^\beta$ in the electroweak part and $\epsilon^{\mu\nu\rho\sigma}k'_\rho p_\sigma$ in the hadronic part (see Eq. (\ref{Letamumu}) and Eq. (\ref{Hetamumu})). Here, we only show the main part related to the singularity structure and omit other parts of the integrand.

The poles in each denominator occur when $|\bm{k}| = 0$ (the first photon), $|\bm{k}| = 0$ (the muon), $|\bm{k}| = \frac{m_\eta}{2}$ (the second photon), and $m_\eta = E_\gamma(\bm{k}) + E_{\pi}(\bm{q}) + E_{\pi}(\bm{k}' + \bm{q})$ (the $\pi\pi$ states). When $m_\eta < 4m_\pi$, the condition for a pinch singularity or endpoint singularity is that both $|\bm{k}| = 0$ and $m_\eta = E_{\pi}(\bm{q}) + E_{\pi}(\bm{k}' + \bm{q})$ are satisfied. Under this condition, the propagators of the first photon, muon, and $\pi\pi$ are singular simultaneously. The integral near this point is
\begin{equation}
	\int \frac{d^3 k}{(2\pi)^3}\delta_L(\bm{k}' - \bm{k})\frac{\bm{k} \cdot \bm{k}'}{2|\bm{k}| \cdot |\bm{k}| \cdot m_\eta}\int\frac{d^3 q}{(2\pi)^3}\frac{1}{2E_\pi(\bm{q}) 2E_\pi(\bm{k}' + \bm{q}) |\bm{k}|} \sim \int d^3 k \frac{1}{|\bm{k}|^2}.
\end{equation}
Here, we see that the integral structure protects $\hat{I}(\bm{k}')$ from infrared singularity. This singularity-free property of $\hat{I}(\bm{k}')$ needs to be proven for specific physical problems. For example, the decay $\eta' \to \mu^+ \mu^-$, where $m_{\eta'} > 4m_\pi$, involves more complex on-shell $4\pi$ intermediate states and no longer has this singularity-free property. Such case is beyond the scope of our paper.

\subsubsection{Cusp effect\label{sec:cusp-two}}
Similar to the single-particle case, $\hat{I}(\bm{k}')$ has the cusp effect near the on-shell threshold $s = (m_\eta - k^0)^2 - (\bm{k}')^2 = (2m_\pi)^2$, which leads to the power-law finite-volume effect. This effect, originating from crossing the two-particle on-shell threshold, is also encountered in the three-particle quantization condition, and is avoided by defining physical quantities in a modified $\mathrm{PV}$ prescription (denoted as the $\widetilde{\mathrm{PV}}$ prescription)~\cite{Hansen:2014eka}. In our work, this cusp effect cannot be avoided since the decay problem involves physical matrix element, instead of the modified one in $\widetilde{\mathrm{PV}}$ prescription. In this part, we first review the cusp effect with only one $\pi\pi$ loop~\cite{Hansen:2014eka}, i.e., ignoring the rescattering effect. Then the correction formula for the general case with the rescattering effect is given. This formula can be derived using the Bethe-Salpeter (BS) kernel expansion, which is detailed in Appendix~\ref{sec:Append cusp corr}.

The one $\pi\pi$ loop contribution to the hadronic function, denoted as $H^{\infty,(0)}(k',p)$, is given by
\begin{equation}\label{H0}
	H^{\infty,(0)}(k',p) = \frac{i}{2} \int \frac{d^3 q}{(2\pi)^3} \frac{h(s) \Gamma^{\text{in}}(\bm{q}, -\bm{k}' - \bm{q}) \Gamma^{\text{out}}(\bm{q}, -\bm{k}' - \bm{q})}{2E_{\pi}(\bm{q}) 2E_{\pi}(\bm{k}' + \bm{q}) \left(m_\eta - k^0 - E_{\pi}(\bm{q}) - E_{\pi}(\bm{k}' + \bm{q}) + i\epsilon\right)} + \cdots,
\end{equation}
where $\Gamma^{\text{in}}(\bm{q}, -\bm{k}' - \bm{q})$ and $\Gamma^{\text{out}}(\bm{q}, -\bm{k}' - \bm{q})$ are the physical infinite-volume matrix elements defined in Eq.~(\ref{Gammatilde}). For simplicity, here we omit the indices $lm$. The momenta of two pions are $\bm{q}$ and $-\bm{k}' - \bm{q}$. We extract the two-particle singularity in the denominator $\left(m_\eta - k^0 - E_{\pi}(\bm{q}) - E_{\pi}(\bm{k}' + \bm{q}) + i\epsilon\right)$ and omit other nonsingular contributions. $h(s)$ is an ultraviolet cutoff function that is infinitely differentiable, spherically symmetric, and satisfies
\begin{equation}\label{hs}
	h(s) = \begin{cases}
		0 &\quad s < 0;\\
		1 &\quad s > (2m_\pi)^2.
	\end{cases}
\end{equation}
In range $0<s<(2m_\pi)^2$, $h(s)$ smoothly interpolates between 0 and 1. One example of $h(s)$ is given in Ref.~\cite{Hansen:2014eka} (denoted as $H(\vec{k})=J(s/(2m_\pi)^2)$ there).

We can rewrite the above formula in the center-of-mass frame of $\pi\pi$ as
\begin{equation}
	H^{\infty,(0)}(k',p)=\frac i2\int\frac{q^* d(q^*)^2}{(2\pi)^3} \frac{h(s)}{(q_0^*)^2-(q^*)^2+i\epsilon}\left(\frac{m_\eta-k^0+E_{\pi\pi}}{8E_{\pi\pi}E_{\pi\pi}^*}\right)\int d\Omega_{\hat{q}^*}\Gamma^{\text{in}}\Gamma^{\text{out}},
\end{equation}
where $E_{\pi\pi}$ and $E_{\pi\pi}^*$ are the $\pi\pi$ energy in the rest frame of $\eta$ and in the center-of-mass frame of $\pi\pi$, respectively. $q^*$ is the momentum of one pion in the center-of-mass frame of $\pi\pi$. If two pions are on-shell, $q^*$ has the value $q_0^*=\sqrt{s/4-m_\pi^2}$. Letting $w=(q^*)^2$ and $z=(q_0^*)^2$, the integral structure in the above formula is
\begin{equation}
	\begin{aligned}
		f^{i\epsilon}(z)=\int_0^\infty\frac{\sqrt{w}g(w,z)}{z-w+i\epsilon}dw=\left\{\begin{aligned}
			&\mathcal{P} \int_0^\infty\frac{\sqrt{w}g(w,z)}{z-w}dw-i\pi\sqrt{z}g(z,z),&\quad z> 0,\\
			&\int_0^\infty\frac{\sqrt{w}g(w,z)}{z-w}dw,&\quad z<0,
		\end{aligned}\right.
	\end{aligned}
\end{equation}
where $g(w,z)$ is a nonsingular function and $z=0$ corresponds to the threshold $|\bm{k}'|=\sqrt{(m_\eta-k^0)^2-4m_\pi^2}$. This formula is very similar to Eq.~(\ref{fzsingle}) in the single-particle case, demonstrating that the real and imaginary parts of this integral are not smooth near the threshold $z=0$, i.e., the cusp effect. To avoid this non-smoothness, the $\widetilde{\mathrm{PV}}$ prescription is defined in Ref.~\cite{Hansen:2014eka} as
\begin{equation}\label{cuspf}
	\begin{aligned}
		f^{\widetilde{\mathrm{PV}}}(z)&=f^{i\epsilon}(z)-\left\{\begin{aligned}
			&-i\pi\sqrt{z}g(z,z),&\quad z> 0,\\
			&\pi\sqrt{-z}g(z,z),&\quad z<0.
		\end{aligned}\right.
	\end{aligned}
\end{equation}
For $z>0$, the subtracted term is the imaginary part of $f^{i\epsilon}(z)$, and the $\widetilde{\mathrm{PV}}$ prescription is equivalent to the principal value integral. For $z<0$, the definition is modified by analytically continuing the real part from $z>0$ to $z<0$, thereby restoring the smoothness of the real part.

The non-smooth part of the hadronic function can be determined by the difference between the $i\epsilon$ and $\widetilde{\mathrm{PV}}$ prescriptions as
\begin{equation}
\begin{aligned}
	H^{\text{cusp},(0)}(k',p)&=\widetilde{H}^{\infty,(0)}(k',p)-H^{\infty,(0)}(k',p),\\
	\widetilde{H}^{\infty,(0)}(k',p)&=\frac i2\widetilde{\mathrm{PV}}\int\frac{d^3 q}{(2\pi)^3}\frac{h(s)\Gamma^{\text{in}}(\bm{q},-\bm{k}'-\bm{q})\Gamma^{\text{out}}(\bm{q},-\bm{k}'-\bm{q})}{2E_{\pi}(\bm{q})2E_{\pi}(\bm{k}'+\bm{q})\left(m_\eta-k^0-E_{\pi}(\bm{q})-E_{\pi}(\bm{k}'+\bm{q})\right)}+\cdots.
\end{aligned}
\end{equation}

Using partial-wave expansion, we can simplify $H^{\text{cusp},(0)}(k',p)$ to a form similar to Eq.~(\ref{dHV}):
\begin{equation}\label{Hcusp_0}
	\begin{aligned}
		H^{\text{cusp},(0)}(k',p)=&\sum_{l,l',m,m'}i\frac{h(s)}{16\pi\sqrt{s}}A^\infty_{\pi\pi,lm;l'm'}(-\bm{k}',E_{\pi\pi})|_{E_{\pi\pi}=m_\eta-k^0}\\ &\times \begin {cases}-iq^*\quad &s> (2m_\pi)^2,\\
		|q^*|\quad &0< s \leq (2m_\pi)^2.\\
		\end {cases}
	\end {aligned}
\end {equation}
In the region where $s > (2m_\pi)^2$, this non-smooth part is just the absorptive part of the hadronic matrix elements, which also appears in $\Delta H^V(k',p)$ (the term $i$ in the bracket in Eq.~(\ref{dHV})). In the region where $0 < s \leq (2m_\pi)^2$, $q^* = i|q^*| = i\sqrt{m_\pi^2 - \frac{s}{4}}$ is extended to be an imaginary variable, and the matrix element $A^\infty_{\pi\pi,lm;l'm'}$ is extended to be a function of $q^* = i|q^*|$.

The general case with rescattering effect can be derived by the BS kernal expansion, which is detailed in Appendix~\ref{sec:Append cusp corr} as
\begin{equation}\label{Hcusp_all}
	\begin{aligned}
		H^{\text{cusp}}(k',p)=&\sum_{l,l',m,m'}i\frac{h(s)}{16\pi\sqrt{s}}A^\infty_{\pi\pi,lm;l'm'}(-\bm{k}',E_{\pi\pi})|_{E_{\pi\pi}=m_\eta-k^0}\\ &\times \begin {cases}q^*
		(\tan(\delta_l)-i)\quad &s> (2m_\pi)^2,\\
		|q^*|(\tanh(\sigma_l)+1)\quad &0< s \leq (2m_\pi)^2.\\
		\end {cases}
		\end {aligned}
\end {equation}
Comparing Eq.~(\ref{Hcusp_0}) and Eq.~(\ref{Hcusp_all}), we see that the rescattering effect introduces terms related to the scattering phase shift. In the region where $0< s\leq (2m_\pi)^2$, the phase shift $\delta_l$ is extended to an imaginary value, $\delta_l(s)=i\sigma_l(s)$. Since around the threshold $s\sim (2m_\pi)^2$, the phase shift is close to zero, the rescattering effect only introduces a small modification to $H^{\text{cusp},(0)}(k',p)$.

Next, we consider the cusp effect in $\Delta I_2$. Since $\hat{I}(\bm{k}')$ has no singularity, the power-law finite-volume effects can only arise from the cusp effect. This cusp effect can occur if $2m_\pi$ is lighter than $m_\eta - k_0$, which is the input energy to the hadronic part. The finite-volume correction $\Delta I_2$ is then given by
\begin{equation}\label{Icusp}
	\begin{aligned}
		\Delta I_2 &= \Delta I_{2,\text{cusp}} + \Delta I_{2,\text{exp}},\\
		\Delta I_{2,\text{cusp}} &= \frac{1}{L^3} \sum_{\bm{k}'\in\Gamma} \int \frac{d^3 k}{(2\pi)^3} \delta_L(\bm{k}' - \bm{k}) \int_C \frac{d k^0}{2\pi} \theta(m_\eta - k^0 - 2m_\pi) L^\infty(k) H^{\text{cusp}}(k',p) \\
		&\quad - \int \frac{d^3 k}{(2\pi)^3} \int_C \frac{d k^0}{2\pi} \theta(m_\eta - k^0 - 2m_\pi) L^\infty(k) H^{\text{cusp}}(k,p),
	\end{aligned}
\end{equation}
where $\Delta I_{2,\text{cusp}}$ is the power-law finite-volume effect, and $\Delta I_{2,\text{exp}}$ is the residual exponentially suppressed finite-volume effect.

As mentioned earlier, in the $s > (2m_\pi)^2$ region, the absorptive part of the hadronic function appears in both $\Delta I_1$ (the term $-iq^*$ in Eq.~(\ref{Hcusp_all})) and $\Delta I_2$ (the term $i$ in the bracket in Eq.~(\ref{dHV})). If we add $\Delta I_1$ and $\Delta I_2$ together, this absorptive part will cancel out in the finite-volume integration of $\frac{1}{L^3}\sum_{\bm{k}'\in\Gamma} \int \frac{d^3 k}{(2\pi)^3} \delta_L(\bm{k}' - \bm{k})$, leaving only the contribution in the infinite-volume integration of $\int \frac{d^3 k}{(2\pi)^3}$. Thus, this absorptive part will only appear as an infinite-volume correction.

\subsubsection{Suppression of cusp effect in P-wave}
The correction formula for $\Delta I_2$ in Eq.~(\ref{Icusp}) applies to general physical problems if the summand $\hat{I}(\bm{k}')$ has no singularity. In this section, we consider the explicit form of $A^{\infty,\mu\nu}_{\pi\pi,lm;l'm'}(-\bm{k}',E_{\pi\pi})$ in the $\eta \to \mu^+\mu^-$ decay and show that the cusp effect is strongly suppressed if the $\pi\pi$ particles are in the P-wave.

In $\eta\to\mu^+\mu^-$ decay, the relevant infinite-volume matrix elements are
\begin{equation}\label{FB}
	\begin{aligned}
		\langle 0|J^\nu_{em}|\pi\pi(q,q')\rangle&=i(q-q')^\nu F^{(\pi)}(s)\\
		\langle \pi\pi(q,q')|J^\mu_{em}|\eta(p)\rangle&=i\varepsilon^{\mu\rho\sigma\lambda}p_\rho q_\sigma q'_\lambda B_\eta(s,s_\gamma)
	\end{aligned}
\end{equation}
where $F^{(\pi)}(s)$ is the timelike form factor of the pion, and $B_\eta(s,s_\gamma)$ is the form factor describing the $\eta-\pi\pi\gamma$ coupling~\cite{Venugopal:1998fq}. The invariant mass square of $\pi\pi$ and $\gamma^{(*)}$ are defined as $s=(q+q')^2$ and $s_\gamma=(q-p-p')^2$. In this decay, $\pi\pi$ has the same quantum number as the virtual photon, so they are in the P-wave. Substituting Eq.~(\ref{FB}) into Eq.~(\ref{Alm}) gives
\begin{equation}
	\begin{aligned}
	\sum_{l,l',m,m'}A^{\infty,\mu\nu}_{\pi\pi,lm;l'm'}(-\bm{k}',E_{\pi\pi})=&\sum_{m,m'}A^{\infty,\mu\nu}_{\pi\pi,1m;1m'}(-\bm{k}',E_{\pi\pi})\\
	=&\frac{2 }{3}q^{*2}F^{(\pi)}(s)B_\eta(s,s_\gamma)\varepsilon^{\mu\nu\rho\sigma}p_\rho k'_\sigma
	\end{aligned}
\end{equation}
where $k'=(E_{\pi\pi},-\bm{k}')$ is the energy-momentum of $\pi\pi$. Using this explicit form, $H^{\text{cusp},\mu\nu}(k',p)$ can be simplified as
\begin{equation}\label{Hcusp_eg}
	\begin{aligned}
		H^{\text{cusp},\mu\nu}(k',p)=&i\frac{\varepsilon^{\mu\nu\rho\sigma}p_\rho k'_\sigma}{24\pi\sqrt{s}}F^{(\pi)}(s)B_\eta(s,s_\gamma)h(s)|_{E_{\pi\pi}=m_\eta-k^0}\\ &\times\begin{cases}(q^*)^3
			\left(\tan(\delta_l)-i\right)\quad &s> (2m_\pi)^2\\
			-|q^*|^3\left(\tanh(\sigma_l)+1\right)\quad &0< s\leq (2m_\pi)^2\\
		\end{cases}
	\end{aligned}
\end{equation}
Then $\Delta I_{2,\text{cusp}}$ can be given by Eq.~(\ref{Icusp}). Here, the cusp effect is caused by $(q^*)^3=[(m_\eta-k_0)^2-4m_\pi^2-(\bm{k}')^2]^{3/2}/2$ instead of $q^*=\sqrt{(m_\eta-k_0)-4m_\pi^2-(\bm{k}')^2}/2$, so $H^{\text{cusp},\mu\nu}(k',p)$ is not only a continuous function but also has a continuous first derivative with respect to $\bm{k}'$, and the discontinuity only appears in the second derivative. Therefore, in the $\eta\to\mu^+\mu^-$ decay, the cusp effect is suppressed as $O(1/L^2)$.

To explain the physical origin of this suppression, we notice that the matrix elements involved here are very similar to the decay amplitude of a $1 \to 2$ decay. It is well known that the decay amplitude of a two-body decay with angular momentum $l$ is proportional to $(q^*)^l$, leading to the suppression of the decay amplitude in the high partial waves. Here, the angular momentum projection of $\langle 0 | J_{\text{em}}^{\nu} | \pi\pi(q,q') \rangle$ or $\langle \pi\pi(q,q') | J_{\text{em}}^\mu | \eta \rangle$ is also proportional to $(q^*)^l$ (as demonstrated in Appendix \ref{sec:Append cusp}), so the cusp effect of partial wave $l$ is given by $|q^*|^{2l+1}$, leading to an $O(1/L^{l+1})$ power-law finite-volume effect. Thus, the cusp effect is strongly suppressed in higher partial waves.

In the P-wave case, the $\rho$ resonance enhancement is another reason for the suppression of the cusp effect, which also appears in $\Delta I_1$ as discussed earlier. Compared to the resonance region $s \sim m_\rho^2$, the non-smooth region $s \sim (2m_\pi)^2$ only contributes a relatively small part to the total decay width. On the other hand, the exponentially suppressed part of $\Delta I_2$ is dominated by this $\rho$ resonance region and, therefore, is not suppressed.

Finally, we numerically verify the P-wave suppression of $\Delta I_{2,\text{cusp}}$ in the $\eta \to \mu^+\mu^-$ decay. For convenience of numerical implementation, we only consider the one $\pi\pi$ loop case and ignore the rescattering effect. The hadronic function $H^{(0),\mu\nu}(k',p)$ is given in Eq.~(\ref{H0}), and the non-smooth part $H^{\text{cusp},(0),\mu\nu}(k',p)$ is given by
\begin{equation}\label{Hcusp_eg_LO}
	\begin{aligned}
		H^{\text{cusp},(0),\mu\nu}(k',p) =& i \frac{\varepsilon^{\mu\nu\rho\sigma} p_\rho k'_\sigma}{24\pi \sqrt{s}} F^{(\pi)}(s) B_\eta(s,s_\gamma) h(s) \bigg|_{E_{\pi\pi}=m_\eta - k^0} \\
		&\times \begin{cases}
			-i(q^*)^3 & s > (2m_\pi)^2, \\
			-|q^*|^3 & 0 < s \leq (2m_\pi)^2,
		\end{cases}
	\end{aligned}
\end{equation}

In this simplified case, we can calculate the finite-volume correction $\Delta I_2 = \Delta I_{2,\text{cusp}} + \Delta I_{2,\text{exp}}$ numerically. We present the details of the numerical implementation in Appendix~\ref{sec:Append eta_test}. Fig.~\ref{fig:res2} shows the volume suppression behaviors of $\Delta I_{2,\text{cusp}}$ and $\Delta I_{2,\text{exp}}$. At typical lattice volumes, the exponentially suppressed part $\Delta I_{2,\text{exp}}$ (red line) is the dominant finite-volume effect in $\Delta I_2$, causing a relative error of a few percent. To reduce this error to below one percent, a volume larger than approximately 5 fm (for the real part) and approximately 4.5 fm (for the imaginary part) is needed, which current lattice calculations could achieve. The power-law suppressed part $\Delta I_{2,\text{cusp}}$ (blue line) is strongly suppressed since $\pi\pi$ are in the P-wave and can be safely ignored.

\begin{figure}[htbp]
	\centering
	\includegraphics[width=0.75\textwidth]{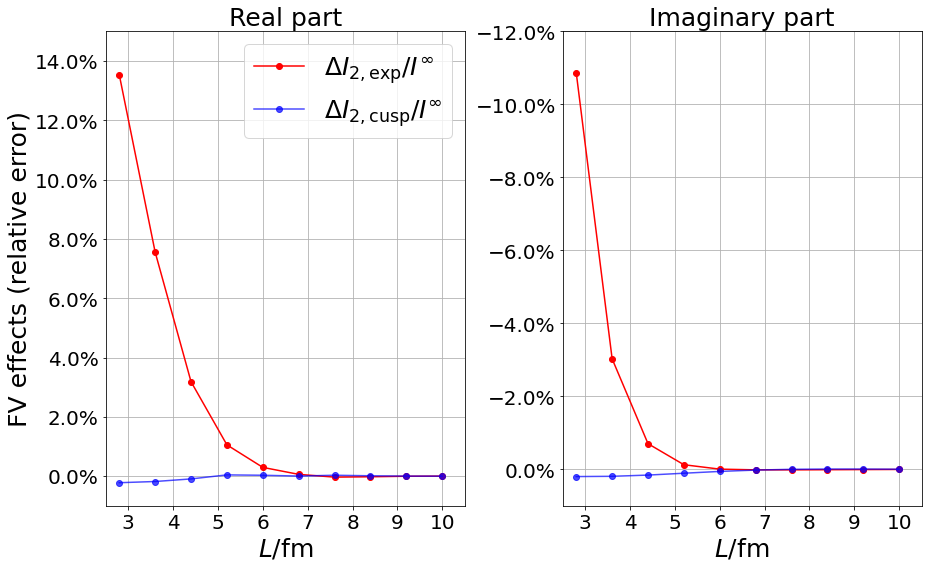}
	\caption{ In the model that neglects the rescattering effect, we numerically calculate the volume suppression behavior of the power-law suppressed part $\Delta I_{2,\text{cusp}}$ and the exponentially suppressed part $\Delta I_{2,\text{exp}}$. In the plot, we show the relative errors. $I^\infty$ is also calculated by this model. The left and right plots show the finite-volume effects of the real and imaginary parts of the decay amplitude, respectively.\label{fig:res2}}
\end{figure}

In other physical processes where a virtual photon generates the $\pi\pi$ intermediate state, such as $K_L \to \mu^+\mu^-$ and $K^+ \to \ell^+\nu_\ell \ell'^+\ell'^-$, the cusp effect exhibits a similar P-wave suppression behavior. In these processes, $\Delta I_2$ is dominated by the exponentially suppressed finite-volume effect, which can be well controlled by increasing the volume in current lattice calculations. For processes where the two particles are in the S-wave, the cusp effect does not have this suppression, and we need to correct it using the method provided here.

\section{Conclusions}
In this paper, we investigate how to apply the EW$_\infty$ method to processes with loop integrals that involve complex electroweak weight functions. We develop a general approach to correct systematic errors such as finite-volume effects, and apply it to cases with single-particle and two-particle intermediate states. We identify two sources of systematic errors:
\begin{itemize}
	\item $\Delta I_1$ involves finite-volume effects and temporal truncation effects of hadronic matrix elements under discrete momentum. If the intermediate states are lighter than the initial state, this temporal trunction effects could appear as unphysical terms with exponentially divergent time dependence. These effects have been investigated in past research works. We summarize the correction formulas for cases with single-particle and two-particle intermediate states. 
	\item $\Delta I_2$ is a new finite-volume effect introduced by the EW$_\infty$ method. Its volume suppression behavior depends on the smoothness of the discrete momentum summand. The advantage of EW$_\infty$ is that the infrared singularity in the electroweak part does not cause singularity of the summand itself, thus greatly suppressing finite-volume effects. The power-law finite-volume effect in $\Delta I_2$ is caused by the cusp effect if the intermediate states are lighter than the initial state. In the single-particle case, we can introduce the IVR method to avoid the cusp effect, ensuring that finite-volume effects are always exponentially suppressed. In the two-particle case, we provide a correction formula for the cusp effect and explain that it is strongly suppressed and can be ignored if two particles are in the P wave.
\end{itemize}

The theoretical method established in this paper has a wide range of application prospects, such as the QED correction of various physical processes, the two-photon exchange contributions of mesons and nucleons, the decay width of radiative decays of mesons, etc. For example, our method provides the correction formula for two-particle finite-volume effects for the two-photon exchange contribution in $K_L \to \mu^+\mu^-$ decay~\cite{Chao:2024vvl}. Since $\eta \to \mu^+\mu^-$ and $K_L \to \mu^+\mu^-$ have similar initial mass and the same two-particle intermediate states, we expect that the cusp effect in $K_L \to \mu^+\mu^-$ decay can also be ignored. Thus, the finite-volume effects of $K_L \to \mu^+\mu^-$ are dominated by power-law suppressed effects in $\Delta I_1$ and the exponentially suppressed effects in $\Delta I_2$. Through the future lattice calculation, the size of this two-photon exchange contribution can be accurately determined, greatly reducing the theoretical uncertainty in the $K_L \to \mu^+\mu^-$ decay.

\acknowledgements
X.Y.T and X.F. would like to thank Peter Boyle, Norman Christ, Luchang Jin, Taku Izubuchi, and our RBC and UKQCD Collaboration colleagues for helpful discussions and support. X.F. has been supported in part by NSFC of China
under Grant No. 12125501, No. 12070131001, and No. 12141501, and National Key Research and Development Program of China under No. 2020YFA0406400. X.Y.T has been supported by US DOE Contract DESC0012704(BNL).

\appendix
\section{\label{sec:App-EM}The definition of quantities in Euclidean space}
Here, we introduce the definition of quantities in Euclidean space used in the main text. First, we define the coordinates $x_E=(\tau,\bm{x}_E)$ and energy-momentum $k_E=(k_E^0,\bm{k}_E)$ of Euclidean space, which relate to the physical quantities $x=(t,\bm{x})$ and $k=(k^0,\bm{k})$ of Minkowski space by Wick rotation
\begin{equation}
	\begin{aligned}
		t\to -i\tau &,\quad k^0\to -i k_E^0. 
	\end{aligned}
\end{equation}
The definition of the spatial components $\bm{x}$ and $\bm{k}$ is consistent in Euclidean space and Minkowski space, so we do not distinguish them in the following text. The momentum of Euclidean space and Minkowski space satisfy the relations
\begin{equation}
	\begin{aligned}
		k^2&=-k_E^2=(k^0)^2-\bm{k}^2=-\left((k_E^0)^2+\bm{k}^2\right),\\
		x^2&=-x_E^2=t^2-\bm{x}^2=-\left(\tau^2+\bm{x}^2\right),\\
		k\cdot x&=-k_E\cdot x_E=k^0 t-\bm{k}\cdot\bm{x}=-\left(k_E^0 \tau+\bm{k}\cdot\bm{x}\right). \\
	\end{aligned}
\end{equation}

Next, we define the operators in Euclidean space and Minkowski space. The coordinate dependence of the Minkowski space operator $J(x)$ and the Euclidean space operator $J_E(x_E)$ is
\begin{equation}
	\begin{aligned}
		J(x)&=e^{i\hat{H}t-i\hat{k} \cdot \bm{x}}J(0)e^{-i\hat{H}t+i\hat{k} \cdot \bm{x}},\\
		J_E(x_E)&=e^{\hat{H}\tau-i\hat{k} \cdot \bm{x}}J_E(0)e^{-\hat{H}\tau+i\hat{k} \cdot \bm{x}},\\
	\end{aligned}
\end{equation}
where $\hat{H}$ and $\hat{k}$ are the Hamiltonian operator and momentum operator of the QCD system, respectively. $J(0)$ and $J_E(0)$ are the current operators in Minkowski space and Euclidean space, and their difference is due to the different conventions of the $\gamma$ matrix definition. In this paper, we choose $\gamma_E^4=\gamma^0$ and $\vec{\gamma}_E=-i\vec{\gamma}$. For the current operators, such as the electromagnetic current $J_{\text{em}}^\mu=\frac 23 \bar{u}\gamma^\mu u-\frac 13 \bar{d}\gamma^\mu d$, the difference between the definitions in Euclidean space and Minkowski space is
\begin{equation}
	J_E^4(0)=J^0(0),\quad J^j_E(0)=-i J^j(0)\quad(j=1,2,3).
\end{equation}

To compensate for this difference, we define
\begin{equation}
	J_{2}(0)\otimes J_{1}(0)=c_{ME}J_{2,E}(0)\otimes J_{1,E}(0)
\end{equation}
where we omit the Lorentz indices and use the symbol $\otimes$ to indicate the structure of multiplying two current operators. For example, when both currents are electromagnetic currents, they satisfy the relation $J_{\text{em}}^\mu\otimes J_{\text{em},\mu}=J_{\text{em},E}^\mu\otimes J_{\text{em},E}^\mu$, so in this case $c_{ME}=1$. $c_{ME}$ differs when the currents take different Lorentz indices, but it is always a constant coefficient.

\section{\label{sec:Append sum}The relation between the smoothness of the summand and the finite-volume effects}
In this section, we prove that the volume suppression behavior of $\Delta I_2$ is determined by the smoothness of the summand. To this end, we first review the Poisson summation formula and then generalize its conclusion.

\subsection{The Poisson summation formula}
We first introduce the Poisson summation formula and the relation between the finite-volume effects and the smoothness of the summand:
\begin{thm}\label{thm:A1}
	For a function $\tilde{f}(\bm{k})$, the difference between discrete summation and continuous integral is given by the Poisson summation formula
	\begin{equation}
		\frac{1}{L^3}\sum_{\bm{k}}\tilde{f}(\bm{k})=\int \frac{d^3 k}{(2\pi)^3} \tilde{f}(\bm{k})+\sum_{\bm{l}\neq\bm{0}}\int \frac{d^3 k}{(2\pi)^3} e^{iL\bm{l}\cdot\bm{k}}\tilde{f}(\bm{k}),
	\end{equation}
	where the range of discrete momentum sum is $\bm{k}=(2\pi/L)\bm{n}$, $\bm{n}=(n_x,n_y,n_z)$ are three-dimensional integers. The difference between summation and integral (which we call the finite-volume effect) suppresses with increasing volume $L$ and this suppression behavior is related to the smoothness of $\tilde{f}(\bm{k})$ as a function of $\bm{k}$:
	\begin{enumerate}
		\item If $\tilde{f}(\bm{k})$ has singularities, then the finite-volume effect suppresses as $O(1/L)$;
		\item If $\tilde{f}(\bm{k})$ is continuously differentiable up to order $N$, then the finite-volume effect suppresses as $O(1/L^{N+1})$;
		\item If $\tilde{f}(\bm{k})$ is infinitely differentiable, then the finite-volume effect suppresses as $O(e^{-\Lambda L})$, where $\Lambda$ is the hadronic mass scale in the problem.
	\end{enumerate}
\end{thm}

We briefly explain the reason for the above relation. Define the reverse Fourier transform of $\tilde{f}(\bm{k})$ as $f(\bm{x})$. Then the behavior of the finite-volume effect is determined by the suppression behavior of $f(\bm{x})$ at large $|\bm{x}|$ since
\begin{equation}
	\sum_{\bm{l}\neq\bm{0}}\int \frac{d^3 k}{(2\pi)^3} e^{iL\bm{l}\cdot\bm{k}}\tilde{f}(\bm{k})=\sum_{\bm{l}\neq\bm{0}}\int d^3x f(L\bm{l}-\bm{x}).
\end{equation}

The suppresion behavior of $f(\bm{x})$ at large $|\bm{x}|$ is determined by the smoothness of $\tilde{f}(\bm{k})$, which is given by the following theorem~\cite{Reed:1975uy,Trefethen}
\begin{thm}\label{thm:decay} For an integrable, continuous function $\tilde{f}(\bm{k})$,
	\begin{enumerate}
		\item If $\tilde{f}(\bm{k})$ has a continuous integrable derivative of order $i$, where $0\leq i\leq N$, and the derivative of order $N+1$ is integrable, then
		\begin{equation}
			f(\bm{x})=O(|\bm{x}|^{-N-1}), \quad \text{as}\quad |\bm{x}|\to\infty.
		\end{equation}
		
		\item If $\tilde{f}(\bm{k})$ can be extended to a function $\tilde{f}(\bm{\zeta})$ of complex variable $\bm{\zeta}=\bm{k}+i\bm{\eta}$, and $\tilde{f}(\bm{\zeta})$ is an analytic function in the region $| \bm{\eta}|<a~(a>0)$ on the complex plane, and for any given $|\bm{\eta}|<a$, $\tilde{f}(\bm{\zeta})$ is integrable, then
		\begin{equation}
			f(\bm{x})=O(e^{-a|\bm{x}|}), \quad \text{as}\quad |\bm{x}|\to\infty.
		\end{equation}
	\end{enumerate}
\end{thm}
The two parts of Theorem \ref{thm:decay} correspond to power-law and exponential suppressions of finite-volume effects in Theorem \ref{thm:A1}, respectively. In physical problems, the dependence of $\tilde{f}(\bm{k})$ on $\bm{k}$ often appears as energy $E(\bm{k})=\sqrt{\bm{k}^2+m^2}$, whose analytical range is $|\mathrm{Im}\bm{k}|<a=m$. Therefore, the exponential suppression rate is given by the hadronic mass scale in the problem.

\subsection{Generalization to the new summation form}
We extend the conclusion of the previous section to the summation form used in this paper:
\begin{thm}
	Let $\tilde{l}(k)$ and $\tilde{h}(k)$ be functions of four-momentum $k=(k^0,\bm{k})$, and define their loop integral forms
	\begin{equation}
		\begin{aligned}
			I^{\infty}&=\int \frac{d^3k}{(2\pi)^3}\int \frac{dk^0}{2\pi}\tilde{l}(k)\tilde{h}(k),\\
			I^{(L)}&=\frac{1}{L^3}\sum_{\bm{k}'\in\Gamma}\int \frac{d^3k}{(2\pi)^3}\delta_L(\bm{k}'-\bm{k})\int \frac{dk^0}{2\pi}\tilde{l}(k)\tilde{h}(k')=\frac{1}{L^3}\sum_{\bm{k}'\in\Gamma}\hat{I}(\bm{k}'),
		\end{aligned}
	\end{equation}
	with $k=(k^0,\bm{k})$, $k'=(k^0,\bm{k}')$. $\bm{k}'\in\Gamma$ is the discrete momentum in finite volume. The finite-volume effect $I^{(L)}-I^{\infty}$ is related to the smoothness of $\hat{I}(\bm{k}')$ by (here, $\bm{k}'$ is extended to be a continuous real variable)
	\begin{enumerate}
		\item If $\hat{I}(\bm{k}')$ has singularities, then the finite-volume effect suppresses as $O(1/L)$;
		\item If $\hat{I}(\bm{k}')$ is continuously differentiable up to order $N$, then the finite-volume effect suppresses as $O(1/L^{N+1})$;
		\item If $\hat{I}(\bm{k}')$ is infinitely differentiable, then the finite-volume effect suppresses as $O(e^{-\Lambda L})$, where $\Lambda$ is the hadronic mass scale in the problem.
	\end{enumerate}
\end{thm}

\begin{proof}[Proof]
	We split the finite-volume correction $I^{(L)}-I^{\infty}$ into two parts as
	\begin{equation}
		\label{Iinout}
		\begin{aligned}
			I^{(L)}-I^{\infty}=\left(\frac{1}{L^3}\sum_{\bm{k}'\in\Gamma}-\int\frac{d^3 k'}{(2\pi)^3}\right)\hat{I}(\bm{k}')+\left(\int\frac{d^3 k'}{(2\pi)^3}\hat{I}(\bm{k}')-I^{\infty}\right).
		\end{aligned}
	\end{equation}
	
	For the first part, it has the same form as the Poisson summation formula. According to the conclusion of the previous section, the volume suppression behavior of this part depends directly on the smoothness of $\hat{I}(\bm{k}')$.
	
	For the second part, since there is a $\delta_L(\bm{k}'-\bm{k})$ function in $\hat{I}(\bm{k}')$, the integration of $\hat{I}(\bm{k}')$ is not $I^\infty$. Their difference can be written in coordinate space as
	\begin{equation}
		\begin{aligned}
			&\int \frac{d^3 k'}{(2\pi)^3}\hat{I}(\bm{k}')-I^{\infty}=\int_{>V} d^3x g(\bm{x})\\
			&g(\bm{x})=\int\frac{d^3 k'}{(2\pi)^3}\int\frac{d^3 k}{(2\pi)^3}e^{i(\bm{k}'-\bm{k})\cdot\bm{x}}\int \frac{dk^0}{2\pi}\tilde{l}(k)\tilde{h}(k')
		\end{aligned}
	\end{equation}
	where $\int_{>V}d^3x$ denotes the integral outside the finite volume. Thus, the suppression behavior of this part depends on the suppression behavior of $g(\bm{x})$ at large $|\bm{x}|$. According to Theorem \ref{thm:decay}, the suppression behavior of $g(\bm{x})$ is determined by the smoothness of its Fourier transform
	\begin{equation}
		\tilde{g}(\bm{q})=\int d^3x e^{i\bm{q}\cdot\bm{x}} g(\bm{x}) =\int \frac{d^3k}{(2\pi)^3}\int\frac{dk^0}{2\pi}\tilde{l}(k)\tilde{h}(k')|_{\bm{k}'=\bm{k}-\bm{q}}
	\end{equation}
	
	Next, we compare the integral structure of $\tilde{g}(\bm{q})$ with that of
	\begin{equation}
		\hat{I}(\bm{k}')=\int \frac{d^3k}{(2\pi)^3}\delta_L(\bm{k}'-\bm{k})\int \frac{dk^0}{2\pi}\tilde{l}(k)\tilde{h}(k')
	\end{equation}
	For a given volume, $\delta_L(\bm{k}'-\bm{k})$ is a smooth function. The difference between $\hat{I}(\bm{k}')$ and $\tilde{g}(\bm{q})$ in integral structure is just a variable substitution $\bm{q}=\bm{k}-\bm{k}'$, so they have exactly the same analytic properties. Therefore, the volume suppression behavior of this part is also determined by the smoothness of $\hat{I}(\bm{k}')$ (or equivalently, by the smoothness of $\tilde {g} (\bm{q}) $). This completes the proof of this theorem.
	\end {proof}
	
\section{Derivation of $\Delta I_{2,\text{cusp}}$ with rescattering effects\label{sec:Append cusp corr}}
We first introduce the Bethe–Salpeter kernel expansion of the hadronic matrix element. We define the two-particle propagators as
\begin{equation}
	\begin{aligned}
		G^{i\epsilon}_2(q)&=\frac{1}{2} S(q)S(p-k'-q),\\
		S(q)&=\frac{iZ_\pi(q^2)}{q^2-m_\pi^2+i\epsilon},
	\end{aligned}
\end{equation}
where the energy-momentum of the two particles are $q$ and $p-k'-q$, respectively. $S(q)$ is the fully dressed propagator of the $\pi$, and $\frac 12$ is the symmetry factor when the two particles are identical. For the loop integral of multiplying any function $L(q),R(q)$ with $G^{i\epsilon}_2(q)$, we define a shorthand symbol
\begin{equation}
	\label{Gloop}
	\begin{aligned}
		L G^{i\epsilon}_2 R&= \int \frac{d^4q}{(2\pi)^4}L(q)G^{i\epsilon}_2(q)R(q)\\
		&=\int \frac{d^3q}{(2\pi)^3} L(q)\bar{G}^{i\epsilon}_2(\bm{q})R(q)|_{q^0=E_{\pi}(\bm{q})}+\int \frac{d^4q}{(2\pi)^4}L(q)\hat{G}^{i\epsilon}_2(q)R(q),
	\end{aligned}
\end{equation}
where, under the condition of $\sqrt{s}<4m_\pi$, $\hat{G}^{i\epsilon}_2(q)$ is an analytic function without singularities~\cite{Luscher:1986pf}. $\bar{G}^{i\epsilon}_2(\bm{q})$ extracts the main two-particle singularity as
\begin{equation}
	\label{Gbar}
	\bar{G}^{i\epsilon}_2(\bm{q})=\frac i2\frac{h(s)}{2E_{\pi}(\bm{q})2E_{\pi}(\bm{k}'+\bm{q})\left(m-k^0-E_{\pi}(\bm{q})-E_{\pi}(\bm{k}'+\bm{q})+i\epsilon\right)},
\end{equation}
Here, $h(s)$ is defined in Eq.~(\ref{hs}). Similar form as $\bar{G}^{i\epsilon}_2(\bm{q})$ has also been used in study the one $\pi\pi$ loop case in Eq.~(\ref{H0}). 

The infinite-volume matrix element $H^{\infty}$ can be expanded as
\begin{equation}
	\begin{aligned}
		H^{\infty}=\bar{\Gamma}^{\text{in}} G_2^{i\epsilon}\bar{\Gamma}^{\text{out}}+\bar{\Gamma}^{\text{in}} G_2^{i\epsilon}iB_2 G_2^{i\epsilon}\bar{\Gamma}^{\text{out}}+\cdots
	\end{aligned},
\end{equation}
where $\bar{\Gamma}^{\mathrm{in}}$ and $\bar{\Gamma}^{\mathrm{out}}$ are the two-particle irreducible (2PI) vertices. $B_2$ is the two-particle irreducible Bethe–Salpeter kernel. 

The $\widetilde{\mathrm{PV}}$ prescription is equivalent to defining a new two-particle propagator~\cite{Hansen:2014eka}
\begin{equation}\label{PVtilde}
	G_2^{\widetilde{PV}}(q)=G_2^{i\epsilon}(q)-i\widetilde{\rho}(s)
\end{equation}
with
\begin{equation}\label{rhotilde}
	\widetilde{\rho}(s)=h(s)\begin{cases}
		\frac{-i\sqrt{\frac{s^2}{4}-m_\pi^2}}{16\pi\sqrt{s}}\quad &s>(2m_\pi)^2;\\
		\frac{\sqrt{m_\pi^2-\frac{s^2}{4}}}{16\pi\sqrt{s}}\quad&0< s\leq (2m_\pi)^2.\\
	\end{cases}
\end{equation}

We can use the $\widetilde{\mathrm{PV}}$ prescription to define an infinitely differentiable hadronic matrix element as
\begin{equation}\label{Htilde}
	\widetilde{H}^{\infty}=\bar{\Gamma}^{\text{in}} G_2^{\widetilde{PV}}\bar{\Gamma}^{\text{out}}+\bar{\Gamma}^{\text{in}} G_2^{\widetilde{PV}}iB_2 G_2^{\widetilde{PV}}\bar{\Gamma}^{\text{out}}+\cdots.
\end{equation}
Then we can isolate the non-smooth part in $H^{\infty}$ as
\begin{equation}
	\begin{aligned}
		H^{\text{cusp}}&=\widetilde{H}^{\infty}-H^{\infty}\\&=\Gamma^{\text{in}}(-i\widetilde{\rho})\Gamma^{\text{out}}+\Gamma^{\text{in}}(-i\widetilde{\rho})i\mathcal{M}_2(-i\widetilde{\rho})\Gamma^{\text{out}}+\cdots\\
		&=\Gamma^{\text{in}}\frac{-i\widetilde{\rho}}{1-\widetilde{\rho}\mathcal{M}}\Gamma^{\text{out}},
	\end{aligned}
\end{equation}

Using partial wave expansixons in Eq.~(\ref{Gammatilde}) and Eq.~(\ref{M2}), we can simplify $H^{\text{cusp}}(k',p)$ as 
\begin{equation}\label{Hcusp}
	\begin{aligned}
		H^{\text{cusp}}(k',p)=&\sum_{l,l',m,m'}i\frac{h(s)}{16\pi\sqrt{s}}A^\infty_{\pi\pi,lm;l'm'}(-\bm{k}',E_{\pi\pi})|_{E_{\pi\pi}=m_\eta-k^0}\\ &\times \begin {cases}q^*
		(\tan(\delta_l)-i)\quad &s> (2m_\pi)^2\\
		|q^*|(\tanh(\sigma_l)+1)\quad &0< s \leq (2m_\pi)^2\\
		\end {cases}
		\end {aligned}
\end {equation}
In the region of $0< s\leq (2m_\pi)^2$, $q^*=i|q^*|=i\sqrt{m_\pi^2-\frac{s}{4}}$ is an imaginary variable, and we can extend the phase shift $\delta_l$ and the matrix element $A^\infty_{\pi\pi,lm;l'm'}$ from functions of the real variable $q^*$ to the functions of the imaginary variable $q^*=i|q^*|$. The phase shift $\delta_l$ becomes an imaginary phase shift $\delta_l(s)=i\sigma_l(s)$.

\section{Partial wave expansion of matrix elements in the P wave\label{sec:Append cusp}}
Here, we show that $l=1$ partial wave expansion of general matrix elements $\langle f(p)|J_1|\pi\pi(q,q')\rangle$ or $\langle \pi\pi(q,q') |J_2|i(p)\rangle$ is proportional to $(q^*)^l=q^*$. The cases with higher angular momentum quantum number can also be proved similarly. The P wave expansion of $\langle f |J^\mu_1|\pi\pi(q,q')\rangle$ is given by
\begin{equation}
	\Gamma^{\mu}_{1m}=\frac{1}{\sqrt{4\pi}}\int d\Omega_{\hat{q}^*}Y_{1m}(\hat{q}^*)\langle f(p)|J^\mu_1|\pi\pi(q,q')\rangle.
\end{equation}

The general Lorentz structure of $\langle f|J^\mu_1|\pi\pi(q,q')\rangle$ can be written as
\begin{equation}
	\begin{aligned}
		\langle f(p)|J^\mu_1|\pi\pi(q,q')\rangle&=\epsilon^{\mu\rho\sigma\lambda}p_\rho (q+q')_\sigma (q-q')_\lambda F_1(x,y,z)\\&+(q+q')^\mu F_2(x,y,z)\\&+(p-q-q')^\mu F_3(x,y,z)\\&+(q-q')^\mu F_4(x,y,z),
	\end{aligned}
\end{equation}
where $F_i(x,y,z)$ are form factors that depend on three independent Lorentz invariants $x=(q+q')^2$, $y=(p-q-q')^2$, and $z=p\cdot(q-q')$. $x$ and $y$ are related to $s$ and $s_\gamma$ in Eq.~(\ref{FB}). $z$ is proportional to $\cos\theta_\pi$, where $\theta_\pi$ is the angle between the single particle momentum $\bm{q}^*$ in the center-of-mass and the total two-particle momentum $\bm{q}+\bm{q}'$.

The nonzero P wave projection requires that $F_{i}(x,y,z)|_{\text{P wave}}$ depend on $z$ as
\begin{equation}
	\begin{aligned}
		F_{1,4}(x,y,z)|_{\text{P wave}}&=f_{1,4}(x,y)+z^2 f'_{1,4}(x,y)+\cdots\\
		F_{2,3}(x,y,z)|_{\text{P wave}}&=zf_{2,3}(x,y)+z^3 f'_{2,3}(x,y)+\cdots
	\end{aligned}
\end{equation}
Hence, the Lorentz structure in the P wave are
\begin{equation}
	\begin{aligned}
		\langle f|J^\mu_1|\pi\pi(q,q')\rangle|_{\text{P wave}}&=\epsilon^{\mu\rho\sigma\lambda}p_\rho (q+q')_\sigma (q-q')_\lambda f_1(x,y)\\&+p\cdot(q-q') (q+q')^\mu f_2(x,y)\\&+p\cdot(q-q') (p-q-q')^\mu f_3(x,y)\\&+(q-q')^\mu f_4(x,y)+\cdots
	\end{aligned}
\end{equation}
We see that the P-wave matrix element is always proportional to $(q-q')$, which leads to the P-wave expansion being proportional to the center-of-mass momentum $q^*$ since
\begin{equation}
	\frac{1}{\sqrt{4\pi}}\int d\Omega_{\hat{q}^*}Y_{1m}(\hat{q}^*) (q-q')\propto q^*
\end{equation}

For higher angular momentum $l$, it can be proved similarly that the partial wave expansion of $\langle f(p)|J_1|\pi\pi(q,q')\rangle$ is proportional to $(q^*)^l$.

\section{Numerical test of $\Delta I_2$ in $\eta\to\mu^+\mu^-$ decay\label{sec:Append eta_test}}
\subsection{Numerical implementation of $F^{(\pi)}(s)$ and $B_\eta(s,s_\gamma)$}
To test the behavior of $\Delta I_2$ in $\eta\to\mu^+\mu^-$ decay numerically, we need the infinite-volume matrix elements $\langle 0|J^\nu_{em}|\pi\pi(q,q')\rangle$ and $\langle \pi\pi(q,q')|J^\mu_{em}|\eta(p)\rangle$ in Eq.~(\ref{FB}) as inputs. In phenomenological studies, the form factors $F^{(\pi)}(s)$ and $B_\eta(s,s_\gamma)$ can be well described by the vector meson dominance model as~\cite{Venugopal:1998fq}
\begin{equation}
	\begin{aligned}
		F^{(\pi)}(s)&=\frac{-m_\rho^2}{s-m_\rho^2},\\
		B_\eta(s,s_\gamma)&=B_\eta(0,0)(1+\frac{s}{2m_\rho^2})\frac{-m_\rho^2}{s-m_\rho^2}\frac{-m_\rho^2}{s_\gamma-m_\rho^2},
	\end{aligned}
\end{equation}
where $B_\eta(0,0)$ is determined by the Chiral anomaly. For large $s, s_\gamma\sim m_\rho$, the above formula is modified by the GS model as~\cite{Gounaris:1968mw} 
\begin{equation}
	\frac{-m_\rho^2}{s-m_\rho^2}\to\frac{-m_\rho^2-\Pi_\rho(s)}{s-m_\rho^2-\Pi_\rho(s)},
\end{equation}
where $\Pi_\rho(s)$ encodes the $\rho-\pi\pi$ interaction. Its imaginary part is related to the $\rho$ width by $\mathrm{Im}[\Pi_\rho(s)]=-m_\rho \Gamma_\rho$, and its real part can be obtained by a dispersion relation.

These form factors can also be obtained by lattice QCD calculations~\cite{Feng:2014gba,Andersen:2018mau,Erben:2019nmx,Meyer:2011um}. To test the qualitative behavior of $\Delta I_2$ in this paper, we simply use the form of GS model.

\subsection{Numerical implementation of $\Delta I_{2,\text{cusp}}$}
Eq.~(\ref{Hcusp_eg}) and Eq.~(\ref{Hcusp_eg_LO}) show the non-smooth part of the matrix element in the $\eta\to\mu^+\mu^-$ decay with and without rescattering effects, respectively. In the numerical test, we neglect the rescattering effects. By substituting $H^{\text{cusp},(0)}(k',p)$ and $L^\infty(k)$ into Eq.~(\ref{Icusp}), and performing the contour integration in the $k^0$ direction, we obtain the correction formula of the cusp effect $\Delta I_{2,\text{cusp}}=\tilde{I}_{\text{cusp}}^{(L)}-\tilde{I}_{\text{cusp}}^{\infty}$ with the contribution of the non-smooth part in finite volume ($\tilde{I}_{\text{cusp}}^{(L)}$) and infinite volume ($\tilde{I}_{\text{cusp}}^{\infty}$). $\tilde{I}_{\text{cusp}}^{(L)}$ is defined as
\begin{equation}
	\begin{aligned}
		&\tilde{I}_{\text{cusp}}^{(L)}=\frac{1}{L^3}\sum_{\bm{k}'\in\Gamma} \int \frac{d^3 k}{(2\pi)^3}\delta_L(\bm{k}'-\bm{k})\int \frac{d k^0}{2\pi} \theta(m-k^0-2m_\pi)L_{\mu\nu}^\infty(k) H^{\text{cusp},(0),\mu\nu}(k',p)
		\\=&\frac{1}{L^3}\sum_{\bm{k}'\in\Gamma}\left(\frac{1}{12\pi^2m_\eta^2}\right)\int_{0}^{m_\eta-2m_\pi} d|\bm{k}| \frac{i(q^*)^3}{\sqrt{s}}\frac{F^{(\pi)}(s)B_\eta(s,s_\gamma)h(s)g(|\bm{k}|,\bm{k}')}{m-2|\bm{k}|}\theta(s-(2m_\pi)^2)\\
		+&\frac{1}{L^3}\sum_{\bm{k}'\in\Gamma}\left(\frac{1}{12\pi^2m_\eta^2}\right)\int_{0}^{m_\eta-2m_\pi}  d|\bm{k}| \frac{|q^*|^3}{\sqrt{s}}\frac{F^{(\pi)}(s)B_\eta(s,s_\gamma)h(s)g(|\bm{k}|,\bm{k}')}{m-2|\bm{k}|}\theta(s)\theta((2m_\pi)^2-s)\\
	\end{aligned}
\end{equation}
where the first and the second term correspond to the contributions of the non-smooth part above and below the threshold $s=(2m_\pi)^2$, respectively. The contour integration picks up the singularity of the electroweak part at $k^0=|\bm{k}|$ (the residues at other singularities do not contribute to the cusp effect), so here $s=(m_\eta-|\bm{k}|)^2-|\bm{k}'|^2$. The angular integration of $\bm{k}$ is defined in
\begin{equation}\label{gkk'}
	g(|\bm{k}|,\bm{k}')=\int d\Omega_{\hat{k}}\delta_L(\bm{k}'-\bm{k})\frac{\bm{k}\cdot\bm{k}'}{1-\beta\cos\theta},
\end{equation}
where $\theta$ is defined in Eq.~(\ref{LEcalc}) as the angle between the photon momentum $\bm{k}$ and the muon momentum $\bm{p}^-$. 

Taking the infinite volume limit of $L\to\infty$ of the above equation, $\tilde{I}_{\text{cusp}}^{\infty}$ is given by
\begin{equation}
	\begin{aligned}
		\tilde{I}_{\text{cusp}}^{\infty}=&\int \frac{d^3 k}{(2\pi)^3}\int \frac{d k^0}{2\pi}\theta(m-k^0-2m_\pi)L_{\mu\nu}^\infty(k) H^{\text{cusp},(0),\mu\nu}(k,p)
		\\=&\left(\frac{1}{6\pi m_\eta^2\beta}\right)\ln\left(\frac{1+\beta}{1-\beta}\right)\int_{0}^{\frac{m_\eta^2-4m_\pi^2}{2m_\eta}} d|\bm{k}| \frac{i(q^*)^3}{\sqrt{s}}\frac{F^{(\pi)}(s)B_\eta(s,s_\gamma)h(s)|\bm{k}|^2}{m-2|\bm{k}|}\\
		+&\left(\frac{1}{6\pi m_\eta^2\beta}\right)\ln\left(\frac{1+\beta}{1-\beta}\right)\int_{\frac{m_\eta^2-4m_\pi^2}{2m_\eta}}^{m_\eta-2m_\pi} d|\bm{k}| \frac{|q^*|^3}{\sqrt{s}}\frac{F^{(\pi)}(s)B_\eta(s,s_\gamma)h(s)|\bm{k}|^2}{m-2|\bm{k}|}
	\end{aligned}
\end{equation}

We can calculate $\Delta I_{2,\text{cusp}}$ using the above two formulas. When considering rescattering effects, we need to replace $i(q^*)^3 \to (i-\tan \delta_l)(q^*)^3$ in the first term and $|q^*|^3 \to (1+\tanh \sigma_l)|q^*|^3$ in the second term. We take $h(s)=1-\exp(- \exp(-(2m_\pi)^2/s)/(1-s/(2m_\pi)^2))$ to make the summand $\hat {I}(\bm {k}')$ as smooth as possible near the threshold after removing the non-smooth part.

\subsection{Numerical implementation of $\Delta I_{2,\text{exp}}$}
In this part we exlain how to calculate $\Delta I_{2,\text{exp}}$ under the assumption of ignoring rescattering effects. Our idea is to numerically implement
\begin{equation}\label{It1loop}
	\tilde{I}^{(L)}=\frac{1}{L^3}\sum_{\bm{k}'\in\Gamma}\int \frac{d^3 k}{(2\pi)^3}\delta_L(\bm{k}'-\bm{k})\int_{C} \frac{d k^0}{2\pi} L_{\mu\nu}^\infty(k)H^{\infty,(0),\mu\nu}(k',p)
\end{equation}
where the one $\pi\pi$ loop contribution $H^{\infty,(0),\mu\nu}(k',p)$ is given by Eq.~(\ref{H0}). Then we obtain $\Delta I_{2}$ by
\begin{equation}\label{Iexp1}
	\Delta I_{2}=\tilde{I}^{(L)}-I^{\infty}=\tilde{I}^{(L)}-\lim_{L\to\infty}\tilde{I}^{(L)}
\end{equation}
Removing the power-law suppression part $\Delta I_{2,\text{cusp}}$, we can get the exponential suppression part
\begin{equation}\label{Iexp2}
	\Delta I_{2,\text{exp}}=\Delta I_{2}-\Delta I_{2,\text{cusp}}
\end{equation}

Next, we describe the numerical implementation of Eq.~(\ref{It1loop}). Since this equation involves a high-dimensional integral at each momenta, it is difficult to calculate it directly in momentum space. Therefore, we transform it back to the coordinate space for calculation. We introduce a temporal cutoff $t_s$ to decompose Eq.~(\ref{It1loop}) into two parts as
\begin{equation}\label{Isplit}
	\begin{aligned}
		\tilde{I}^{(L)}&=\frac{1}{L^3}\sum_{\bm{k}'\in\Gamma}\int \frac{d^3 k}{(2\pi)^3}\delta_L(\bm{k}'-\bm{k})\int_{C} \frac{d k^0}{2\pi} L_{\mu\nu}^\infty(k)H^{(T),(0),\mu\nu}(k',p)\\&+\frac{1}{L^3}\sum_{\bm{k}'\in\Gamma}\int \frac{d^3 k}{(2\pi)^3}\delta_L(\bm{k}'-\bm{k})\int_{C} \frac{d k^0}{2\pi} L_{\mu\nu}^\infty(k)(H^{\infty,(0),\mu\nu}(k',p)-H^{(T),\mu\nu}(k',p)),
	\end{aligned}
\end{equation}
where we define
\begin{equation}\label{HT1loop}
	H^{(T),(0),\mu\nu}(k',p)=\frac{i}{2}\int\frac{d^3 q}{(2\pi)^3}\frac{\langle 0|J^\nu_{em}|\pi\pi(q,q')\rangle\langle \pi\pi(q,q')|J^\mu_{em}|\eta(p)\rangle}{2E_\pi(\bm{q})2E_\pi(\bm{k}'+\bm{q})(m_\eta-k^0-E_{\pi\pi})}(1-e^{(m_\eta-k^0-E_{\pi\pi})T/2})
\end{equation}
The first line In Eq.~(\ref{Isplit}) can be calculated in the coordinate space as
\begin{equation}\label{IX}
	\begin{aligned}
		&\frac{1}{L^3}\sum_{\bm{k}'\in\Gamma}\int \frac{d^3 k}{(2\pi)^3}\delta_L(\bm{k}'-\bm{k})\int_{C} \frac{d k^0}{2\pi} L_{\mu\nu}^\infty(k)H^{(T),(0),\mu\nu}(k',p)\\
		=&\int d^3x \int^0_{-T/2} d\tau L_E^{\infty,\mu\nu}(\tau,\bm{x}) c_{ME}H_{E}^{(L),(0),\mu\nu}(\tau,\bm{x})
	\end{aligned}
\end{equation}
where $L_E^{\infty,\mu\nu}(\tau,\bm{x})$ is given by Eq.~(\ref{LEetall}), and $H_{E}^{(L),(0)}(\tau,\bm{x})$ is defined in the coordinate space as
\begin{equation}
	\begin{aligned}
		c_{ME}H_{E}^{(L),(0),\mu\nu}(\bm{x},\tau)&=\frac{1}{L^3}\sum_{\bm{k}'\in\Gamma} \tilde{H}^{(0),\mu\nu}(-\bm{k}',\tau)e^{i\bm{k}'\cdot\bm{x}}\\
		\tilde{H}^{(0),\mu\nu}(-\bm{k}',\tau)&=-
		\frac{\epsilon^{\mu\nu\rho\sigma}p_\rho k'_{\sigma}}{24\pi^2}\int_{\sqrt{4m_\pi^2+\bm{k}'^2}}^\infty dE_{\pi\pi}\frac{(a^*)^3}{\sqrt{s}}F^{(\pi)}(s)B_\eta(s,s_\gamma)e^{-(m_\eta-E_{\pi\pi})\tau}
	\end{aligned}
\end{equation}

The second line in Eq.(\ref{Isplit}) is a correction for the temporal truncation effect, where we only need to consider the exponentially divergent terms. We can calculate it in momentum space as
\begin{equation}\label{DIT}
	\begin{aligned}
		&\frac{1}{L^3}\sum_{\bm{k}'\in\Gamma}\int \frac{d^3 k}{(2\pi)^3}\delta_L(\bm{k}'-\bm{k})\int_{C} \frac{d k^0}{2\pi} L_{\mu\nu}^\infty(k)(H^{\infty,(0),\mu\nu}(k',p)-H^{(T),(0),\mu\nu}(k',p))\\
		=&\frac{1}{L^3}\sum_{\bm{k}'\in\Gamma}\frac{1}{12\pi^3 m_\eta^2}\int_0^\infty dk\int_{\sqrt{4m_\pi^2+\bm{k}'^2}}^\infty dE_{\pi\pi} \frac{(a^*)^3}{\sqrt{s}}\frac{F^{(\pi)}(s)B_\eta(s,s_\gamma)g(|\bm{k}|,\bm{k}')}{(m-2|\bm{k}|+i\epsilon)(m_\eta-|\bm{k}|-E_{\pi\pi}+i\epsilon)}e^{(m-|\bm{k}|-E_{\pi\pi})T/2},
	\end{aligned}
\end{equation}
where $g(|\bm{k}|,\bm{k}')$ is defined in Eq.~(\ref{gkk'}).

By combining Eq.~(\ref{IX}) and Eq.~(\ref{DIT}), we obtain $\tilde{I}^{(L)}$ for a given volume $L$. One can verify that, after summing over these two terms, the final result $\tilde{I}^{(L)}$ is independent of the sufficient large $T$. Then we can calculate $\Delta I_{2,\text{exp}}$ using Eq.~(\ref{Iexp1}) and Eq.~(\ref{Iexp2}).

\bibliography{ref}
	
\end{document}